\documentclass[pdflatex,sn-mathphys-num]{sn-jnl}

\usepackage{graphicx}
\usepackage{amsmath,amssymb,amsfonts}
\usepackage{algorithm}
\usepackage{algorithmic}
\usepackage{tikz}
\usepackage{chemfig}
\usepackage{longtable}
\usepackage{multirow}
\usepackage{pdflscape}
\usepackage{listings}
\usepackage{booktabs}
\usepackage{tabularx}
\usepackage{array}
\usepackage[normalem]{ulem}
\setchemfig{atom sep=2em}

\usetikzlibrary{arrows.meta,positioning,fit,calc,shapes.geometric}

\usepackage[version=4]{mhchem}
\input{figure/syn-style.sty}

\hypersetup{
  pdftitle={On the realizability of abstract reaction networks with real
    molecules and reactions},
  pdfauthor={Tuyet-Minh Phan, Daniel Merkle, Tieu-Long Phan,
    Christoph Flamm, and Peter F. Stadler},
  pdfkeywords={abstract reaction networks, balanced bijection,
    chemical graph transformation, constrained realizability,
    backtracking, NP-completeness},
  bookmarksdepth=3
}
\newtheorem{theorem}{Theorem}
\newtheorem{problem}[theorem]{Problem}
\newtheorem{definition}[theorem]{Definition}
\newtheorem{corollary}[theorem]{Corollary} 
\newtheorem{proposition}[theorem]{Proposition}
\newtheorem{lemma}[theorem]{Lemma}
\newtheorem{obs}[theorem]{Observation}
\newcommand{\coloneqq}{\mathrel{:=}}
\newcommand{\king}{\mathsf{K}}
\DeclareMathOperator{\supp}{supp}
\DeclareMathOperator{\Aut}{Aut}
\DeclareMathOperator{\orb}{orb}

\lstdefinestyle{shell}{
  language=bash,
  basicstyle=\small\ttfamily,
  columns=fullflexible,
  keepspaces=true,
  breaklines=true,
  showstringspaces=false,
  frame=single,
  framerule=0.4pt
}
\providecommand{\ContinuedFloat}{\addtocounter{figure}{-1}}

\begin{document}

\title[Realizability of abstract reaction networks]{On the realizability of
abstract reaction networks with real molecules and reactions}

\author*[1,2]{\fnm{Tuyet-Minh} \sur{Phan}}
\email{minh@bioinf.uni-leipzig.de}

\author[3,4]{\fnm{Daniel} \sur{Merkle}}
\email{daniel.merkle@uni-bielefeld.de}

\author[1,3]{\fnm{Tieu-Long} \sur{Phan}}
\email{tieu@bioinf.uni-leipzig.de}

\author[2]{\fnm{Christoph} \sur{Flamm}}
\email{xtof@tbi.univie.ac.at}

\author*[1,2,5,6,7,8]{\fnm{Peter F.} \sur{Stadler}}
\email{studla@bioinf.uni-leipzig.de}

\affil[1]{\orgdiv{Bioinformatics Group, Department of Computer Science,
Interdisciplinary Center for Bioinformatics, and School for Embedded and
Composite Artificial Intelligence (SECAI)}, \orgname{Leipzig University},
\orgaddress{\street{H{\"a}rtelstra{\ss}e 16--18}, \city{Leipzig},
\postcode{D-04107}, \country{Germany}}}

\affil[2]{\orgdiv{Department of Theoretical Chemistry},
\orgname{University of Vienna},
\orgaddress{\street{W{\"a}hringerstra{\ss}e 17}, \city{Vienna},
\postcode{A-1090}, \country{Austria}}}

\affil[3]{\orgdiv{Department of Mathematics and Computer Science},
\orgname{University of Southern Denmark},
\orgaddress{\city{Odense M}, \postcode{DK-5230}, \country{Denmark}}}

\affil[4]{\orgdiv{Algorithmic Cheminformatics Group, Faculty of Technology
and Center for Biotechnology (CeBiTec)}, \orgname{Bielefeld University},
\orgaddress{\street{Universit{\"a}tsstra{\ss}e 25}, \city{Bielefeld},
\postcode{D-33501}, \country{Germany}}}

\affil[5]{\orgname{Max Planck Institute for Mathematics in the Sciences},
\orgaddress{\street{Inselstra{\ss}e 22}, \city{Leipzig},
\postcode{D-04103}, \country{Germany}}}

\affil[6]{\orgdiv{Facultad de Ciencias},
\orgname{Universidad Nacional de Colombia},
\orgaddress{\city{Bogot{\'a}}, \country{Colombia}}}

\affil[7]{\orgdiv{Center for non-coding RNA in Technology and Health},
\orgname{University of Copenhagen},
\orgaddress{\street{Ridebanevej 9}, \city{Frederiksberg},
\postcode{DK-1870}, \country{Denmark}}}

\affil[8]{\orgname{Santa Fe Institute},
\orgaddress{\street{1399 Hyde Park Rd.}, \city{Santa Fe},
\state{NM}, \postcode{87501}, \country{United States}}}

\abstract{Abstract reaction networks appear not only as models of chemical
  reactions but also as models of complex systems, with applications in
  areas such as ecology and epidemiology, and as one of several alternative
  paradigms for non-standard computation. It is therefore of interest to
  determine whether an abstract reaction network can be realized by a
  concrete set of molecules and plausible chemical reaction mechanisms. It
  is known that a reaction network has a realization in terms of chemical
  graphs (i.e., Lewis structures) if and only if it is conservative. Here
  we consider the problem of assigning a set $M$ of known molecules to a
  set $X$ of abstract entities in a reaction network $(X,R)$ such that each
  reaction satisfies mass balance and adheres to one of an allowed set of
  chemical reaction mechanisms. We show that this problem is
  NP-complete. Nevertheless, it can be solved in practice using a
  backtrace-and-prune algorithm inspired by the VF2 family of algorithms
  originally designed for the subgraph isomorphism problem.}


\maketitle
\raggedbottom

\section{Introduction}

Chemical reaction networks present an interesting model of ``non-standard''
computation. \emph{Stochastic Chemical Reaction Networks} (sCRNs) consider
a set of chemical reactions acting on a finite number of molecules in a
well-stirred solution according to standard chemical kinetics equations.
Their computational power and their relationship to other models of
computation is surveyed in detail in \cite{Cook:09}. They are ``almost''
Turing universal, in the sense that they can approximate any computable
function with probability of error less than $\varepsilon$ for any
$\varepsilon>0$ \cite{Soloveichik:08}. Moreover, there are several mild
extensions that make sCRNs Turing universal: for instance, it suffices to
define a special class of ``fast'' reactions that occur before all other
``slow'' reactions \cite{Cook:09}.

In this setting, a program is an abstract reaction network over a set $X$
of entities that is specified as a set $R$ of formal transformations $r$ of
the form
\begin{equation}
  \label{eq:reaction} 
  \sum_{x\in X} s^-_{xr} x \longrightarrow
  \sum_{x\in X} s^+_{xr} x \,.
\end{equation}
The non-negative coefficients $s^-_{xr}$ and $s^+_{xr}$ are known as the
\emph{stoichiometric coefficients}. If $s^-_{xr}=s^+_{xr}=0$, then the
entity $x$ does not take part in reaction $r$. If $s^-_{xr}>0$ then $x$ is
a reactant of $r$, and if $s^+_{xr}>0$, then $x$ is said to be the product
of $r$. The matrix $\mathbf{S}$ with entries
$\mathbf{S}_{xr}\coloneqq s^+_{xr}-s^-_{xr}$ is known as the stoichiometric
matrix of the reaction network. Note that all reactions are considered
irreversible. A reversible reaction thus is represented by two reactions:
$r$ and its inverse $\bar r$ satisfying $s^-_{x\bar{r}}=s^+_{xr}$ and
$s^+_{x\bar{r}}=s^-_{xr}$. Moreover $\bar{\bar{r}}=r$.

It is natural to ask under which conditions there actually is a set $M$ of
chemical species, i.e., molecules and an assignment $\varphi:X\to M$ such
that, for each $r\in R$, the assignment of concrete molecules $\varphi(x)$
to the abstract entities $x\in X$ results in real-world chemical reaction
network. We write $\varphi(r)$ for the reaction
\begin{equation}
  \label{eq:assigned-r}
  \sum_{x\in X} s^-_{xr} \varphi(x) \longrightarrow
  \sum_{x\in X} s^+_{xr} \varphi(x) 
\end{equation}
expressed in terms of the chemical species. The abstract reaction networks
that satisfy the laws of mass and atom conservation were characterized in
\cite{Mueller:22a}. Assuming that all reactions are at least conceptually
reversible, it is necessary and sufficient that the abstract reaction
network is \emph{conservative} \cite{Horn:72a}, i.e., that a strictly
positive reaction invariant exists. In practice the molecular mass serves
this purpose. An assignment $\varphi:X\to M$ can always be found in terms
of abstract sum formulas composed of a number of formal moieties determined
by the rank of the stoichiometric matrix $\mathbf{S}$. From these sum
formulas, it is then also possible to assign molecular graphs that are
``chemical'' in the sense that vertices and edges are labeled by atoms
types and bond orders, and chemical valence rules are satisfied
\cite{Mueller:22a}. In other words, an assignment exists if $(X,R)$ is
conservative and are free to choose $M$ at will from the entire universe
$\mathfrak{C}$ of chemical graphs.

It remains an open question, however, if and how more detailed chemical
constraints can be accommodated. Most naturally, it is of immediate
interest to restrict molecules to particular sets or classes of molecules
and to limit reactions to particular types of chemical transformations,
such as named reactions \cite{Li:25}. Here we are concerned with the basic
computational problems that arise in this context. We first ask how a
abstract reaction networks $(X,R)$ can be populated with the given set of
molecules such that each reaction is balanced. We will show here that this
Sudoku-like problem, which we will refer to as \textsc{Balanced Bijection},
is NP-complete and can be phrased as a subgraph isomorphism problem. In a
refined version of the problem, one may additionally require that each
reaction follows one of a prescribed set of reaction mechanisms, specified
as graph-transformation rules \cite{Andersen:16a}, or, equivalently, as
subgraphs of \emph{Imaginary Transition States} \cite{Fujita:86} in the
sense of Hendrickson's classification of reactions \cite{Hendrickson:97},
see also \cite{GonzalezLaffitte:24b}.  We shall see that this additional
restriction does not affect the computational complexity of the problem.

This contribution is organized as follows. We first introduce the necessary
notation and formalize the problem in a very general version and show that
it can be mapped subgraph isomorphism problems. Moreover, we convince
ourselves that the reactions can be restricted with polynomial-time effort
to reactions with pre-defined reaction centers of bounded size.  We then
consider the simplest version, in which the list of molecules is given
completely and only a bijection between abstract entities in the network
and the set of molecules needs to be computed. This problem already is
NP-complete. We describe a backtrace-and-prune algorithm inspired by the
VF2 approach \cite{Cordella:04,Juettner:18} to the subgraph isomorphism
problem that solves interestingly large problems efficiently. Detailed
benchmarking shows that the solutions are practically achievable with
limited computational resources for reaction networks with more than a
hundred species and reactions.  Concluding remarks focus on more general
versions of this problem and the limits of the present approach.

\section{Theory}

\subsection{Chemical Graphs} 

For simplicity of presentation we consider here only neutral molecules that
can be represented with Lewis formulas that do not require unpaired
electrons and explicit charges. In this case we work with an universe
$\mathfrak{C}$ of vertex-labeled \emph{chemical graphs} that allow concise
mathematical definition and are equivalent to Lewis formulas.  Bonding
electron pairs are denoted by edges and non-bonding non-bonding electron
pairs are denoted by loops. The bond order is determined by the number of
parallel edge, i.e., chemical graph are multigraphs with loops.  The degree
$\deg(u)$ of a vertex is the number of incident edges plus twice the number
of loops. The degree of each vertex $\deg(u)$ is then determined completely
by the atom label $a(u)$ in coincides with its number of outer-shell
electrons.
\begin{definition}
  A chemical graph $G\in\mathfrak{C}$ is a vertex labeled multigraph with
  loops such that the degree $\deg(u)$ of a vertex $u$ in $G$ is a function
  of the vertex label $a(u)$.
\end{definition}
This definition of chemical multigraphs is explored in more detail in
\cite{Flamm:25b}. In particular, the we may assume that number of atom
types, i.e., distinct vertex labels, and the maximal vertex degree is
bounded. We do allow, however, that the graphs, i.e., the molecules become
arbitrarily large. The structural formula of a molecule is a connected
chemical graph.  We write $\mathfrak{M}\subseteq\mathfrak{C}$ for the set
of all \emph{molecular graphs}, i.e., the connected chemical graphs.
Throughout, we assume that the map $\varphi:X\to M$ is injective, i.e.,
different abstract entities in $(X,R)$ denote different molecules.

The \emph{sum-formula} $\sigma(m)$ of a molecule $m$ is then simply the
multi-set of vertex labels (or -- more conveniently -- a vector indexed by
the atom types $a$ with entries $\sigma_a(m)$) counting the number of atoms
of type $a$ in molecule $m$. Note that two distinct molecules $m,m'\in M$
are isomers if $\sigma(m)=\sigma(m')$.

\subsection{Chemical Reactions} 

A chemical reaction $r$, correspondingly, is a pair $G \longrightarrow H$
of two chemical graphs with vertex (atom) sets $V(G)$ and $V(H)$,
respectively. Since chemical reactions preserve atoms, the total number of
atoms of each type is the same in reactants and the products of every
reaction $r\in R$:
\begin{equation}
  \label{eq:balanced}
  \sum_{x\in X} s^-_{xr} \sigma(\varphi(x)) =
  \sum_{x\in X} s^+_{xr} \sigma(\varphi(x))
  \quad{or, equivalently}\quad
  \sum_{x\in X}  \sigma(\varphi(x)) \mathbf{S}_xr = 0\,.
\end{equation}
We say that a reaction $r$ is \emph{balanced} if it satisfies
Eq.(\ref{eq:balanced}). We will assume through this contribution that all
reactions must be balanced. In this case $(\sigma_a(x))_{x\in X}$ is a left
kernel vector of $\mathbf{S}$. This statement simply captures the fact the
number of atom of each type $a$ is a reaction invariant for all reactions.
\begin{obs}
  A reaction $G \longrightarrow H$ is balanced if and only there is a
  bijection $\alpha:V(G)\to V(H)$ that conserved vertex labels, i.e.,
  $a(\alpha(u))=a(u)$.
\end{obs}
The bijection $\alpha$ is called the \emph{atom-to-atom map} (AAM).  The
connected components of $G$ and $H$ are the reactant and product molecules,
respectively. In this setting, the stoichiometric coefficients are
represented by a corresponding number of copies the molecule graphs.  It is
not difficult to see that, since every reaction $r$ is balanced, it also
preserves the sum of the vertex degrees, and thus the total number of
electron pairs.

We remark that it is possible in principle to extend this formalism to
Lewis formulas in general, thus accommodating unpaired electrons and
explicit charges. Details on this more general setting can be found in
\cite{Holzschuh:26w}; the additional details do not affect any of the
theoretical results presented in this contribution.

Prop.~47 and Prop.~48 in \cite{Mueller:22a} can be rephrased in the
language of this contribution a follows:
\begin{proposition}
  \label{prop:exists} 
  There exists an injective map $\varphi:X\to \mathfrak{M}$ such that every
  reaction is balanced if and only if the reaction network $(X,R)$ is
  conservative.
\end{proposition}
It is also worth noting that it can be checked in polynomial time whether a
reaction network $(X,R)$ is conservative, i.e., whether there is a strictly
positive reaction invariant \cite{Mueller:22a}. In this contribution we
consider several related problems that arise when additional restrictions
are imposed on the admissible molecules and/or the admissible types of
reactions.

\begin{problem}
  \label{problem:1} 
  Given an abstract reaction network $(X,R)$ and a set $M\subseteq
  \mathfrak{M}$ of molecules, is there an injective map $\varphi: X\to M$
  such that all reactions $\varphi(r)$ are balanced?
\end{problem}

Given an abstract reaction network $(X,R)$ and an assignment $\varphi:X\to
M$ the total size $N$ of the input is given by $|X|+|R|$ plus the total
number of non-zero stoichiometric coefficient (i.e., the size of $(X,R)$,
plus the total size of all allowed molecular graphs in $M$.
\begin{lemma}
  \label{lem:checkvarphi}
  It can be decided in linear time in $N$ whether $\varphi$ is a valid
  assignment of molecular graphs for the abstract network $(X,R)$, whether
  every reaction in $(X,R)$ is balanced for given $\varphi$.
\end{lemma}
\begin{proof}
  It suffices to first parse, in linear time, all molecular graphs $m\in
  \varphi(X)$ into their sum formulas $\sigma(m)$, i.e., into vectors of
  bounded size, and then to check for each reaction $r\in R$ whether
  equ.(\ref{eq:balanced}) holds. The total effort for this step is
  proportional to the total number of non-zero stoichiometric coefficients,
  and thus in $O(N)$.
\end{proof}

\subsection{Reaction mechanisms} 
\label{sect:rm}

For the latter, we formalize the idea of a \emph{reaction mechanism} as
follows. First, we note that, due to the existence of the bijection $\alpha$, the \emph{imaginary transition state graph} $\Upsilon\coloneqq
\Upsilon(G,H,\alpha)$ with vertex set $V(\Upsilon)=V(G)$, vertex labels
$a(u)$ for all $u\in V(G)$ and edges and loops defined as follows: For each
pair $u,v\in V(G)$ we match as many $u,v$ edges in $G$ and
$\alpha(u),\alpha(v)$ edges in $H$. Remaining unmatched $u,v$ edges in $G$
are labeled `\texttt{--}', while remaining unmatched edges in $H$ are
labeled `\texttt{+}'. In chemical terms, edges labeled `\texttt{--}'
correspond to a reduction of bond order or bond breaking, while
`\texttt{+}' denotes the increase of bond order or bond formation (see Fig.~\ref{fig:realizability-problem}B). Loops are
treated analogously.  We refer to \cite{GonzalezLaffitte:24b,Laffitte:25a}
for an equivalent formalism that uses edge labels instead of multiple
edges. The \emph{reaction center} of $G \longrightarrow H$ consists of the
vertices in $V(G)$ that are incident to an edge in $\Upsilon$ that is
labeled `\texttt{--}' or `\texttt{+}'.  A reaction mechanism is the
characterized by a (not necessarily induced) vertex labeled subgraph
$\Gamma$ of $\Upsilon$ that contains the reaction center and all edges
labeled `\texttt{--}' or `\texttt{+}'. From $\Gamma$ one can then extract
the graphs $\Gamma^-$ as $V(\Gamma^-)=V(\Gamma^-)$ and the unlabeled and
`\texttt{--}' labeled edges. Correspondingly, $\Gamma^+$ has the vertices
$V(\Gamma^+)=V(\Gamma^-)$ and the unlabeled and `\texttt{+}' labeled edges
of $\Gamma$. By construction $\Gamma^-$ is a subgraph of $G$ and $\Gamma^+$
is a subgraph $H$. The restriction of the AAM $\hat\alpha$ to $V(\Gamma)$
is label-preserving bijection $V(\Gamma^-)=V(\Gamma^+)$. The triple
$(\Gamma^-,\Gamma^+,\hat\alpha)$ defined a reaction mechanism in which a
local subgraph $\Gamma^-$ is transformed into subgraph $\Gamma^+$. We note
in passing that the triple $(\Gamma^-,\Gamma^+,\hat\alpha)$ can be
re-interpreted as double-pushout graph rewriting rule, see
e.g.\ \cite{Beier:26a}.

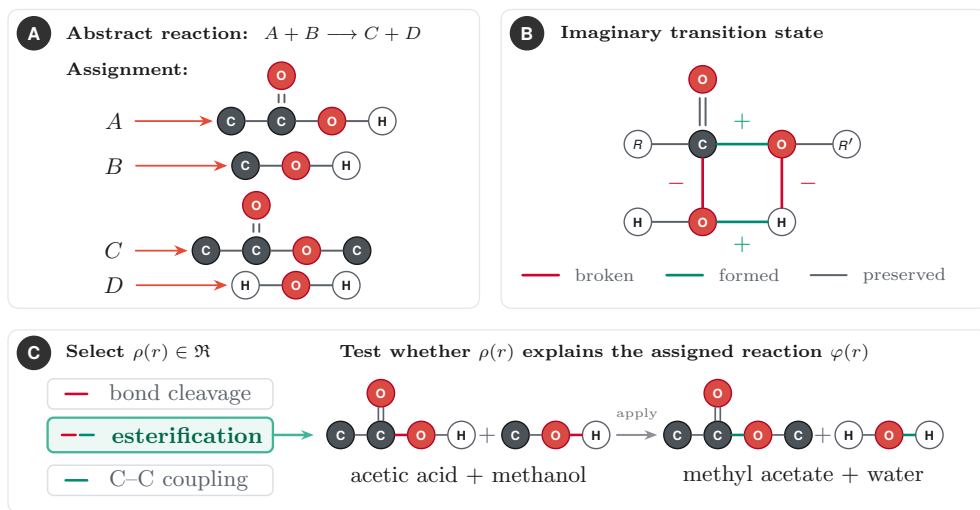
\begin{figure}[htbp]
  \centering
  \resizebox{\linewidth}{!}{\begin{tikzpicture}[x=1cm, y=1cm, font=\rmfamily\footnotesize,
  line cap=round, line join=round,
  fmolC/.style={aC,minimum size=3.8mm,
    font=\figatomfont\bfseries\tiny},
  fmolO/.style={aO,minimum size=3.8mm,
    font=\figatomfont\bfseries\tiny},
  fmolH/.style={aH,minimum size=3.8mm,
    font=\figatomfont\bfseries\tiny},
  fmolR/.style={aR,minimum size=3.8mm,
    font=\figatomfont\itshape\tiny}]

  \draw[jpanel] (0.00,3.20) rectangle ( 6.55,7.35);
  \draw[jpanel] (6.85,3.20) rectangle (13.60,7.35);
  \draw[jpanel] (0.00,0.35) rectangle (13.60,2.90);

  \node[badge] at (0.36,7.02) {A};
  \node[ptitle] at (0.68,7.02) {Abstract reaction:};
  \node[font=\rmfamily\bfseries\footnotesize,text=ink] at (4.65,7.01)
    {$A+B\longrightarrow C+D$};
  \node[ptitle] at (0.68,6.5) {Assignment:};

  \begin{scope}[xshift=5.5mm]
  \foreach \i/\name/\y in {1/A/5.80, 2/B/5.18, 3/C/4.00, 4/D/3.52}{
    \node[lab] (x\i) at (0.92,\y) {$\name$};
  }

  \node[fmolC] (a1) at (2.55,5.80) {C};  \node[fmolC] (a2) at (3.25,5.80) {C};
  \node[fmolO] (a3) at (3.25,6.43) {O};  \node[fmolO] (a4) at (3.95,5.80) {O};
  \node[fmolH] (a5) at (4.65,5.80) {H};
  \draw[bnd] (a1)--(a2);  \draw[bnd] (a2)--(a4);  \draw[bnd] (a4)--(a5);
  \dbond{a2}{a3}

  \node[fmolC] (b1) at (2.75,5.18) {C};  \node[fmolO] (b2) at (3.45,5.18) {O};
  \node[fmolH] (b3) at (4.15,5.18) {H};
  \draw[bnd] (b1)--(b2);  \draw[bnd] (b2)--(b3);

  \node[fmolC] (c1) at (2.20,4.00) {C};  \node[fmolC] (c2) at (2.90,4.00) {C};
  \node[fmolO] (c3) at (2.90,4.63) {O};  \node[fmolO] (c4) at (3.60,4.00) {O};
  \node[fmolC] (c5) at (4.30,4.00) {C};
  \draw[bnd] (c1)--(c2);  \draw[bnd] (c2)--(c4);  \draw[bnd] (c4)--(c5);
  \dbond{c2}{c3}

  \node[fmolH] (d1) at (2.75,3.52) {H};  \node[fmolO] (d2) at (3.45,3.52) {O};
  \node[fmolH] (d3) at (4.15,3.52) {H};
  \draw[bnd] (d1)--(d2);  \draw[bnd] (d2)--(d3);

  \foreach \i/\t in {1/a1, 2/b1, 3/c1, 4/d1}{ \draw[map] (x\i)--(\t); }
  \end{scope}

  \node[badge] at (7.21,7.02) {B};
  \node[ptitle] at (7.55,7.02) {Imaginary transition state};

  \node[fmolC] (C)  at ( 9.65,5.48) {C};
  \node[fmolO] (O1) at (10.75,5.48) {O};
  \node[fmolH] (H)  at (10.75,4.40) {H};
  \node[fmolO] (O2) at ( 9.65,4.40) {O};
  \node[fmolO] (Oc) at ( 9.65,6.38) {O};
  \node[fmolR] (R)  at ( 8.75,5.48) {R};
  \node[fmolR] (Rp) at (11.65,5.48) {R$'$};
  \node[fmolH] (Hb) at ( 8.75,4.40) {H};      

  \draw[bctx] (C)--(R);  \draw[bctx] (O1)--(Rp);  \draw[bctx] (O2)--(Hb);
  \dbond{C}{Oc}

  \draw[bfrm] (C)--(O1)  node[midway,above=2pt,tagp] {$+$};
  \draw[bbrk] (O1)--(H)  node[midway,right=3pt,tagm] {$-$};
  \draw[bfrm] (H)--(O2)  node[midway,below=2pt,tagp] {$+$};
  \draw[bbrk] (O2)--(C)  node[midway,left=3pt,tagm]  {$-$};

  \begin{scope}[shift={(7.15,3.66)}]
    \draw[bbrk] (0.00,0)--(0.50,0);  \node[capt,font=\rmfamily\footnotesize,anchor=west] at (0.60,0) {broken};
    \draw[bfrm] (2.00,0)--(2.50,0);  \node[capt,font=\rmfamily\footnotesize,anchor=west] at (2.60,0) {formed};
    \draw[bctx] (4.00,0)--(4.50,0);  \node[capt,font=\rmfamily\footnotesize,anchor=west] at (4.60,0) {preserved};
  \end{scope}

  \node[badge] at (0.36,2.58) {C};

  \node[ptitle] at (0.68,2.58) {Select $\rho(r)\in\mathfrak{R}$};
  \draw[rulecard] (0.55,1.80) rectangle (3.68,2.22);
  \draw[rulesel]  (0.55,1.20) rectangle (3.68,1.68);
  \draw[rulecard] (0.55,0.60) rectangle (3.68,1.02);

  \draw[bbrk,line width=1.0pt] (0.80,2.01)--(1.10,2.01);
  \node[capt,anchor=west] at (1.28,2.01) {bond cleavage};

  \draw[bbrk,line width=1.0pt] (0.75,1.44)--(0.94,1.44);
  \draw[bfrm,line width=1.0pt] (1.00,1.44)--(1.19,1.44);
  \node[font=\rmfamily\bfseries\scriptsize,text=pfrm!65!black,
        anchor=west] at (1.32,1.44) {esterification};

  \draw[bfrm,line width=1.0pt] (0.80,0.81)--(1.10,0.81);
  \node[capt,anchor=west] at (1.28,0.81) {C--C coupling};

  \draw[rulearr,draw=pfrm!75,line width=1.0pt]
    (3.68,1.44)--(4.25,1.44);

  \node[ptitle] at (4.48,2.58)
    {Test whether $\rho(r)$ explains the assigned reaction $\varphi(r)$};

  \node[fmolC] (rc1) at (4.62,1.44) {C};
  \node[fmolC] (rc2) at (5.18,1.44) {C};
  \node[fmolO] (ro1) at (5.18,2.02) {O};
  \node[fmolO] (ro2) at (5.74,1.44) {O};
  \node[fmolH] (rh1) at (6.30,1.44) {H};
  \draw[bnd,line width=0.65pt] (rc1)--(rc2);
  \draw[bnd,line width=0.65pt] ([xshift=-1pt]rc2.north)--([xshift=-1pt]ro1.south);
  \draw[bnd,line width=0.65pt] ([xshift= 1pt]rc2.north)--([xshift= 1pt]ro1.south);
  \draw[bbrk,line width=1.0pt] (rc2)--(ro2);
  \draw[bnd,line width=0.65pt] (ro2)--(rh1);

  \node[capt] at (6.66,1.44) {+};
  \node[fmolC] (rc3) at (7.05,1.44) {C};
  \node[fmolO] (ro3) at (7.61,1.44) {O};
  \node[fmolH] (rh2) at (8.17,1.44) {H};
  \draw[bnd,line width=0.65pt] (rc3)--(ro3);
  \draw[bbrk,line width=1.0pt] (ro3)--(rh2);

  \node[capt,text=ink] at (6.40,0.88) {acetic acid + methanol};

  \draw[rulearr,draw=ink2!70,line width=0.8pt] (8.46,1.44)--(9.02,1.44);
  \node[font=\rmfamily\tiny,text=mut] at (8.74,1.74) {apply};

  \node[fmolC] (pc1) at (9.30,1.44) {C};
  \node[fmolC] (pc2) at (9.86,1.44) {C};
  \node[fmolO] (po1) at (9.86,2.02) {O};
  \node[fmolO] (po2) at (10.42,1.44) {O};
  \node[fmolC] (pc3) at (10.98,1.44) {C};
  \draw[bnd,line width=0.65pt] (pc1)--(pc2);
  \draw[bnd,line width=0.65pt] ([xshift=-1pt]pc2.north)--([xshift=-1pt]po1.south);
  \draw[bnd,line width=0.65pt] ([xshift= 1pt]pc2.north)--([xshift= 1pt]po1.south);
  \draw[bfrm,line width=1.0pt] (pc2)--(po2);
  \draw[bnd,line width=0.65pt] (po2)--(pc3);

  \node[capt] at (11.32,1.44) {+};
  \node[fmolH] (ph1) at (11.68,1.44) {H};
  \node[fmolO] (po3) at (12.24,1.44) {O};
  \node[fmolH] (ph2) at (12.80,1.44) {H};
  \draw[bnd,line width=0.65pt] (ph1)--(po3);
  \draw[bfrm,line width=1.0pt] (po3)--(ph2);

  \node[capt,text=ink] at (11.05,0.88) {methyl acetate + water};
\end{tikzpicture}}
  \caption{From molecular assignment to reaction-mechanism validation.
  \textbf{(A)} The abstract reaction $A+B\to C+D$ is realized by assigning
  each abstract entity to a candidate molecular graph.
  \textbf{(B)} The imaginary transition state records bonds broken and formed
  by the assigned reaction while retaining the unchanged context.
  \textbf{(C)} A mechanism $\rho(r)\in\mathfrak{R}$ is selected and tested
  to determine whether it explains the assigned reaction $\varphi(r)$; here,
  esterification maps acetic acid and methanol to methyl acetate and water.}
  \label{fig:realizability-problem}
\end{figure}

We may also turn this construction around and start from a pair of chemical
graphs and a reaction mechanism $(\Gamma^-,\Gamma^+,\hat\alpha)$ and two
chemical graphs $G$ and $H$. Now consider a particular subgraph embedding
of $\Gamma^-$ in $G$ and $\Gamma^+$ in $H$. Then the graphs $\hat G$ and
$\hat H$ obtained by deleting from $G$ the edges labeled `\texttt{--}' in
$\Gamma^-$ and from $H$ the edges labeled `\texttt{+}' in $\Gamma^+$. Then
these embeddings \emph{explain} the reaction $G$ to $H$ if (and only if)
the bijection $\hat\alpha$ can extended to a bijection $\alpha:V(G)\to
V(H)$ \cite{GonzalezLaffitte:24b}. Denote by $\mathfrak{R}$ a set of
reaction mechanisms, i.e., triples $(\Gamma^-,\Gamma^+,\hat\alpha)$.  Then
there is an assignment $\rho: R\to\mathfrak{R}$ that associates a mechanism
with each reaction $r\in R$. Importantly, we assume that $|V(\Gamma)|$ is
bounded, i.e., that all reaction mechanisms are determined by subgraphs of
bounded, finite size. Thus $\mathfrak{R}$ is also a finite set. In the
context of graph grammar models of chemistry, see e.g. \cite{Andersen:16a},
this means that all rules have bounded size and that there is also only a
bounded number of possible graph transformation rules.

Restricting the reactions $\varphi(r)$ to certain reaction mechanisms (see Fig.~\ref{fig:realizability-problem}C)
naturally leads to a restricted version of Problem~\ref{problem:1}:
\begin{problem}
  \label{problem:2}
  Given an abstract reaction network $(X,R)$, a set $M\subseteq
  \mathfrak{M}$ of molecules and a set $\mathfrak{R}$ of reaction
  mechanism, is there an injective map $\varphi: X\to M$ such that all
  reactions $\varphi(r)$ are balanced and explained by a reaction mechanism
  $\rho(r)\in\mathfrak{R}$.
\end{problem}

In this case, it is less obvious that a solution to the problem, i.e., a
pair $(\varphi,\rho)$, can indeed be verified efficiently. In fact, this
depends on the boundedness of the patterns defining the reaction mechanisms
and the bounded degree of the chemical graphs. 

\begin{lemma}
  \label{lem:checkrho}
  Let $\mathfrak{R}$ be a set of reaction mechanisms with bounded size, and
  let $G\longrightarrow H$ be a reaction. Then it can be decided in
  polynomial time (w.r.t.\ to the size of $G$) whether there is a mechanism
  $(\Gamma^-,\Gamma^+,\hat\alpha)\in\mathfrak{R}$ that explains the
  reaction $G\longrightarrow H$.
\end{lemma}
\begin{proof}
  Since $|V(\Gamma^-)|\le K$ for some finite bound $K$, the vertex degree
  of a chemical graphs is bounded by the maximum valency of an atom, it is
  possible to list all subgraph matches of $\Gamma^-$ in $G$ and all
  subgraph matches of $\Gamma^+$ in $H$ in on $O(n^K)$ total time, where
  $n$ is the size of the input graphs $G$ and $H$, respectively. In each
  positive case the graphs $\hat G$ and $\hat H$ can be constructed in
  linear time.  Since both $\hat G$ and $\hat H$, being subgraphs of $G$
  and $H$, have bounded degree, it is possible to test whether they are
  isomorphic in polynomial time using the variant of Luks' algorithm for
  colored graphs \cite{Luks:82} using unique colors for each of the pairs
  of vertices in the embeddings of $\Gamma^-$ in $G$ and $\Gamma^+$ in $H$
  that correspond to each other via $\hat\alpha$. The total running time
  therefore remains polynomial.
\end{proof}

Consider the problem of checking, for a given abstract reaction network
$(X,R)$, an assignment of molecules $\varphi:X\to\mathfrak{M}$ to the abstract entities, and an assignment $\rho:R\to\mathfrak{R}$ of mechanisms to
the reactions, whether all reactions $r$ are balanced and can be explained
by the mechanisms $\rho(r)$. The extra input, $\rho$, compared to
Lemma~\ref{lem:checkvarphi} comprises at most $|R|$ bounded size triples
$(\Gamma^-,\Gamma^+,\hat\alpha)$ and thus is bounded by $N$. By
Lemma~\ref{lem:checkvarphi}, validity of $\varphi$ can be checked in linear
time, and Lemma~\ref{lem:checkrho} yields polynomial effort for validating
$\rho$:
\begin{proposition}
  Let $(X,R)$ be an abstract reaction network,
  $\varphi:X\to\mathfrak{M}$ and $\rho:R\to\mathfrak{R}$. Then it can be
  verified in polynomial time in $N$ whether all reactions $r\in R$ are
  balanced and explained by $\rho(r)$.
\end{proposition}

\subsection{K{\"o}nig graphs}

A reaction network is faithfully represented by a directed bipartite
multigraph $\king(X,R)$ with vertex set $V=X\cup R$, called the K{\"o}nig
graph in the context of hypergraphs. It is also known as the
species-reaction (SR) graph in the literature on (chemical) reaction
networks, see e.g.\ \cite{Feinberg:19}. The edges of $\king(X,R)$ and their
multiplicities are defined by the non-zero stoichiometric coefficient:
There is a directed edge with multiplicity $s^-_{xr}$ from $x$ to $r$ and a
directed edge with multiplicity $s^+_{xr}$ from $r$ to $x$. Note that
the direction is important since otherwise reactants and products cannot be
be distinguished. The graph $\king(X,R)$ is known as the K{\"o}nig graph of
the reaction network. Figure~\ref{fig:koenig}A--C shows how three
formal reactions are encoded as a directed bipartite graph.
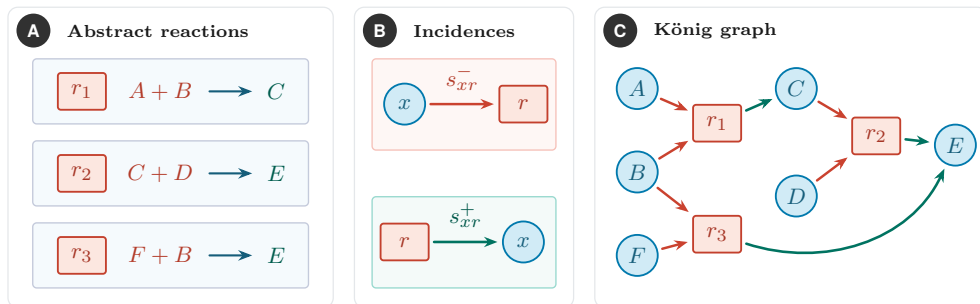
\begin{figure}[t]
  \centering
  \resizebox{\linewidth}{!}{\begin{tikzpicture}[
  x=1cm,y=1cm,font=\rmfamily\footnotesize,
  line cap=round,line join=round,
  species/.style={circle,draw=npgBlue!85!black,fill=npgBlue!18,
    line width=0.75pt,minimum size=5.8mm,inner sep=0pt,
    font=\rmfamily\bfseries\scriptsize,text=npgNavy},
  reaction/.style={rounded corners=1.2pt,draw=npgRed!85!black,
    fill=npgSalmon!24,line width=0.75pt,minimum width=6.8mm,
    minimum height=5.2mm,inner sep=1pt,
    font=\rmfamily\bfseries\scriptsize,text=npgRed!72!black},
  reactionrow/.style={rounded corners=1.6pt,draw=npgGrayBlue!48,
    fill=npgBlue!4,line width=0.55pt,minimum width=39mm,
    minimum height=9.2mm,inner sep=2pt},
  reactantcard/.style={rounded corners=1.6pt,draw=npgRed!34,
    fill=npgSalmon!8,line width=0.55pt},
  productcard/.style={rounded corners=1.6pt,draw=npgGreen!40,
    fill=npgTeal!10,line width=0.55pt},
  inedge/.style={-{Stealth[length=1.9mm,width=1.45mm]},
    draw=npgRed!88!black,line width=1.0pt,
    shorten >=1.8pt,shorten <=1.8pt},
  outedge/.style={-{Stealth[length=1.9mm,width=1.45mm]},
    draw=npgGreen!72!black,line width=1.05pt,
    shorten >=1.8pt,shorten <=1.8pt}
]
  \draw[spanel] (-0.35,-0.25) rectangle (4.25,4.00);
  \draw[spanel] ( 4.55,-0.25) rectangle (7.65,4.00);
  \draw[spanel] ( 7.95,-0.25) rectangle (13.50,4.00);

  \node[badge] at (0.01,3.67) {A};
  \node[ptitle] at (0.36,3.67) {Abstract reactions};

  \foreach \i/\y/\lhs/\rhs in {
    1/2.81/{A+B}/C,
    2/1.65/{C+D}/E,
    3/0.49/{F+B}/E}{
    \node[reactionrow] at (1.95,\y) {};
    \node[reaction] at (0.695,\y) {$r_{\i}$};
    \node[font=\rmfamily\bfseries\small,text=npgRed!78!black,inner sep=0pt]
      at (1.805,\y) {$\lhs$};
    \node[font=\rmfamily\bfseries\small,text=npgGreen!58!black,inner sep=0pt]
      at (3.445,\y) {$\rhs$};
    \draw[-{Stealth[length=2.0mm,width=1.5mm]},draw=npgNavy,
      line width=0.85pt] (2.50,\y)--(3.105,\y);
  }

  \node[badge] at (4.91,3.67) {B};
  \node[ptitle] at (5.26,3.67) {Incidences};

  \draw[reactantcard] (4.80,1.98) rectangle (7.40,3.27);
  \node[species] (xin) at (5.26,2.63) {$x$};
  \node[reaction] (rin) at (6.94,2.63) {$r$};
  \draw[inedge] (xin)--(rin)
    node[midway,above=2pt,font=\rmfamily\bfseries\scriptsize,
      text=npgRed!78!black] {$s^-_{xr}$};

  \draw[productcard] (4.80,0.03) rectangle (7.40,1.32);
  \node[reaction] (rout) at (5.26,0.68) {$r$};
  \node[species] (xout) at (6.94,0.68) {$x$};
  \draw[outedge] (rout)--(xout)
    node[midway,above=2pt,font=\rmfamily\bfseries\scriptsize,
      text=npgGreen!55!black] {$s^+_{xr}$};

  \node[badge] at (8.31,3.67) {C};
  \node[ptitle] at (8.66,3.67) {König graph};

  \node[species] (A) at (8.55,2.85) {$A$};
  \node[species] (B) at (8.55,1.67) {$B$};
  \node[species] (F) at (8.55,0.49) {$F$};
  \node[species] (C) at (10.80,2.85) {$C$};
  \node[species] (D) at (10.80,1.33) {$D$};
  \node[species] (E) at (13.05,2.05) {$E$};

  \node[reaction] (r1) at (9.68,2.36) {$r_1$};
  \node[reaction] (r3) at (9.68,0.79) {$r_3$};
  \node[reaction] (r2) at (11.93,2.18) {$r_2$};

  \draw[inedge] (A)--(r1);
  \draw[inedge] (B)--(r1);
  \draw[inedge] (F)--(r3);
  \draw[inedge] (B)--(r3);
  \draw[inedge] (C)--(r2);
  \draw[inedge] (D)--(r2);

  \draw[outedge] (r1)--(C);
  \draw[outedge] (r2)--(E);
  \draw[outedge] (r3) to[out=-18,in=245] (E);
\end{tikzpicture}}
  \caption{Directed K{\"o}nig-graph construction.
    (A) Example reaction network.
    (B) Reactant incidences point from species to reactions
    (red), whereas product incidences point from reactions to species
    (green).
    (C) The resulting directed species--reaction bipartite
    graph $\king(X,R)$.}
  \label{fig:koenig}
\end{figure}

\noindent The formal sums
\begin{equation*}
    \sum_{x\in X} s^-_{xr} x \qquad\text{and}\qquad
    \sum_{x\in X} s^+_{xr} x 
\end{equation*}
in equ.(\ref{eq:reaction}) are called the \emph{complexes} of $(X,R)$ and
consist, for every $r\in R$ of the in-neighbors and out-neighbors,
respectively, with multiplicities given by the stoichiometric coefficients,
and thus the multiplicity of directed $(x,r)$ and $(r,x)$ edges,
respectively. Let
\begin{equation}
  \label{eq:delta}
  \delta(X,R) \coloneqq
  \max_{r\in R} \max \left\{\sum_{x\in X} s^-_{xr},
                            \sum_{x\in X} s^+_{xr}\right\}
\end{equation}
be the maximal in- or out-degree of a reaction vertex in the K{\"o}nig
graph, i.e., the maximal number of molecules appearing as a reactant or
product complex of a reaction in $R$. Denote the set of complexes in
$(X,R)$ by $\mathcal{C}(X,R)$ and let $\mathfrak{M}^{(\delta)}$ be the set
of multisets of molecules with cardinality at most $\delta$. Since the
complexes of $\mathcal{C}(X,R)$ correspond to multisets with a most
$\delta$ elements (counting multiplicities), the construction of
$\mathfrak{M}^{(\delta)}$ ensures that it harbors the image of all
complexes for any possible map $\varphi$:
\begin{obs}
  Let $(X,R)$ be a reaction network, $\delta\coloneqq \delta(X,R)$ given by
  equ.(\ref{eq:delta}), and $\varphi:X\to\mathfrak{M}$ any injective map.
  The for every complex $G\in\mathcal{C}(X,R)$ it holds that
  $\varphi(G)\in\mathfrak{M}^{(\delta)}$.
\end{obs}
For every $U\in \mathfrak{M}^{(\delta)}$ the composition vectors
$\sigma(U)$ is well defined. Thus balanced reactions $U\longrightarrow W$
with $U,W\in \mathfrak{M}^{(\delta)}$ are well-defined by requiring
$\sigma(U)=\sigma(W)$. Moreover, if a set of admissible reaction mechanisms
$\mathfrak{R}$ is given, Lemma~\ref{lem:checkrho} ensures that it can be
checked in polynomial time whether or not $U\longrightarrow W$ can be
explained by a reaction mechanism in $\mathfrak{R}$. Let us denote by
$\mathfrak{E}^*$ the set of balanced reactions, or balanced
reactions with an explanation in $\mathfrak{R}$, respectively. Then
$(\mathfrak{M},\mathfrak{E}^*)$ is in particular a conservative reaction
network and hence there is a well-defined K{\"o}nig graph
$\king(\mathfrak{M},\mathfrak{E}^*)$.  

\begin{theorem}
  \label{thm:subgraphiso}
  Suppose $\delta$ is bounded and $\mathcal{M}$ is finite. Then
  Problem~\ref{problem:1} and Problem~\ref{problem:2} have a \texttt{yes}
  answer if and only if $\king(X,R)$ is isomorphic to an induced subgraph
  of $\king(\mathfrak{M},\mathfrak{E}^*)$.
\end{theorem}
\begin{proof}
  Consider an injective map $\varphi:X\to\mathfrak{M}$ and a reaction
  $r=(G\longrightarrow H)$ in $(X,R)$. Then $q\coloneqq
  (\varphi(G),\varphi(H))\in \mathfrak{E}^*$ if and only if
  $\sigma(\varphi(G))=\sigma(\varphi(H)))$. Since $q$ is uniquely defined
  because $R$ is a set, we may set $q=\varphi(r)$. Now consider the edges
  $(x,r)$ with multiplicity $s^-_{xr}$ and $(r,x)$ with multiplicity
  $s^+_{xr}$ in $\king(X,R)$. Then there are edges
  $(\varphi(x),\varphi(r))$ with multiplicity $s^-_{\varphi(x)\varphi(r)}$
  and $(\varphi(r),\varphi(x))$ with multiplicity
  $s^+_{\varphi(r)\varphi(x)}$ if and only if $s^-_{xr}\ne 0$ and
  $s^+_{xr}\ne 0$, respectively. We distinguish two cases: (i) if
  $\varphi(r)\in\mathfrak{E}^*$ for all $r\in R$, then $\king(X,R)$ is
  isomorphic to a subgraph of $\king(\mathfrak{M},\mathfrak{E}^*)$ and
  $\varphi$ is a witness for a \texttt{yes} instance of
  Problem~\ref{problem:1} or Problem~\ref{problem:2}. (ii) Otherwise, there
  is some reaction $r=(G\longrightarrow H)\in R$ for which
  $\sigma(\varphi(G))\ne\sigma(\varphi(H))$ and thus
  $(\varphi(G),\varphi(H))\notin\mathfrak{E}^*$. In particular, therefore,
  there is no reaction vertex $q\in\king(\mathfrak{M},\mathfrak{E}^*)$ with
  in-neighborhood $\varphi(G)$ and out-neighborhood $\varphi(G)$ with edge
  multiplicities given by $s^-_{\varphi(x)q}$ and $s^+_{\varphi(x)q}$,
  respectively, and hence $\varphi$ does not extend to a subgraph
  isomorphism from $\king(X,R)$ into
  $\king(\mathfrak{M},\mathfrak{E}^*)$. Thus $(X,R)$ and $\mathfrak{M}$
  form a \texttt{no} instance of Problem~\ref{problem:1} or
  Problem~\ref{problem:2} if and only if no injective map $\varphi:X\to R$
  exists that extends to subgraph isomorphism.
\end{proof}

For any finite set $\mathcal{M}$ of molecular graphs, therefore, we can
transform the Problems~\ref{problem:1} and~\ref{problem:2} in polynomial
time into an induced subgraph isomorphism (\textsc{ISI}) problem for
bipartite graphs. The \textsc{ISI} is NP-complete since it in particular
subsumes both the \textsc{Independent Set} and \textsc{Clique} problem. It
remains hard when both $G$ and the putative subgraph $H$ are
disjoint unions of paths, and thus is particular for bipartite graphs
\cite{Damaschke:91}. The NP-completeness proof for \textup{\textsc{Balanced
    Bijection}} below shows, in fact, that the finite variants considered
here are NP-complete: Problem~\ref{problem:1} is NP-hard because
\textup{\textsc{Balanced Bijection}} is a special case. Moreover, the
  same reduction translated to condensation reactions for specific sets of
  oligomers, gives NP-hardness for Problem~\ref{problem:2}.  Together with
Lemmas~\ref{lem:checkvarphi} and~\ref{lem:checkrho}, for bounded
mechanisms, both problems are in NP and hence NP-complete in this finite
setting.

\section{Balanced Bijections}

\subsection{NP-Completeness} 

In this section, we restrict the general realizability question to
scenarios where the molecular search space $\mathfrak{M}$ is known in
advance. Moreover, we ignore the complication of prescribed reaction
mechanisms for the moment and consider the following restricted version
of the Problem~\ref{problem:1}:
\begin{problem} \textup{\textsc{Balanced Bijection.}}
  \label{problem:M=X}
  Given an abstract reaction network $(X,R)$ and set
  $M\subseteq\mathfrak{M}$ of molecules with $|M|=|X|$, is there a
  bijective map $\varphi: X\to M$ such that all reactions are balanced?
\end{problem}

\begin{theorem}
  \label{thm:BB-NPc}
  \textup{\textsc{Balanced Bijection}} is NP-complete. This remains true
  even if every reaction is of the restricted form $x+y\longrightarrow z$
  with $x,y,z\in X$ and $M\subseteq\mathbb{N}$.
\end{theorem}
\begin{proof}
  Clearly the problem is in NP since a bijection $\varphi:X\to M$ serves as
  a certificate that can be checked in linear time by
  Lemma~\ref{lem:checkvarphi}. We prove NP-hardness by reducing from
  \textsc{Hamiltonian Path} in simple undirected graphs $G=(V,E)$ with
  $V=\{1,\dots,n\}$ and $|E|=m$. If $G$ is disconnected and $n>1$, then it
  has no Hamiltonian path and the reduction may output any fixed
  \texttt{no}-instance. One easily checks that $X=\{a,b,c\}$, $R=\{a+b\to
  c\}$ and $M=\{1,2,4\}$ has no balanced bijection $\varphi$ because no two
  distinct elements of $M$ sum to the third. Thus it is a
  \texttt{no}-instance. If $n=1$, $G$ has a Hamiltonian path consisting of
  the single vertex, and the reduction may output the trivial
  \texttt{yes}-instance $X=\{a\}$, $R=\emptyset$, $M=\{1\}$. From now on
  we assume that $n\ge 2$ and $G$ is connected, i.e., $m\ge n-1$.

  The key ingredient is a Sidon encoding of the vertices of $G$. Let $p$
  be a fixed prime number satisfying $n\le p\le 2n$. The existence of $p$
  is guaranteed by Bertrand's postulate, see \cite{Ramanujan:1919} for
  a short proof. Such an integer can be found in linear time, see
  e.g.\ \cite{Gries:78}. For $t=0,\dots,p-1$ we define
  \begin{equation*}
    s_t \coloneqq 2pt + (t^2\bmod p)
  \end{equation*}

  \par\noindent\textbf{Claim:} The set $S_p\coloneqq \{s_0,\dots,s_{p-1}\}$ 
  has the Sidon property, i.e., $s_a+s_b=s_c+s_d$ implies
  $\{a,b\}=\{c,d\}$.
  \begin{small}
  \par\noindent\textit{Proof of Claim.} Writing $r_t\coloneqq t^2\bmod p$
  yields $2p(a+b)+(r_a+r_b)=2p(c+d)+(r_c+r_d)$. Since both residue sums lie
  between $0$ and $2p-2$ we observe that $a+b=c+d$. Moreover, we have
  $2pa+2pb +(a^2+b^2 \bmod p)=2pc+2pd +(c^2+d^2 \bmod p)$ and thus
  $a^2+b^2\equiv c^2+d^2$ $\mod p$. Together this implies $ab\equiv cd$
  $\mod p$, and thus the two unordered pairs are the same in
  $\mathbb{F}_p$, and hence also the same as pairs of integers in
  $\{0,\dots,p-1\}$. \hfill$\triangleleft$ 
  \end{small}
  \medskip

  The sequence $s_t$ is increasing, with consecutive values differing by at
  least $2p-(p-1)=p+1>0$. To each vertex $i$ in $G$ we assign the code
  $c_i\coloneqq s_{i-1}$. Thus the largest vertex code is $c_n=s_{n-1}$.
  Now set $L\coloneqq s_{n-1}$ and $K\coloneqq 2L+1$. For every vertex
  $i\in V$ and every edge $ij\in E$ in $G$ we create the formulas
  \begin{equation}
    \begin{split}
      v_i    &\coloneqq K + c_i\\
      v_{ij} &\coloneqq v_i+v_j = 2K+c_i+c_j
    \end{split}
  \end{equation}
  and set $M\coloneqq\{v_i:i\in V\}\cup\{v_{ij}:ij\in
  E\}\subseteq\mathbb{N}$.  The vertex formulas are pairwise distinct
  because the $s_i$ and hence the $c_i$ are pairwise distinct. The edge
  formulas are pairwise distinct because the sums $v_i+v_j$ are pairwise
  distinct due the Sidon property.  and vertex and edge formulas are
  distinct because of the choice of the offset. Thus $|M|=|V|+|E|=n+m$.

  The entity set $X$ consists of three types: $p_1,\dots,p_n$ are positions
  in a possible Hamiltonian path, $q_1,\dots,q_{n-1}$ are represent the
  edges between consecutive positions in the Hamiltonian path, and
  $g_1,\dots,g_{m-n+1}$ are unused garbage species for the remaining edge
  formulas. We have $|X|=n+(n-1)+(m-n+1)=n+m=|M|$ entities. The reaction
  set comprises for $k=1,\dots,n-1$ the reaction
  \begin{equation}
    p_k+p_{k+1}\longrightarrow q_k
  \end{equation}
  There are no reactions involving the garbage spaces $g_j$.
  \smallskip

  \par\noindent\textbf{Claim:} If $G$ has a Hamiltonian path then
  $(X,R)$ has balanced assignment $\varphi:X\to M$. 
  \par\noindent\textit{Proof of Claim.} Let $i_1,\dots i_n$ be
  a Hamiltonian path in $G$. Set $\varphi(p_k)=v_{i_k}$ for $k=1,\dots,n$.
  Since $i_ki_{k+1}\in E$, we have the edge formula $v_{i_k i_{k+1}}\in
  M$ and hence we can assign $\varphi(q_k)=v_{i_k i_{k+1}}$ for
  $k=1,\dots,n-1$. By construction every reaction is balanced:
  \begin{equation}
    \varphi(p_k)+\varphi(p_{k+1})
    =
    v_{i_k}+v_{i_{k+1}}
    =
    v_{i_k i_{k+1}}
    =
    \varphi(q_k).
  \end{equation}
  All remaining edge formulas are assigned arbitrarily and bijectively to the
  garbage species $g_j$. Thus $\varphi$ is a balanced bijection.
  \hfill$\triangleleft$ 

  \par\noindent\textbf{Claim:} Suppose $\varphi:X\to M$ is a balanced
  bijection for $(X,R)$. Then $G$ has a Hamiltonian path. Recall that
  vertex formulas lie in the interval $[K,K+L]$, while edge formulas lie in
  $[2K,2K+2L]$. Since $K=2L+1$ these intervals are separated, and an edge
  formula plus anything is already too large to be the product of any
  reaction. Now consider a reaction $p_k+p_{k+1}\longrightarrow q_k$.  If
  either reactant is assigned an edge formula, the sum of the reactant
  formulas would be at least $3K$. This is impossible since in $w\in M$
  satisfied $w\le 2K+2L<3K$. Thus both reactant must be assigned a vertex
  formula. Their sum is at least $2K$, hence the product $q_k$ must be
  assigned and edge formula. Since $n\ge 2$, every position species
  $p_1,\dots,p_n$ in at least one of the reactions, and hence all of them
  are assigned a vertex formula. Since there exactly $n$ position species
  and $n$ vertex formulas, the assignment is bijection and hence there is a
  permutation $i_1,\dots, i_n$ of $V$ such that $\varphi(p_k)=v_{i_k}$. For
  each $k$, the reaction $p_k+p_{k+1}\longrightarrow q_k$ is balanced,
  hence $\varphi(q_k)=v_{i_k}+v_{i_{k+1}}=2K+c_{i_k}+c_{i_{k+1}}$. Since
  $\varphi(q_k)\in M$, is corresponds to some vertex formula
  $v_{ab}=2K+c_a+c_b$ with $ab\in E$. The Sidon property now implies
  $\{i_k,i_{k+1}\}=\{a,b\}$ and thus $\{i_k,i_{k+1}\}=\{a,b\}$. Therefore,
  every consecutive pair in $i_1,i_2,\dots,i_n$ is an edge in $G$. Since
  the sequence contains every vertex exactly once, it is a Hamiltonian
  path. \hfill$\triangleleft$

  The construction has polynomial size.  Indeed, the numbers used are
  bounded by $O(n^2)$, so their binary encodings have $O(\log n)$ bits.
  Hence \textsc{Balanced Bijection} is NP-hard.  Together with membership
  in NP, the problem is NP-complete.
\end{proof}  
  
We argue, finally, that the problem remains difficult if reaction
mechanisms are prescribed. To see this, it suffices to consider
condensation reactions of the generic form shown in
Fig.~\ref{fig:condensation-mechanism}A,
where $[w]$ denotes an oligomeric structure consisting of $w$ copies of the
same moiety. Formally, if this byproduct is not suppressed in the notation,
we add a single waste species $w_{\mathrm{XY}}$ with molecule \chemfig{XY}
and replace each reaction $p_k+p_{k+1}\longrightarrow q_k$ in the reduction
by $p_k+p_{k+1}\longrightarrow q_k+w_{\mathrm{XY}}$. Since the same fixed
waste molecule occurs in every such reaction, subtracting this constant
term recovers exactly the balance equations used in the NP-completeness
proof above. A real-life chemical example is the condensation of
homo-peptides. In this case, \chemfig{X}$=$\chemfig{H},
\chemfig{Y}$=$\chemfig{OH}, and the monomeric unit is e.g.\ glycine,
\chemfig{H-NHCH_2CO-OH}, minus the N-terminal \chemfig{H} and the
C-terminal \chemfig{OH}. The reaction mechanism is then explicitly given by the peptide-forming rule in
Fig.~\ref{fig:condensation-mechanism}B.

\begin{figure}[htbp]
  \centering
  \resizebox{\linewidth}{!}{
%
%
\begin{tikzpicture}[
  x=1cm,y=1cm,font=\rmfamily\small,
  line cap=round,line join=round,
  leg/.style={anchor=west,font=\rmfamily\footnotesize,text=ink2}
]
  \newcommand{\rowscale}{1.02}

  \draw[spanel] (0.00,3.15) rectangle (13.60,4.65);
  \draw[spanel] (0.00,0.00) rectangle (13.60,2.95);
  \node[badge] at (0.36,4.31) {A};
  \node[badge] at (0.36,2.61) {B};

  \node[ptitle] at (0.68,4.31) {Generic rule};
  \node[ptitle] at (0.68,2.61) {Glycine instance};

  \node[inner sep=0] at (6.80,3.73) {%
    \scalebox{\rowscale}{%
      \normalsize
      $X-[u]\mathbin{\textcolor{pbrk}{-}}Y
      \;+\;
      X\mathbin{\textcolor{pbrk}{-}}[v]-Y
      \quad\longrightarrow\quad
      X-[u]\mathbin{\textcolor{pfrm}{-}}[v]-Y
      \;+\;
      X\mathbin{\textcolor{pfrm}{-}}Y$%
    }%
  };

  \node[inner sep=0] at (6.80,1.43) {%
    \scalebox{\rowscale}{%
      \normalsize
      \chemfig{-C(-[:90]H)(-[:270]H)-C(=[:90]O)%
        -[:0,,,,draw=pbrk,line width=1.35pt]OH}\;\raisebox{-2pt}{$+$}\;
      \raisebox{-3.5pt}{\chemfig{H-[,,,,draw=pbrk,line width=1.35pt]%
        N(-[:90]H)-C(-[:90]H)(-[:270]H)-}}\quad
      \raisebox{-2pt}{$\longrightarrow$}\quad
      \chemfig{-C(-[:90]H)(-[:270]H)-C(=[:90]O)%
        -[,,,,draw=pfrm,line width=1.45pt]%
        N(-[:90]H)-C(-[:90]H)(-[:270]H)-}\;\raisebox{-2pt}{$+$}\;
      \raisebox{-3.5pt}{\chemfig{H-[,,,,draw=pfrm,line width=1.45pt]OH}}%
    }%
  };

  \draw[draw=pbrk,line width=1.35pt] (4.88,0.34)--(5.38,0.34);
  \node[leg] at (5.50,0.34) {broken};
  \draw[draw=pfrm,line width=1.45pt] (7.72,0.34)--(8.22,0.34);
  \node[leg] at (8.34,0.34) {formed};
\end{tikzpicture}}
  \caption{(A) Generic oligomer condensation with fixed byproduct \(XY\).
    (B) Peptide-bond formation with water release.}
  \label{fig:condensation-mechanism}
\end{figure}
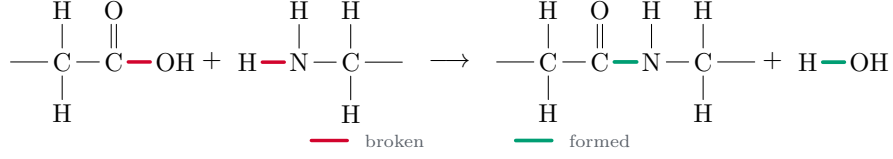
We note that for $u\ne v$ there are in fact two inequivalent applications
of this rule, depending on whether the N-terminus of $u$ or $v$ reacts; both
yield the same reaction product, however. 

The construction shows that the reduction in the proof of
  Thm.~\ref{thm:BB-NPc} can be ``implemented'' e.g.\ by the condensation of
  glycine oligomers with peptide formation as the only reaction rule. Hence
  Problem~\ref{problem:2} is not easier than
  Problem~\ref{problem:1}. Moreover, Lemma~\ref{lem:checkrho} ensures that
  the applicability of a bounded rule can be performed in polynomial time.
  Hence we have
\begin{corollary}
  \textup{\textsc{Balanced Bijection}} remains NP-complete if a set
  $\mathfrak{R}$ of admissible reaction mechanisms of bounded size is
  prescribed.
\end{corollary}

\subsection{A Backtrace and Prune Algorithm} 

Problem~\ref{problem:M=X} can reduced to \textsc{SIS} as we have seen in
Theorem~\ref{thm:subgraphiso} and thus it can be solved at least in
principle by any of the commonly employed methods for this much more common
problem. However, even for $\delta=2$ the graph $\Gamma$ has $O(|X|^4)$
vertices, and thus will in general much larger than $\king(X,R)$. We can
thus expect that solving the equivalent induced subgraph isomorphism
problem is impractical.

The idea of the most common Backtracing and Pruning approaches use for
\textsc{SIS}, such as Ullmann's algorithm \cite{Ullmann:76} or VF2 and its
variants \cite{Cordella:04,Juettner:18}, however, also can be used to solve
Problem~\ref{problem:M=X} directly. To this end we introduce a set of
matches $\Phi\subseteq X\times M$ such that $(x,m), (x',m')\in\Phi$ and
$x=x'$ iff $m=m'$. Thus $\Phi$ encodes part of the desired injective map
$\varphi:X'\to M$ for a subset $X(\Phi)\subseteq X$. Moreover, we denote by
$\overline{X}\coloneqq X\setminus X(\Phi)$ the set of as yet unassigned
entities. For each $x\in X$, we denote by $\mathcal{D}_x$ its set of
possible assignments. The key algorithmic idea of backtracing and pruning
is to construct a search tree whose vertices are the partial assignments
$\Phi$ such that the children of $\Phi$ are of the form
$\Phi'=\Phi\cup\{(x^*,m^*)\}$ where $x^*\in\overline{X}$ and
$m^*\in\mathcal{D}_{x^*}$. Of course, $\overline{X}$ and the list of
remaining candidate matches $\mathcal{D}$ must be updated in each step by
removing entries that are no longer feasible. Denote by $\mathcal{D}'$ the
candidate list obtained from $\mathcal{D}$ after adding the match
$(x^*,m^*)$ to $\Phi$. At the very least, all pairs $(x^*,m)$ and $(x,m^*)$
are removed in $\mathcal{D}'$. Clearly, a branch can be abandoned if and
only if $\mathcal{D}'_{y}=\emptyset$ for some $y\in\overline{X}$. In this
case we set $\mathcal{D'}=\emptyset$. This branching and pruning process is
illustrated in Fig.~\ref{fig:searchtree}.

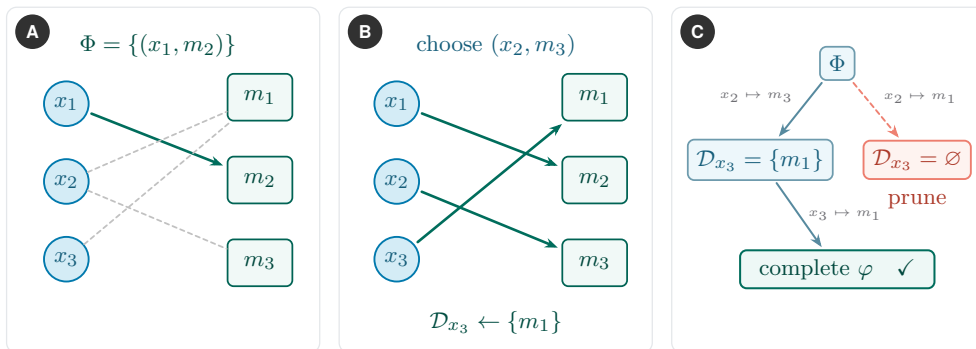
\begin{figure}[htbp]
  \centering
  \resizebox{\linewidth}{!}{\begin{tikzpicture}[
  x=1cm,y=1cm,font=\rmfamily\footnotesize,
  line cap=round,line join=round,
  entity/.style={circle,draw=npgBlue!85!black,fill=npgBlue!17,
    line width=0.75pt,minimum size=6.2mm,inner sep=0pt,
    font=\rmfamily\bfseries\scriptsize,text=npgNavy},
  molecule/.style={rounded corners=2.4pt,draw=npgGreen!62!black,
    fill=npgTeal!13,line width=0.7pt,minimum width=9mm,
    minimum height=6.5mm,font=\rmfamily\scriptsize,
    text=npgGreen!45!black},
  candidate/.style={draw=ink!27,line width=0.7pt,
    dash pattern=on 1.8pt off 1.5pt},
  chosen/.style={-{Stealth[length=1.8mm,width=1.35mm]},
    draw=npgGreen!68!black,line width=1.05pt,
    shorten >=1.5pt,shorten <=1.5pt},
  removed/.style={draw=npgRed!75,line width=0.85pt,
    dash pattern=on 2pt off 1.4pt},
  live/.style={rounded corners=2.6pt,draw=npgNavy!62,
    fill=npgBlue!7,line width=0.7pt,inner sep=3.5pt,
    font=\rmfamily\scriptsize,text=npgNavy},
  reject/.style={rounded corners=2.6pt,draw=npgRed!72,
    fill=npgSalmon!12,line width=0.75pt,inner sep=3.5pt,
    font=\rmfamily\scriptsize,text=npgRed!75!black},
  solution/.style={rounded corners=2.6pt,draw=npgGreen!68!black,
    fill=npgTeal!15,line width=0.8pt,inner sep=3.5pt,
    font=\rmfamily\scriptsize,text=npgGreen!45!black},
  branch/.style={-{Stealth[length=1.8mm,width=1.35mm]},
    draw=npgNavy!72,line width=0.8pt,shorten >=1pt,shorten <=1pt},
  deadbranch/.style={-{Stealth[length=1.8mm,width=1.35mm]},
    draw=npgRed!70,line width=0.75pt,dash pattern=on 2pt off 1.5pt,
    shorten >=1pt,shorten <=1pt}
]
  \draw[spanel] (0.00,0.00) rectangle (4.35,4.80);
  \draw[spanel] (4.62,0.00) rectangle (8.97,4.80);
  \draw[spanel] (9.24,0.00) rectangle (13.60,4.80);
  \node[badge] at (0.36,4.46) {A};
  \node[badge] at (4.98,4.46) {B};
  \node[badge] at (9.60,4.46) {C};

  \node[entity] (Ax1) at (0.83,3.46) {$x_1$};
  \node[entity] (Ax2) at (0.83,2.38) {$x_2$};
  \node[entity] (Ax3) at (0.83,1.30) {$x_3$};
  \node[molecule] (Am1) at (3.52,3.55) {$m_1$};
  \node[molecule] (Am2) at (3.52,2.40) {$m_2$};
  \node[molecule] (Am3) at (3.52,1.25) {$m_3$};
  \draw[chosen] (Ax1)--(Am2);
  \draw[candidate] (Ax2)--(Am1);
  \draw[candidate] (Ax2)--(Am3);
  \draw[candidate] (Ax3)--(Am1);
  \node[font=\rmfamily\scriptsize,text=npgGreen!48!black]
    at (2.10,4.27) {$\Phi=\{(x_1,m_2)\}$};

  \node[entity] (Bx1) at (5.42,3.46) {$x_1$};
  \node[entity] (Bx2) at (5.42,2.38) {$x_2$};
  \node[entity] (Bx3) at (5.42,1.30) {$x_3$};
  \node[molecule] (Bm1) at (8.17,3.55) {$m_1$};
  \node[molecule] (Bm2) at (8.17,2.40) {$m_2$};
  \node[molecule] (Bm3) at (8.17,1.25) {$m_3$};
  \draw[chosen] (Bx1)--(Bm2);
  \draw[chosen] (Bx2)--(Bm3);
  \draw[chosen] (Bx3)--(Bm1);
  \node[font=\rmfamily\scriptsize,text=npgNavy]
    at (6.80,4.27) {choose $(x_2,m_3)$};
  \node[font=\rmfamily\scriptsize,text=npgGreen!48!black]
    at (6.80,0.43) {$\mathcal D_{x_3}\gets\{m_1\}$};

  \begin{scope}[xshift=1.2mm]
  \node[live] (root) at (11.42,4.02) {$\Phi$};
  \node[live] (left) at (10.35,2.68) {$\mathcal D_{x_3}=\{m_1\}$};
  \node[reject] (right) at (12.53,2.68) {$\mathcal D_{x_3}=\varnothing$};
  \node[solution,minimum width=27mm] (leaf) at (11.42,1.16)
    {complete $\varphi$\quad $\checkmark$};
  \draw[branch] (root)--(left)
    node[midway,above left=-1pt,font=\rmfamily\tiny,text=mut]
      {$x_2\mapsto m_3$};
  \draw[deadbranch] (root)--(right)
    node[midway,above right=-1pt,font=\rmfamily\tiny,text=mut]
      {$x_2\mapsto m_1$};
  \draw[branch] (left)--(leaf)
    node[midway,right,font=\rmfamily\tiny,text=mut]
      {$x_3\mapsto m_1$};
  \node[font=\rmfamily\scriptsize,text=npgRed!75!black]
    at (12.53,2.10) {prune};
  \end{scope}
\end{tikzpicture}}
  \caption{Backtrace-and-prune search.  (A) Solid green links denote
    assignments; dashed links denote remaining candidates.  (B) Assigning
    $x_2\mapsto m_3$ leaves $\mathcal D_{x_3}=\{m_1\}$.  (C) Propagation
    prunes the infeasible branch and completes the alternative assignment.}
  \label{fig:searchtree}
\end{figure}

\begin{algorithm}[t]
  \caption{Procedure $\textrm{Match}(\Phi,\overline{X},\mathcal{D})$}
  \label{alg:basic_recursion}
  \renewcommand{\algorithmicendif}{\vspace*{-\baselineskip}} 
  \begin{algorithmic}
    \IF{$\overline{X}=\emptyset$}
       \PRINT $\Phi$ and \textbf{terminate} 
    \ENDIF
    \FORALL{$(x^*,m^*)\in\mathcal{D}$}
       \IF{$\textrm{Compatible}(\Phi,(x^*,m^*))$} 
          \STATE $\mathcal{D}'\leftarrow\textrm{Update}(\mathcal{D},x^*,m^*)$
          \IF{$\mathcal{D}'\ne \emptyset$}
             \STATE $\textrm{Match}(\Phi\cup\{(x^*,m^*)\},
                   \overline{X}\setminus\{x^*\},
                   \mathcal{D}')$
          \ENDIF
       \ENDIF
    \ENDFOR
  \end{algorithmic} 
\end{algorithm}

Algorithm~\ref{alg:basic_recursion} outlines a generic Backtracing and
Pruning algorithm. In this recursive version, a call to
$\textrm{Match}(\emptyset,X,X\times M)$ either produces a valid assignment
$\Phi$ the first time it reaches $\overline{X}=\emptyset$, or it terminates
without output. It is not difficult to modify the algorithm to list all
valid assignments: it suffices to continue after yielding the first
assignment $\Phi$.

The practical performance of Algorithm~\ref{alg:basic_recursion} depends
crucially on the two problem specific functions $\textrm{Compatible}()$,
which checks whether an extension of $\Phi$ is admissible and the
$\textrm{Update}()$ of the candidate list $\mathcal{D}$ upon adding the
match $(x^*,m^*)$. The following sections are concerned with
  problem-specific methods to remove as many potential matches from
  $\mathcal{D}$ as possible.

\subsection{Compositional Constraints}
\label{compositional-constraints}

Denote the \emph{support} of a reaction $r\in R$ by
\begin{equation}
  \supp(r) \coloneqq \{ x\in X | s^-_{xr}+s^+_{xr}>0 \}  
\end{equation}
Correspondingly, we write $\supp(G)$ and $\supp(H)$ for the two complexes
$G$ and $H$ comprising a reaction $r=(G\longrightarrow H)$.  We say that a
reaction $r$ is \emph{assigned} (in $\Phi)$ if $x\in\supp(r)$ implies that
there is $(x,m_x)\in\Phi$, i.e., all abstract entities taking part in the
reaction are assigned a concrete molecule. Moreover, a reaction $r$ is
\emph{almost assigned} if there is $x^*\in\supp(r)$ and
$x\in\supp(r)\setminus\{x^*\}$ implies that there is $(x,m_x)\in\Phi$,
i.e., $x^*$ is the only entity appearing in $r$ that has not yet been
assigned. Finally, we say that a reaction is \emph{semi-assigned} if
exactly one of its complexes in assigned. Note that if $r$ is almost
assigned it is also semi-assigned unless the missing entity acts as a
catalyst. 

Compatibility in our setting amounts to checking, for each reaction $r\in
R$, whether $r$ (i) is balanced. This can only be done for completely
assigned reactions and amount to verifying
equ.(\ref{eq:balanced}). Clearly, it suffice in
$\textrm{Compatible}(\Phi,(x^*,m^*))$ to perform this test for completely
assigned reaction with $x^*\in\supp(r)$, since this test has been performed
in earlier steps for all completely assigned reaction with
$x^*\notin\supp(r)$. 

The update of $\mathcal{D}$ is straightforward in principle. In order to
ensure correctness, it suffices to remove all potential assignments
$(y,m')$ from $\mathcal{D}$ that contain $m'\in\overline{X}$, i.e., a
molecule that has already been assigned. A practically efficient algorithm,
however, is obtained only if $\mathcal{D}$ is pruned as aggressively as
possible.

To this end we observe that there is a natural partial order on
$\mathfrak{M}$ defined by the sub-multiset relation of the atomic
composition. More formally, we have   
\begin{equation}
  m\prec m' \iff \sigma(m)\ne\sigma(m') \text{ and } \sigma_a(m) \le
  \sigma_a(m') \text{ for all atom types } a,
\end{equation}
This ``composition subset partial order'' translates to the set of
unassigned entities $\overline{X}$ in a way that depends on the
assigned part of the reaction.
\begin{itemize}
\item Whenever a reaction $r$ is almost assigned, the composition
  $\sigma(y^*)$ of the remaining unassigned species $y^*$ can be computed
  unambiguously from equ.(\ref{eq:balanced}).  This limits
  $\mathcal{D}_{y^*}$ to pairs $(y^*,m')$ with $\sigma(m')=\sigma(y^*)$.
  As a useful byproduct we observe that this restriction of $\mathcal{D}$
  ensures that $\textrm{Compatible}(\Phi,(x^*,m))$ is always true and thus
  the compatibility check becomes trivial and can be omitted.
\item Whenever a reaction is semi-assigned but not almost assigned, there
  are at least two unassigned entities on one side. Assume, w.l.o.g.\
  that the reactants and some of the products are assigned, i.e.,
  $G \longrightarrow H' + \sum_{x\in\overline{X}} s^+_{xr} x$  
  Then we can infer $\sigma(x)\prec \sigma(G)-\sigma(H')$ for each
  $x\in \overline{X}\cap\supp(H)$.
\item For reactions $r$ that are not semi-assigned, we obtain constraints
  only if $r$ has a single reactant or a single product $x$ since
  $\sigma(y)\prec\sigma(x)$ for $y\in\supp(r)\setminus\{x\}$. 
\end{itemize}

These three cased yields the following rules to update $\mathcal{D}$ after
$(x^*,m^*)$ is appended to $\Phi$ and $x^*$ is removed from $\overline{X}$:
\begin{itemize}
\item[(1)] If $x^*\in r$ and $r$ is now almost assigned then compute
  the composition $\sigma(y^*)$ of the remaining unassigned species 
  $y^*$ and set $\mathcal{D}_{y^*}=\{(y^*,m)|\sigma(m)=\sigma(y^*)\}$.
  If $\mathcal{D}_{y^*}=\empty$ then the branch can be abandoned
  altogether, and we may set $\mathcal{D}=\emptyset$.
\item[(2)] If $x^*\in r$ and $r=(G\longrightarrow H)$ is now semiassigned
  but not almost assigned:
  If $G$ is assigned, compute $\tau\coloneqq \sigma(G)-\sigma(H')$ and
  if $H$ is assigned, compute $\tau\coloneqq \sigma(H)-\sigma(G')$.
  For all unassigned entities $x\in\supp(r)\cap\overline{X}$ a candidate 
  match $(x,m)$ is removed from $\mathcal{D}_x$ if $\tau\preceq\sigma(m)$
  since at least two unassigned entities must share the
  ``remainder sum formula'' $\tau$.
  Again we can abandon the branch if $\mathcal{D}_x$ for one of these
  vertices $x$, which case we set $\mathcal{D}=\emptyset$. 
\item[(3)] If $x^*\in r$ and $r$ is a reaction that is not semi-assigned
  with a single reactant $x$ or a single product $x$, then remove from
  $\mathcal{D}_x$ all $(x,m)$ with $\sigma(m)\preceq \sigma(m')$ for all
  assigned entities $y\in\supp(r)\setminus\overline{X}$ and
  $(y,m')\in\Phi$. Again $\mathcal{D}_x=\emptyset$ forces us to abandon the
  branch and thus $\mathcal{D}=\emptyset$.
\end{itemize} 
Algorithm~\ref{alg:updateD} summarizes these considerations.

\begin{algorithm}[t]
  \caption{Procedure $\textrm{Update}(\mathcal{D}',x^*,m^*)$}
  \label{alg:updateD}
  \begin{algorithmic}
    \STATE $\mathcal{D}\leftarrow\mathcal{D}'$
    \FORALL{$r\in R$ with $x^*\in\supp(r)$}
       \IF{$r$ is almost assigned}
          \STATE $\{y^*\}=\supp(r)\cap\overline{X}$
          \STATE $\mathcal{D}_{y^*}=\{(y^*,m)|\sigma(m)=\sigma(y^*)\}$
          \IF{ $\mathcal{D}_{y^*}=\emptyset$ }       
              \RETURN $\emptyset$ 
          \ENDIF       
       \ELSIF{$r$ is semi-assigned}
          \STATE compute $\tau$
          \FORALL{$x\in\supp(r)\cap\overline{X}$}
             \FORALL{$(x,m)\in\mathcal{D}_x$ with $\tau\preceq\sigma(m)$}
                \STATE $\mathcal{D}_x \leftarrow \mathcal{D}_x \setminus
                   \{(x,m)\}$
             \ENDFOR
             \IF{ $\mathcal{D}_{x}=\emptyset$ }       
                 \RETURN $\emptyset$ 
             \ENDIF
          \ENDFOR
       \ELSIF{$\supp(G)=\{z\}$ or $\supp(H)=\{z\}$} 
          \FORALL{$y\in \supp(r)\setminus\overline{X}$}
             \FORALL{$(y,m')\in\Phi$ with $\sigma(m)\preceq \sigma(m')$} 
                \STATE $\mathcal{D}_z \leftarrow \mathcal{D}_z \setminus
                    \{(z,m)\}$
             \ENDFOR
             \IF{ $\mathcal{D}_{x}=\emptyset$ }       
                \RETURN $\emptyset$ 
             \ENDIF                   
          \ENDFOR
       \ENDIF      
     \ENDFOR
     \RETURN $\mathcal{D}$  
  \end{algorithmic} 
\end{algorithm}

A potential disadvantage of working with the partial order on sum formulas
is a comparatively high cost of evaluation. It is possible, however, to
precompute the partial order for $M$ and to store the corresponding DAG.

\subsection{Mass Intervals} 

A computationally easier but somewhat less stringent approach is to
consider only the molecular masses rather than the full composition
vectors. To this end we consider the mass function
$\mu:\mathfrak{M}\to\mathbb{R}^+$ that assigns a mass $\mu(m)=\sum_a
\sigma_a(m)\mu(a)$ to each molecule, where $\mu(a)$ is the mass of the atom
or moiety $a$. For a valid assignment $\varphi:X\to M$ of course
we must have
\begin{equation}
  \sum_{x\in X} s^-_{xr}\mu(\varphi(x))
  =
  \sum_{x\in X} s^+_{xr}\mu(\varphi(x)).
\end{equation}
Then one can maintain for each entity $x\in X$ an interval $I(x)=[a_x,b_x]$,
where $a_x$ is the minimal and $b_x$ is the maximal mass of a molecule that
can still be assigned to $x$. The mass intervals immediately translate to
mass intervals for a complex $G$:
\begin{equation} 
  I(G)\coloneqq \left[
    \sum_{x\in\supp(G)}s_{xG}a_x, \sum_{x\in\supp(G)}s_{xG}b_x \right].
\end{equation}
It becomes impossible to find a valid assignment whenever there is a
reaction $G\longrightarrow H$ such that $I(G)\cap I(H)=\emptyset$.

\begin{figure}[htbp]
  \centering
  \resizebox{\linewidth}{!}{
\begin{tikzpicture}[
    font=\small,
    reaction/.style={anchor=west},
    full interval/.style={line width=1.7mm,line cap=round,draw=gray!35},
    pruned interval/.style={line width=1.7mm,line cap=round,draw=blue!70!black},
    assigned mass/.style={line width=0.25mm,draw=cyan!80!blue},
    tick/.style={line width=0.25mm,draw=black!65},
    poolchip/.style={rounded corners=2pt,draw=gray!48,fill=gray!5,
      line width=0.55pt,minimum width=16.6mm,minimum height=5.2mm},
    poolsel/.style={rounded corners=2pt,draw=blue!70!black,
      fill=blue!6,line width=0.8pt,minimum width=16.6mm,
      minimum height=5.2mm},
    massdot/.style={circle,fill=gray!62,text=white,minimum size=4.2mm,
      inner sep=0pt,font=\rmfamily\bfseries\tiny},
    massdotsel/.style={circle,fill=blue!70!black,text=white,
      minimum size=4.2mm,inner sep=0pt,font=\rmfamily\bfseries\tiny},
    fixedmass/.style={circle,fill=blue!70!black,minimum size=1.7mm,
      inner sep=0pt}
  ]
    \draw[spanel] (0.00,-1.42) rectangle (4.45,4.35);
    \draw[spanel] (4.61,-1.42) rectangle (13.60,4.35);
    \node[badge] at (0.36,4.01) {A};
    \node[badge] at (4.97,4.01) {B};
    \node[ptitle] at (0.68,4.01) {Reaction chain};
    \node[ptitle] at (5.29,4.01) {Candidate mass intervals};

    \begin{scope}[xshift=0.48cm]
    \newcommand{\mx}[1]{5.00 + #1*0.095}

    \node[reaction] at (0,3.55) {$r_1:\quad A+B\longrightarrow C$};
    \node[reaction] at (0,2.95) {$r_2:\quad C+D\longrightarrow E$};
    \node[reaction] at (0,2.35) {$r_3:\quad E+F\longrightarrow G$};
    \node[ptitle] at (0.20,1.72) {Candidate molecules};

    \newcommand{\masscard}[6]{%
      \node[#1] at (#2,#3) {};
      \node[anchor=west,font=\rmfamily\tiny,text=ink]
        at ({#2-0.75},#3) {#4};
      \node[#5] at ({#2+0.62},#3) {#6};
    }
    \masscard{poolchip}{0.80}{1.15}{$\mathrm{H_2}$}{massdot}{2}
    \masscard{poolchip}{2.66}{1.15}{$\mathrm{H_2O}$}{massdot}{18}
    \masscard{poolsel}{0.80}{0.45}{$A{:}\,\mathrm{C_2H_2}$}{massdotsel}{26}
    \masscard{poolchip}{2.66}{0.45}{$\mathrm{CHN}$}{massdot}{27}
    \masscard{poolchip}{0.80}{-0.25}{$\mathrm{C_2H_4}$}{massdot}{28}
    \masscard{poolchip}{2.66}{-0.25}{$\mathrm{C_3H_5N}$}{massdot}{55}
    \masscard{poolchip}{1.73}{-0.95}{$\mathrm{C_3H_7NO}$}{massdot}{73}

    \begin{scope}[yshift=-2mm]
    \draw[->,tick] ({\mx{0}},3.35) -- ({\mx{78}},3.35)
      node[right,font=\scriptsize] {$\mu$};
    \foreach \m in {0,20,40,60}
      \draw[tick] ({\mx{\m}},3.27) -- ({\mx{\m}},3.43)
        node[above,font=\scriptsize] {$\m$};

    \draw[assigned mass] ({\mx{26}},3.27) -- ({\mx{26}},3.43)
      node[above=0.5pt,font=\rmfamily\scriptsize,
        text=cyan!70!blue] {$A$};

    \foreach \name/\lo/\hi/\y in {
      B/2/28/2.63,
      C/28/55/2.07,
      D/2/28/1.51,
      F/2/18/0.39} {
      \node[anchor=east] at (4.80,\y) {$\name$};
      \draw[full interval] ({\mx{2}},\y) -- ({\mx{73}},\y);
      \draw[pruned interval] ({\mx{\lo}},\y) --
        node[midway,above=0.5pt,font=\tiny] {$[\lo,\hi]$}
        ({\mx{\hi}},\y);
    }

    \foreach \name/\mass/\y in {
      E/55/0.95,
      G/73/-0.17} {
      \node[anchor=east] at (4.80,\y) {$\name$};
      \draw[full interval] ({\mx{2}},\y) -- ({\mx{73}},\y);
      \node[fixedmass] at ({\mx{\mass}},\y) {};
      \node[font=\tiny,above=0.5pt] at ({\mx{\mass}},\y)
        {$[\mass,\mass]$};
    }
    \end{scope}

    \draw[full interval] (5.00,-1.02) -- +(0.75,0)
      node[right=0.2cm,font=\rmfamily\footnotesize] {initial domain};
    \draw[pruned interval] (8.65,-1.02) -- +(0.75,0)
      node[right=0.2cm,font=\rmfamily\footnotesize]
      {surviving candidates};
    \end{scope}
  \end{tikzpicture}}
  \caption{Mass-interval pruning.  (A) A partial assignment fixes
  $\mu(A)=26$.  (B) Balance propagation contracts the
  mass intervals of $B$--$G$; candidates outside the blue intervals are
  removed.}
  \label{fig:massprune}
\end{figure}
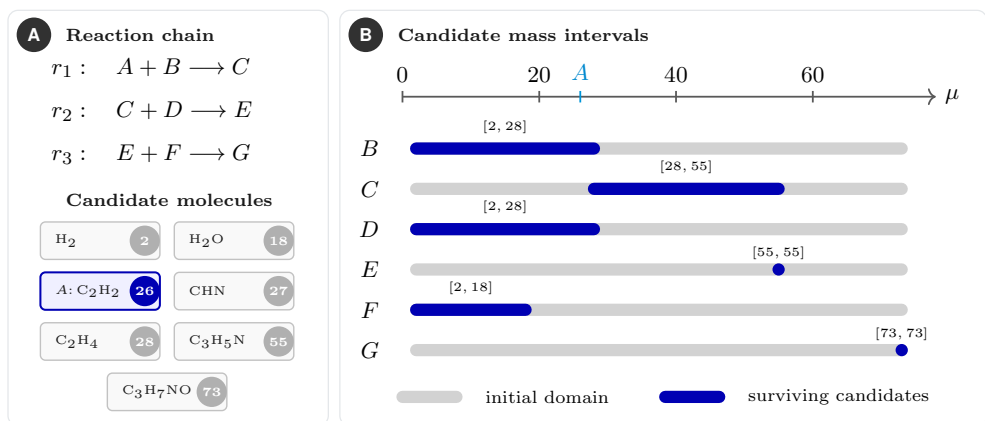

Partial assignments in a reaction $r$ will lead to bounds on the mass
intervals for each of the unassigned entities.  Consider $G\longrightarrow
H$ and assume $x\in\supp(G)\cap\overline{X}$. Then the lower bound $a_x$ is
obtained by substituting for $\mu_x$ the mass the molecules $\varphi(x)$
for each assigned entity $x\in\supp(r)\setminus\overline{X}$, the maximum
admissible mass $b_y$ for each $y\in\supp(G)\cap\overline{X}\setminus\{x\}$
in the same complex, and the minimum mass $a_y$ for unassigned entities
$y\in\supp(H)\cap\overline{X}$ and then computing $a_x$ as $\mu_x$ from the
mass balance equation $\sum_{x} s^-_{xr}\mu_x=\sum_{x}
s^+_{xr}\mu_x$. Similarly, the upper bound $b_x$ is obtained by
substituting the minimum mass $b_y$ for all unassigned entities
$y\in\supp(H)\cap\overline{X}$ and the maximum mass $a_y$ for unassigned
entities $y\in\supp(H)\cap\overline{X}$. Analogous expressions yield the
interval for $x\in\supp(H)\cap\overline{X}$.  Moreover, whenever $a_x$ is
lower than the lightest unassigned molecule the $a_x$ is set to this smallest
available value. Similarly, if $b_x$ is larger than the heaviest unassigned
molecule, $b_x$ is set to this maximal available value.
Fig.~\ref{fig:massprune} illustrates the principle.

When a new assignment $(x^*,m^*)$ is fixed, all reactions with
$x^*\in\supp(r)$ need to be updated, since the interval $I(x^*)$ has collapsed
to a single value, the mass of $m^*$. All entities in
$\supp(r)\cap\overline{X}$ whose mass intervals change are added to the
\emph{frontier} $\mathcal{F}$. If $\mathcal{F}\ne\emptyset$, we remove an
entity $z$ from $\mathcal{F}$ and update the mass intervals of all
unassigned entities $y\in\supp(r)\cap\overline{X}$ in all reactions with
$z\in\supp(r)$. Again all unassigned entities whose intervals are affected
in these update step are added to $\mathcal{F}$. Updating of intervals
continues until $\mathcal{F}=\emptyset$. In this process an entity can
be added to $\mathcal{F}$ more than once. As part of the update procedure,
we check whether $I(G)\cap I(H)=\emptyset$ for each $r=(G\longrightarrow
H)$, and abandon the branch if this is the case.

The updated set $\mathcal{D}$ then comprises the candidate assignments
$(x,m)$ where $x\in\overline{X}$ and $m$ is an as yet unassigned molecule
with mass $\mu(m)\in I(x)$. Note that if the reaction
$r$ is almost assigned then its mass interval is also reduced to a single
value. The set $\mathcal{D}$ could now be pruned further by using
compositional constraints of the previous section.

\subsection{Conservation Laws}
\label{sec:conservation-laws}

The stoichiometric matrix of a conservative reaction networks harbors
additional structural information in its left null space. In particular, it
identifies \emph{obligatory isomers} and \emph{obligatory polymers}.  To
this end we rephrase \cite[Thm.37]{Mueller:22a} as follows:
\begin{proposition}
  Let $\{u^{(1)},\dots,u^{(k)}\}$ be a basis of the left nullspace of
  $\mathbf{S}$, i.e., $u^{(j)}\mathbf{S}=0$ for $1\le j\le k$, and let
  $\varphi:X\to M$ be a balanced assignment. If there is
  $c_1,c_2\in\mathbf{N}$ such that $c_1 u_x^{(j)}=c_2 u_y^{(j)}$ for
  $1\le j\le k$ then $c_1 \sigma(\varphi(x))=c_2\sigma(\varphi(y))$.
\end{proposition}

For obligatory isomers and oligomers we obtain in particular also
additional constraints on their mass intervals: upper and lower bounds are
linked by the same constants, and one member of such a family
has been assigned, the composition and thus also the mass intervals of all
other members collapse to a unique value.

\subsection{Candidate Order}

Much of the improvements in VF2-style algorithms \cite{Cordella:04} derives
from a suitable order of the candidate matches, see
e.g.\ \cite{Juettner:18}.

\paragraph{Entity ordering.}
A natural first principle is to assign highly constrained entities
early. This is captured by the minimum-remaining-values principle from
constraint satisfaction, which gives priority to unassigned entities with
small current domains.  Such an ordering tends to expose contradictions
early and thereby limits the exploration of unproductive branches.  It is
useful to break ties by prioritizing an entity incident to a reaction with
the fewest currently unassigned entities. Assigning such an entity brings
that reaction closest to being almost assigned. It is
obvious, therefore, to prioritize entities $x^*$ for which there is an
almost assigned reaction $r$. In particular, in this case the composition
$\sigma(x^*)$ and the corresponding mass is already known. Similarly, it is
useful to assign entities early that are involved in isomerization
reactions of the form $A\longrightarrow B$, since the result is an almost
assigned reaction.

\paragraph{Value ordering.}
Following the selection of an abstract entity $x$, the order in which
molecules from $\mathcal{D_x}$ are considered may also play a role. A
robust approach is the least-constraining-value heuristic, commonly applied
in the real of constrain satisfaction problems. If a candidate molecule is
an option for many other unassigned entities, assigning it to $x$
eliminates a shared resource that might be strictly required elsewhere.
Conversely, a molecule that occurs in few other domains is more specific to
$x$ and is less likely to restrict the remaining assignment.
To formalize this logic, we define a score $A_x(m)$ for each molecule
$m \in \mathcal{D}_x$ that number of unassigned entities whose 
domains also contains $m$, i.e., 
\begin{equation} 
  A_x(m) \coloneqq
  \bigl|\{\,y \in \overline{X}\setminus\{x\} : m \in D_y\,\}\bigr|.
\end{equation} 
Guided by this metric, the candidate molecules for $x$ are systematically 
considered in increasing order of $A_x(m)$ (see Fig.~\ref{fig:candidate-order}).

\begin{figure}[htbp]
  \centering
  \resizebox{\linewidth}{!}{\begin{tikzpicture}[
  x=1cm,y=1cm,font=\rmfamily\footnotesize,
  line cap=round,line join=round,
  species/.style={circle,draw=npgNavy!78,fill=npgBlue!14,
    line width=0.7pt,minimum size=5.8mm,inner sep=0pt,
    font=\rmfamily\bfseries\scriptsize,text=npgNavy},
  speciespick/.style={species,draw=npgGreen!70!black,fill=npgTeal!25,
    line width=1.0pt,text=npgGreen!42!black},
  reaction/.style={rounded corners=1.3pt,draw=npgRed!75!black,
    fill=npgSalmon!20,line width=0.65pt,minimum width=6.5mm,
    minimum height=4.8mm,inner sep=1pt,
    font=\rmfamily\bfseries\tiny,text=npgRed!72!black},
  count/.style={circle,draw=white,fill=ink!68,line width=0.5pt,
    minimum size=4.4mm,inner sep=0pt,font=\rmfamily\bfseries\tiny,
    text=white},
  countpick/.style={count,fill=npgGreen!64!black},
  candidate/.style={rounded corners=3pt,draw=npgNavy!55,
    fill=white,line width=0.65pt,minimum width=22mm,
    minimum height=17mm,font=\rmfamily\scriptsize,text=npgNavy},
  candidatepick/.style={candidate,draw=npgGreen!70!black,
    fill=npgTeal!13,line width=0.9pt,text=npgGreen!42!black},
  molecule/.style={circle,draw=npgNavy!75,fill=npgBlue!18,
    line width=0.7pt,minimum size=6.5mm,inner sep=0pt,
    font=\rmfamily\bfseries\scriptsize,text=npgNavy},
  sharednode/.style={circle,draw=ink!38,fill=white,line width=0.55pt,
    minimum size=4.8mm,inner sep=0pt,font=\rmfamily\tiny,text=mut},
  inedge/.style={-{Stealth[length=1.6mm,width=1.2mm]},
    draw=npgRed!78!black,line width=0.75pt,
    shorten >=1.4pt,shorten <=1.4pt},
  outedge/.style={-{Stealth[length=1.6mm,width=1.2mm]},
    draw=npgGreen!65!black,line width=0.8pt,
    shorten >=1.4pt,shorten <=1.4pt},
  overlap/.style={draw=npgGrayBlue!72,line width=0.65pt,
    dash pattern=on 1.8pt off 1.4pt},
  choose/.style={-{Stealth[length=1.9mm,width=1.45mm]},
    draw=npgGreen!68!black,line width=1.0pt}
]
  \draw[spanel] (0.00,0.00) rectangle (7.90,4.30);
  \draw[spanel] (8.05,0.00) rectangle (13.60,4.30);
  \node[badge] at (0.36,3.96) {A};
  \node[badge] at (8.41,3.96) {B};

  \node[species] (A) at (0.72,2.90) {$A$};
  \node[speciespick] (B) at (0.72,1.67) {$B$};
  \node[reaction] (r1) at (2.02,2.29) {$r_1$};
  \node[species] (C) at (3.10,2.29) {$C$};
  \node[reaction] (r2) at (4.18,2.29) {$r_2$};
  \node[species] (E) at (5.25,2.29) {$E$};
  \node[reaction] (r3) at (6.34,2.29) {$r_3$};
  \node[species] (G) at (7.20,3.06) {$G$};
  \node[species] (D) at (5.25,0.84) {$D$};
  \node[species] (F) at (6.34,3.48) {$F$};

  \draw[inedge] (A)--(r1);
  \draw[inedge] (B)--(r1);
  \draw[outedge] (r1)--(C);
  \draw[inedge] (C)--(r2);
  \draw[inedge] (D)--(r2);
  \draw[outedge] (r2)--(E);
  \draw[inedge] (E)--(r3);
  \draw[inedge] (F)--(r3);
  \draw[outedge] (r3)--(G);
  \draw[outedge] (r3) to[out=-100,in=0] (D);

  \node[count] at (0.40,3.25) {4};
  \node[countpick] at (0.40,1.32) {2};
  \node[count] at (3.10,2.76) {3};
  \node[count] at (5.25,0.38) {3};
  \node[countpick] at (5.25,2.76) {2};
  \node[count] at (6.34,3.91) {4};
  \node[count] at (7.48,3.40) {3};

  \draw[rounded corners=7pt,draw=npgGreen!62!black,
    line width=0.75pt,dash pattern=on 2.2pt off 1.6pt]
    (0.32,1.26) rectangle (1.12,2.08);
  \node[rounded corners=2.4pt,draw=npgGreen!52,fill=npgGreen!7,
    line width=0.55pt,font=\rmfamily\scriptsize,
    text=npgGreen!45!black,inner sep=2.5pt] at (2.22,0.43)
    {$|\mathcal D_B|=|\mathcal D_E|=2
      \ \Rightarrow\ B\ \checkmark$};
  \draw[choose] (1.15,0.66) to[out=135,in=-75] (0.76,1.24);

  \node[candidatepick] (card4) at (9.53,2.06) {};
  \node[molecule] (m4) at (9.53,2.34) {$m_4$};
  \node[font=\rmfamily\bfseries\scriptsize,text=npgGreen!45!black]
    at (9.53,1.52) {$A_B(m_4)=1$};
  \node[sharednode] (u1) at (9.53,3.58) {$E$};
  \draw[overlap] (m4)--(u1);

  \node[candidate] (card2) at (12.13,2.06) {};
  \node[molecule] (m2) at (12.13,2.34) {$m_2$};
  \node[font=\rmfamily\bfseries\scriptsize,text=npgNavy]
    at (12.13,1.52) {$A_B(m_2)=3$};
  \node[sharednode] (v1) at (11.17,3.36) {$A$};
  \node[sharednode] (v2) at (12.13,3.58) {$C$};
  \node[sharednode] (v3) at (13.09,3.36) {$D$};
  \draw[overlap] (m2)--(v1);
  \draw[overlap] (m2)--(v2);
  \draw[overlap] (m2)--(v3);

  \node[font=\rmfamily\bfseries\scriptsize,text=npgGreen!45!black] (try)
    at (9.53,0.52) {try first};
  \draw[choose] (try.north) -- (card4.south);
\end{tikzpicture}}
  \caption{Network-aware candidate ordering.
    (A) Both $B$ and $E$ have domain size two. Assigning $B$ first leaves
    only two unassigned entities in the three-entity reaction $r_1$,
    whereas $r_3$ contains four entities.
    (B) Candidate $m_4$ is tried before the more widely shared
    $m_2$.}
  \label{fig:candidate-order}
\end{figure}
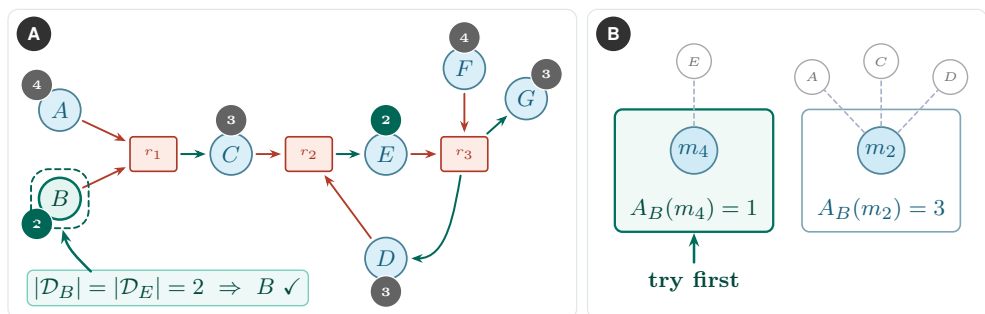

\subsection{Symmetries of $(X,R)$}

The reaction network $(X,R)$ is faithfully represented by its K{\"o}nig
graph $\king(X,R)$. Two networks $(X,R)$ and $(X',R')$ are considered
isomorphic if there exists a bijective mapping
$\vartheta=(\vartheta_X,\vartheta_R)$, with $\vartheta_X: X\to X'$ and
$\vartheta_R:R\to R'$, satisfying
\begin{equation}
  \label{eq:XRiso} 
  s^-_{\vartheta_X(x)\vartheta_R(r)}=s^-_{xr}
  \quad\text{and}\quad 
  s^+_{\vartheta_X(x)\vartheta_R(r)}=s^+_{xr}
\end{equation}
for all $x\in X$ and $r\in R$.  Correspondingly, an automorphism of $(X,R)$
is an isomorphism from the network onto itself, inducing a bijection
$\vartheta:X\to X$. Clearly, if $\varphi:X\to M$ is a balanced bijection,
then $\varphi\circ\vartheta$ is also a balanced bijection.  These
automorphisms form a group $\Aut(X,R)$, which is generally a subgroup of
$\Aut(\king(X,R))$. Any difference between these groups essentially arises
from automorphisms of $\king(X,R)$ that fix $X$, which occur only if there
are different reactions with the same stoichiometry, i.e., duplicate
columns in the stoichiometric matrix. If no such cases exist, any
automorphism of $(X,R)$ extends uniquely to a mapping on the reaction
vertices, and the two groups coincide.

For the performance of the algorithms outline above, it is of immediate
interest to remove parts of the search tree that derive from
automorphisms. Several techniques of symmetry breaking which are routinely
used in constraint programming \cite{Crawford:96,Fahle:01,Gent:06} can also
be utilized here.

To this end, we consider the orbits
$\orb(x)\coloneqq\{y\in X | \exists \vartheta\in\Aut(X,R):
y=\vartheta(x)\}$. We may endow both the set of abstract entities $X$ and
the set of molecules $M$ with arbitrary but fixed orders $<$.  Considering
the restriction of $\varphi$ to an orbit $\orb(x)$ of $\Aut(X,R)$, we
observe that any bijection $\orb(x)\to\varphi(\orb(x))$ is a valid
assignment. We can thus restrict $\varphi$ to be order-preserving on all
orbits, which introduces the additional constraint:
\begin{equation}
  x<y \implies \varphi(x)<\varphi(y) \quad\text{for all } y\in\orb(x) 
\end{equation}

Symmetry can also be used after a partial assignment has been shown to be
impossible. A \emph{no-good} list records a partial assignment pattern that
cannot be extended to a valid realization. Clearly, whenever a pattern
fails, every symmetric copy of the pattern must also fail.  Symmetry breaking 
during search (SBDS) uses this observation by recording not only the failed 
assignment pattern itself, but also all relevant symmetric variants. Recording
such forbidden patterns prevents the same failure from being rediscovered
under a different labeling. In practice, it suffices to store minimal
forbidden patterns, i.e., failed assignments at abandoned vertices closest
to the root of the search tree. 

\begin{figure}[htbp]
  \centering
  \resizebox{\linewidth}{!}{
\begin{tikzpicture}[
  x=1cm,y=1cm,font=\rmfamily\footnotesize,
  line cap=round,line join=round,
  species/.style={circle,draw=npgBlue!85!black,fill=npgBlue!17,
    line width=0.75pt,minimum size=6mm,inner sep=0pt,
    font=\rmfamily\bfseries\scriptsize,text=npgNavy},
  reaction/.style={rounded corners=1.4pt,draw=npgRed!82!black,
    fill=npgSalmon!24,line width=0.75pt,minimum width=7mm,
    minimum height=5.3mm,inner sep=1pt,
    font=\rmfamily\bfseries\scriptsize,text=npgRed!72!black},
  molecule/.style={rounded corners=2.5pt,draw=npgGreen!65!black,
    fill=npgTeal!15,line width=0.75pt,minimum width=9mm,
    minimum height=6.5mm,font=\rmfamily\scriptsize,
    text=npgGreen!45!black},
  state/.style={rounded corners=2.5pt,draw=npgNavy!60,
    fill=npgBlue!7,line width=0.65pt,minimum width=13mm,
    minimum height=7mm,font=\rmfamily\scriptsize,text=npgNavy},
  canon/.style={rounded corners=2.8pt,draw=npgGreen!68!black,
    fill=npgTeal!14,line width=0.8pt,minimum width=34mm,
    minimum height=8mm,font=\rmfamily\scriptsize,
    text=npgGreen!45!black},
  inedge/.style={-{Stealth[length=1.8mm,width=1.35mm]},
    draw=npgRed!85!black,line width=0.95pt,
    shorten >=1.8pt,shorten <=1.8pt},
  outedge/.style={-{Stealth[length=1.8mm,width=1.35mm]},
    draw=npgGreen!70!black,line width=1.0pt,
    shorten >=1.8pt,shorten <=1.8pt},
  mapedge/.style={-{Stealth[length=1.8mm,width=1.35mm]},
    draw=npgNavy!72,line width=0.85pt,
    shorten >=1.5pt,shorten <=1.5pt},
  merge/.style={-{Stealth[length=1.8mm,width=1.35mm]},
    draw=npgGreen!62!black,line width=0.85pt,
    shorten >=1pt,shorten <=1pt},
  note/.style={font=\rmfamily\scriptsize,text=mut},
  edgelab/.style={font=\rmfamily\tiny,text=mut,inner sep=1pt,
    fill=white,rounded corners=1pt}
]
  \draw[spanel] (0.00,3.08) rectangle (13.60,6.85);
  \node[badge] at (0.36,6.51) {A};
  \node[ptitle] at (0.68,6.51)
    {{\boldmath Orbit under \(\Aut(X,R)\)}};

  \begin{scope}[yshift=-1mm]

  \draw[rounded corners=9pt,draw=npgBlue!70!black,fill=npgBlue!7,
    line width=0.85pt,dash pattern=on 2.6pt off 1.8pt]
    (0.72,4.02) rectangle (1.89,6.17);
  \node[font=\rmfamily\bfseries\scriptsize,text=npgNavy,anchor=west]
    at (2.06,6.10) {$\orb(A)=\orb(B)=\{A,B\}$};

  \node[species]  (A)  at ( 1.30,5.65) {$A$};
  \node[species]  (B)  at ( 1.30,4.54) {$B$};
  \node[reaction] (r1) at ( 3.10,5.10) {$r_1$};
  \node[species]  (C)  at ( 5.10,5.10) {$C$};
  \node[reaction] (r2) at ( 7.10,5.10) {$r_2$};
  \node[species]  (E)  at ( 9.10,5.10) {$E$};
  \node[reaction] (r3) at (11.10,5.10) {$r_3$};
  \node[species]  (G)  at (12.90,5.10) {$G$};
  \node[species]  (D)  at ( 9.10,3.70) {$D$};
  \node[species]  (F)  at (10.10,6.05) {$F$};

  \draw[inedge]  (A)--(r1);
  \draw[inedge]  (B)--(r1);
  \draw[outedge] (r1)--(C);
  \draw[inedge]  (C)--(r2);
  \draw[inedge]  (D)--(r2);
  \draw[outedge] (r2)--(E);
  \draw[inedge]  (E)--(r3);
  \draw[inedge]  (F)--(r3);
  \draw[outedge] (r3)--(G);
  \draw[outedge] (r3) to[out=-90,in=0] (D);

  \begin{scope}[shift={(2.06,3.42)}]
    \draw[inedge]  (0.00,0)--(0.65,0);
    \node[note,anchor=west] at (0.76,0) {substrate};
    \draw[outedge] (2.35,0)--(3.0,0);
    \node[note,anchor=west] at (3.11,0) {product};
  \end{scope}
  \end{scope}

  \draw[spanel] (0.00,-0.10) rectangle (13.60,2.78);
  \node[badge] at (0.36,2.44) {B};
  \node[ptitle] at (0.68,2.44)
    {{\boldmath Fixed order \(A<B\)}};

  \node[species]  (BA)  at (1.3,1.79) {$A$};
  \node[species]  (BB)  at (1.3,0.79) {$B$};
  \node[molecule] (Bm1) at (3.0,1.79) {$m_1$};
  \node[molecule] (Bm2) at (3.0,0.79) {$m_2$};

  \draw[{Stealth[length=1.8mm,width=1.35mm]}-{Stealth[length=1.8mm,width=1.35mm]},
    draw=npgBlue!78!black,line width=0.8pt,
    shorten >=1.5pt,shorten <=1.5pt]
    (BA.west) to[out=180,in=180,looseness=1.25] (BB.west);

  \draw[mapedge] (BA)--(Bm2);
  \draw[mapedge] (BB)--(Bm1);

  \node[state,minimum width=18mm,minimum height=9mm,align=left]
    (psi) at (5.05,1.93)
    {\(A\mapsto m_2\)\\[1pt]\(B\mapsto m_1\)};
  \node[canon,minimum width=18mm,minimum height=9mm,align=left]
    (psit) at (5.05,0.69)
    {\(A\mapsto m_1\)\\[1pt]\(B\mapsto m_2\)};

  \draw[mapedge] (Bm1.east) to[out=0,in=180] (psi.west);
  \draw[mapedge] (Bm2.east) to[out=0,in=180] (psit.west);

  \node[canon,minimum width=34mm] (key) at (8.17,1.31)
    {\(\kappa=[A\!:\!m_1,\ B\!:\!m_2]\)};
  \draw[merge] (psi.east) to[out=0,in=150] (key.west);
  \draw[merge] (psit.east) to[out=0,in=210] (key.west);

  \node[canon,minimum width=28mm,minimum height=9mm,align=center]
    (first) at (11.82,1.85)
    {new \(\kappa: (m_1,m_2)\)\\[1pt]expand \(\checkmark\)};
  \node[rounded corners=2.8pt,draw=npgRed!70,fill=npgSalmon!12,
    line width=0.75pt,minimum width=28mm,minimum height=9mm,
    font=\rmfamily\scriptsize,text=npgRed!75!black,align=center] (dup)
    at (11.82,0.55)
    {same \(\kappa: (m_2,m_1)\)\\[1pt]prune \(\times\)};

  \draw[merge] (key.east) to[out=0,in=180] (first.west);
  \draw[-{Stealth[length=1.8mm,width=1.35mm]},draw=npgRed!72,
    line width=0.8pt,dash pattern=on 2pt off 1.5pt,
    shorten >=1pt,shorten <=1pt]
    (key.east) to[out=0,in=180] (dup.west);
\end{tikzpicture}}
  \caption{Symmetry pruning.
    (A) The automorphism $\vartheta=(A\,B)$ exchanges equivalent
    species.
    (B) With \(A<B\) and \(m_1<m_2\), the tuple
    \((m_1,m_2)\) is lexicographically smaller than \((m_2,m_1)\).
    The assignments
    \(A\mapsto m_2,\ B\mapsto m_1\) and
    \(A\mapsto m_1,\ B\mapsto m_2\) share the canonical key
    \(\kappa=[A\!:\!m_1,\ B\!:\!m_2]\).
    The lex-minimal representative is expanded; the lex-larger symmetric
    assignment and its associated no-goods are pruned.}
  \label{fig:symmetry-prune}
\end{figure}
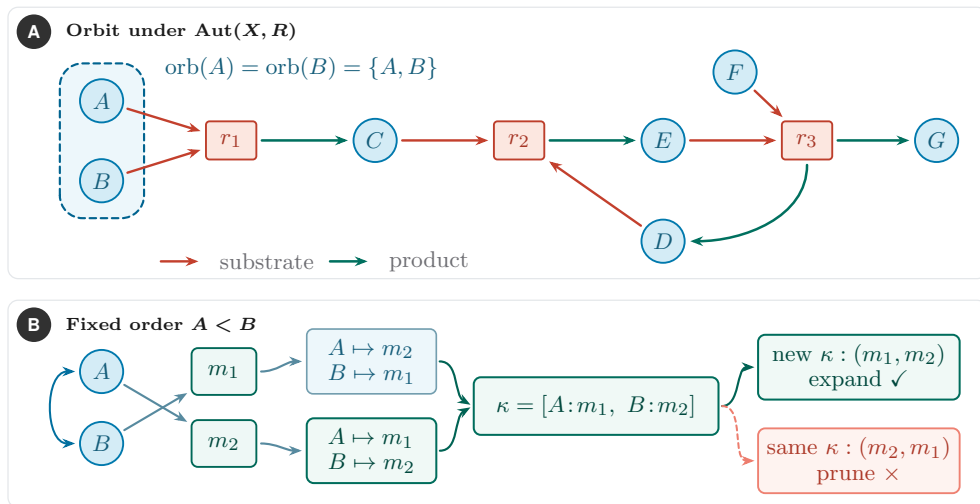

Symmetry breaking by dominance detection (SBDD) provides an alternative
inroad. Here, partial assignments equivalent under the action of
automorphism groups are identified by canonical state keys, computed from
the orbits and the fixed orders on $X$ and $M$ introduced above.  Let
$\psi:X'\to M$ with $X'\subseteq X$ be the partial assignment at a vertex
of the search tree. For each orbit $\omega=\orb(x)$, write
$\omega\cap X'=\{x_1<\cdots<x_k\}$ and order the assigned molecules as
$\{\psi(x_1),\ldots,\psi(x_k)\}=\{m_1<\cdots<m_k\}$. The canonical
representative $\psi^\circ$ is then defined by $\psi^\circ(x_i)=m_i$ for
$i=1,\ldots,k$, and agrees with $\psi$ outside these orbit-wise
rearrangements. In other words, each assigned orbit fragment is replaced by
its order-preserving representative.  The canonical key of the search
vertex is then
\begin{equation}
  \kappa(\psi)
  \coloneqq
  \operatorname{sort}
  \{(x,\psi^\circ(x))\mid x\in X'\}.
\end{equation}
Assignments that differ only by a permutation within an orbit therefore
receive the same canonical key. Once all extensions of a search vertex with
key $\kappa(\psi)$ have been explored, any later partial assignment $\psi'$
with $\kappa(\psi')=\kappa(\psi)$ can be discarded, since every extension
of $\psi'$ is equivalent to an extension of the already explored canonical
representative (see Fig.~\ref{fig:symmetry-prune}).

\subsection{Enforcing Reaction Mechanisms}

Recall from Sect.~\ref{sect:rm} that reaction types are determined by
partial ITS graphs, or -- equivalently -- by a triple
$(\Gamma^-,\Gamma^+,\hat\alpha)$ where $\Gamma^-$ and $\Gamma^+$ are
pattern graphs and $\hat\alpha:V(\Gamma^-)\to V(\Gamma^+)$ is a bijection
on their vertex sets. The reaction is encoded by the changes in the edge
set, i.e., the chemical bonds. From this representation, we can also derive
a \emph{context graph} $K$ consisting of all vertices as well as the edges
that do not change between $\Gamma^-$ and the $\Gamma^+$. The context graph
thus is by construction a spanning subgraph of both $\Gamma^-$ and
$\Gamma^+$, and hence there are subgraph morphisms $l:K\to\Gamma^-$ and
$r:K\to\Gamma^+$ that are bijections on the vertices. We refer to
\cite{GonzalezLaffitte:24b,Beier:26a} for a more detailed and more formal
discussion of this equivalence. Reinterpreting $\Gamma^-$, $\Gamma^+$, and
$K$ as simple graphs with edge labels encoding the multiplicities of edges,
$\Gamma^-\leftarrow K\rightarrow \Gamma^+$, constitutes as graph rewriting
rule in the double-pushout (DPO) formalism as implemented in
\texttt{M{\O}D} \cite{Andersen:16a}. The problem reduces to determining
whether there is a subgraph monomorphism $m:\Gamma^-\to G$ such that
application of the rule $p$ recovers $H$.

Whenever an abstract reaction $r\in R$ is assigned in some node of the
search tree, we therefore check if the assigned reaction
equ.(\ref{eq:assigned-r}) has an explanation by a reaction mechanism in
$\mathfrak{R}$. If this is not the case, the assignment is rejected by
setting $\mathcal{D}'\leftarrow\emptyset$ in
Algorithm~\ref{alg:basic_recursion}. 

\section{Implementation and Benchmarking}
\label{sec:evaluation}

\subsection{The \texttt{SynDOKU} software package} 

The backtrace-and-prune algorithm described above is implemented in Python
as the open-source software \texttt{SynDOKU}
\cite{tuyet_minh_phan_2026_21705350}, available at
\texttt{GitHub}~\footnote{\texttt{https://github.com/tuyetminhphan/SynDOKU}}. It
takes an abstract reaction network $(X,R)$, a finite set of candidate
molecules, and, optionally a finite set of reaction templates as input.  To
connect to chemical databases, molecules are represented by SMILES strings
\cite{Weininger:88}. Molecular formulas, masses, and canonical molecular
representations are computed using \texttt{RDKit}~\cite{Landrum:13}.
Applicability of reaction templates to assigned reactions is performed with
\texttt{SynKit}~\cite{Phan:25SynKit}. Applying the transformation rule
results in predicted product graphs, which, after standardization and
canonicalization are compared to the product graph prescribed by the
partial assignment $\Phi$.
The current implementation assumes that the abstract stoichiometry
represents complete reactions and does not infer missing reactants or
products. When networks are constructed from incomplete database records,
these records may first be curated using reaction-rebalancing methods such
as \texttt{SynRBL} \cite{Phan:24a}.

The same graph-based formalism is used to extract symmetry information from
the abstract reaction network. Using \texttt{SynKit}, the network is
encoded as a directed bipartite incidence graph, in which vertices
represent abstract entities and reactions, while directed edges record
reactant and product roles together with stoichiometric
multiplicities. Automorphisms of this representation identify structurally
interchangeable entities and therefore determine the corresponding
orbits. Computationally, this step combines efficient Weisfeiler--Lehman
refinement procedures \cite{Shervashidze11} with exact
individualization--refinement algorithms for graph automorphism and
canonical labeling \cite{Junttila07,Mckay14}. The detected orbits provide
the symmetry information used for lexicographic orbit ordering,
symmetry-derived exclusion constraints, and dominance detection by
canonical state keys.

The software package computes the obligatory isomer and oligomer relations
described in Sect.~\ref{sec:conservation-laws} directly from the
stoichiometric matrix. An exact rational basis $\{u^{(1)},\dots,u^{(k)}\}$
of the left nullspace of $\mathbf{S}$ is computed with \texttt{SymPy}
\cite{Meurer:17}. For each $x\in X$, define its conservation signature by
\begin{equation}
  \lambda(x)=\bigl(u_x^{(1)},\dots,u_x^{(k)}\bigr)\in\mathbb{Q}^k .
\end{equation}
Every non-zero signature is written as $\lambda(x)=s_xd_x$, where $d_x$ is
an oriented primitive integer vector and $s_x>0$. For signatures with the
same direction, the scale factors are jointly reduced to relatively prime
positive integer levels $h_x$. Zero signatures are omitted, and opposite
directions are treated separately. If $d_x=d_y$, then
\begin{equation}
  h_y\,\sigma(\varphi(x)) = h_x\,\sigma(\varphi(y)).
\end{equation}
Thus, $h_x=h_y$ identifies obligatory isomers. If $h_y=n h_x$ for some
$n\in\mathbf{N}$ with $n>1$, then $\sigma(\varphi(y))=n\sigma(\varphi(x))$,
and $y$ is recorded as an obligatory $n$-oligomer of $x$. The corresponding
candidate domains $\mathcal{D}_x$ are pruned by retaining only molecules
whose sum-formulas satisfy these relations.

\subsection{Benchmark data sets}

\subsubsection{Rule-generated NOG networks}

We generated benchmark instances with the DPO graph-rewriting engine
\texttt{M{\O}D} \cite{Andersen:16a}. A derivation graph was expanded from
seed molecules and reaction rules $\mathfrak{R}$, duplicate and identity
reactions were removed, and a connected subnetwork of prescribed size was
selected. Its molecular graphs define a ground-truth bijection $\varphi^*$;
replacing them by abstract entities yields $(X,R)$ and the candidate set
$M$ with $|X|=|M|$. Each reaction retains its generating rule, providing
reaction-mechanism constraints in addition to atom balance.

We used the non-oxidative glycolysis (NOG) grammar \cite{Andersen:17}, with
eleven seed molecules (Supplemental Fig.~\ref{fig:nog-seeds}) and seven DPO
rules representing five reaction classes (Supplemental
Fig.~\ref{fig:si-nog-dpo} and Table~\ref{tab:nog-rules}). Products were
limited to eight carbon atoms in the initial expansion and twelve in its
extensions. Ten independently selected, nested network series were
constructed with
\begin{equation*}
 |X|=|M|=|R| \in \{42,63,84,105\}.
\end{equation*}
Each original network contains equal numbers of one-to-one, one-to-many,
and many-to-many reactions, and two independently permutable entity pairs,
giving a controlled four-fold symmetry.

To isolate the effect of reaction information, we kept $X$ and $M$ fixed
and extended only the reaction set from $|R|=|X|$ to
$|R|=11|X|/7$. We write NOG-X$x$-R$r$ for $x=|X|$ and $r=|R|$;
NOG-$x$ abbreviates the original NOG-X$x$-R$x$ instance. Isomerization
constrains molecular composition directly, whereas hydrolysis,
phosphoketolase, and carbon-transfer reactions can remain ambiguous under
mass and composition alone and therefore require reaction-template
validation.

\subsubsection{KEGG glycolysis module M00001}

In addition to the artificially generated NOG instances, we consider the
KEGG module M00001, i.e., the Embden--Meyerhof glycolysis pathway from
glucose to pyruvate \cite{Kanehisa:25}. The resulting instance comprises
$|X|=|M|=25$ entities and molecules and $|R|=15$ reactions. Its K{\"o}nig
graph and the abstract reactions are shown in Supplemental
Fig.~\ref{fig:si-m00001-koenig}. To determine the set of reaction
mechanisms $\mathfrak{R}$, we first computed atom-to-atom maps (AAMs) from
reaction SMILES with explicitly represented hydrogen atoms using
\texttt{GraphormerMapper} \cite{Nugmanov:22}. The resulting AAMs were
converted to ITS graphs and clustered with \texttt{SynTemp}
\cite{Phan:25SynTemp}, yielding eleven reaction-center classes. The
reaction center of \texttt{class:9} was corrected manually to include the role of the
\(\mathrm{Fe}/\mathrm{Fe}^{+}\) pair. The eleven reaction centers define
the DPO rules in $\mathfrak{R}$, which are shown in Supplemental
Fig.~\ref{fig:si-m00001-rules}. The 25 molecular graphs are
compiled in Supplemental Fig.~\ref{fig:si-m00001-molecules}. This example
differs from the NOG data in that both the reaction network and the molecular
graphs are derived from a biochemical database. It thus provides an example
in which the reactions have different numbers of reactants and products and
include cofactors as well as non-trivial stoichiometric coefficients.

\subsection{Experimental design}

All configurations enforce atom conservation and reaction-mechanism
validation. We compare five additional strategies: mass-interval pruning
(M), obligatory isomer and oligomer pruning (I/O), compositional pruning
(C), frontier scheduling (F), and symmetry breaking (S). The configurations
B1--B7 compare combinations of the three pruning procedures. Configurations
B8--B12 additionally use frontier scheduling, while B13--B16 also employ
symmetry breaking. The precise configurations are listed in Supplemental
Table~\ref{tab:ablation}.

For the study of the effect of additional reactions, we compare
configurations B13, B15, and B16 on the original and the reaction-enriched
networks. The two instances in each comparison have the same entities and
the same molecules and differ only in the number of reactions. The
corresponding configurations and results are given in Supplemental
Table~\ref{tab:reaction-extended-results}.

Running times are summarized over network seeds 0--9 by the mean and sample
standard deviation as well as by the median and interquartile range.
Comparisons with the ground truth assignment $\varphi^*$ are performed
subject to the symmetries of $(X,R)$. To this end we compare the
assignments $\varphi(\orb(x))$ and $\varphi^*(\orb(x))$ as multisets.  This
comparison regards two assignments as equivalent if they differ only by an
automorphism of the abstract reaction network. In particular, it does not
require an arbitrary correspondence between entities that belong to the
same orbit. In every completed NOG run, the ground-truth assignment was
recovered up to such an automorphism.

\subsection{Benchmark results}

\subsubsection{Network structure and reaction information}
The structural properties of the NOG networks are summarized in
Supplemental Table~\ref{tab:network-information}. For the original
networks, the average reaction degree is $\bar d_R=3.00$ by construction.
The mean value of the maximal reaction degree increases from $12.20$ for
NOG-X42-R42 to $16.30$ for NOG-X105-R105. At the same time, the mean
density of obligatory isomer relations decreases from $4.27\%$ to
$2.84\%$, and the mean density of obligatory oligomer relations decreases
from $12.09\%$ to $9.28\%$. We note that the standard deviations of the
oligomer densities are large. Thus, the ten independently generated series
differ considerably in the amount of information that can be obtained from
the conservation laws.

Adding reactions while holding $X$ and $M$ fixed raises
$\bar d_R$ from $3.00$ to $4.71$ and increases the maximum reaction degree
for all four network sizes. The densities of obligatory isomer and oligomer
relations also increase in all four cases. For NOG-X105-R165, for instance,
the mean isomer density is $3.58\%$, compared to $2.84\%$ for
NOG-X105-R105, while the corresponding oligomer densities are $14.58\%$
and $9.28\%$, respectively. The additional reactions therefore provide
more constraints without changing the set of molecules.

\subsubsection{Ablation study}
Supplemental Table~\ref{tab:ablation} compares the different pruning and
search strategies. As expected, the running times increase with the size of the
network. We observe, however, that frontier scheduling becomes increasingly
important for the larger instances. For example, adding F to compositional
pruning (B3 versus B9) reduces the mean running time from $32.32$ to $16.73$
seconds at X84 and from $235.24$ to $88.81$ seconds at X105. Similarly,
adding F to the I/O+C configuration (B6 versus B11) reduces the X105 mean
from $229.35$ to $100.74$ seconds. Among the configurations with frontier
scheduling, B11 has the lowest mean and median running time up to X84 and the
lowest median running time at X105.

Symmetry breaking provides a further improvement as the instances grow.
Configurations without S return the four assignments induced by the two
controlled transpositions, whereas configurations with S return one
representative. Comparing otherwise identical settings, B11 versus B15
(I/O+C+F) decreases the X105 mean running time from $100.74$ to $51.93$
seconds, and B12 versus B16 (M+I/O+C+F) decreases it from $130.44$ to
$60.45$ seconds. For the smaller instances the differences are much less
pronounced. We also note that enabling all pruning procedures does not
always give the smallest running time. B15 has a smaller mean running time
than B16 for X42, X63, and X84, while B13 has the smallest mean running
time at X105.  Thus the computational effort of an additional pruning step
is not always compensated by a further reduction of the search
tree. Obligatory isomer and oligomer pruning by itself (B2) gives
particularly variable running times; it was too slow to complete the full
benchmark at the two largest sizes.  Nine of ten X84 seeds and seven of ten
X105 seeds completed; the remaining runs exceeded the 1800-second timeout,
so summary statistics for these two incomplete cells are not reported.

Mass-interval and compositional pruning exhibit a notable interaction
with frontier scheduling. Without F, M and C have similar mean running times at
X42 and X63, but M is faster at X84 (B1: 21.98\,s; B3: 32.32\,s) and X105
(B1: 137.46\,s; B3: 235.24\,s). With F, the ranking reverses: C+F is faster
than M+F at every size and reduces the X105 mean from 125.42\,s for B8 to
88.81\,s for B9. Thus, the utility of a pruning criterion cannot be assessed
independently of the variable-ordering strategy with which it is combined.

Figure~\ref{fig:nog-ablation-landscape} summarizes the results of the
ablation study. Panel B shows how frontier scheduling reverses the relative
performance of M and C, while panels C and D show the seed-matched speedups
from adding frontier scheduling and symmetry breaking, respectively.
The seed-level running time distributions for representative configurations
from the three search regimes are shown in
Figure~\ref{fig:nog-running time-distributions}.

\begin{figure}[H]
  \centering
  \includegraphics[width=\linewidth]{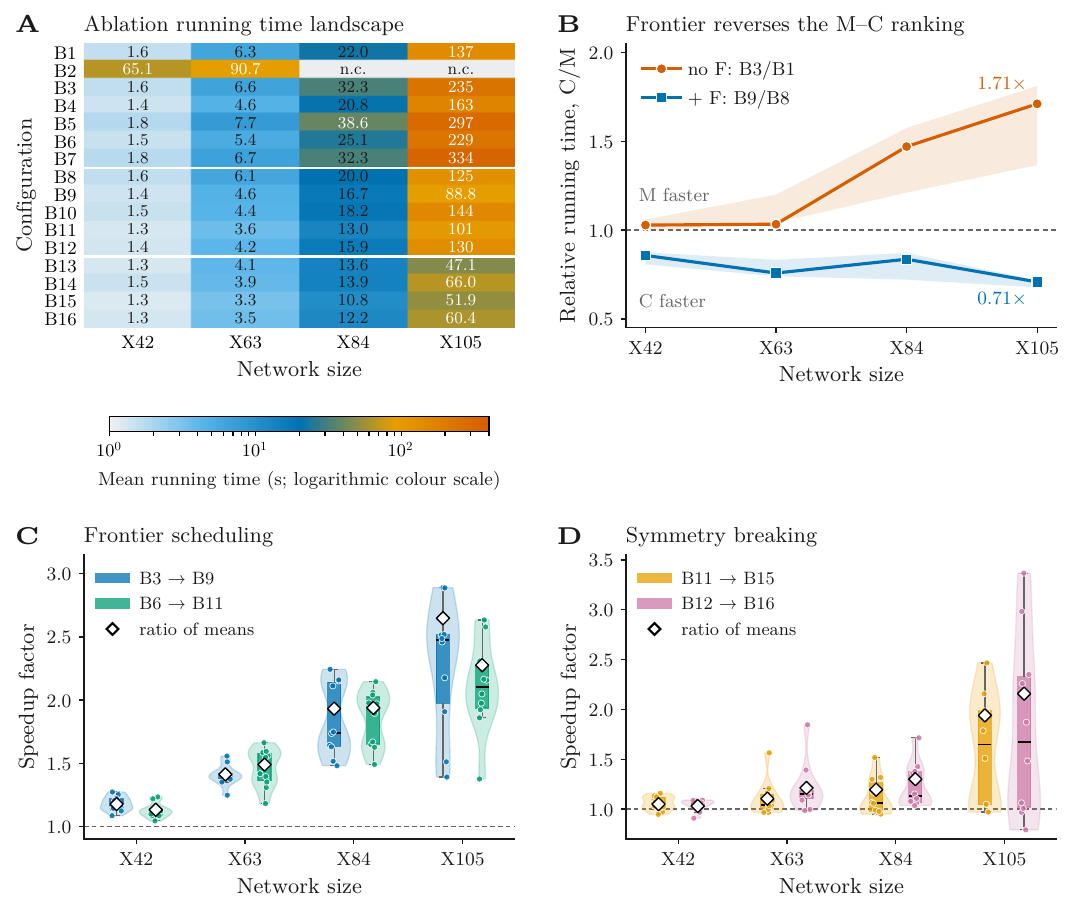}
  \caption{Running time landscape and paired effects of the NOG ablation
    strategies. (A) Mean running time for configurations B1--B16 and the
    four original network sizes; the colour scale is logarithmic. The
    entries marked n.c.\ were not completed because B2 was too slow at X84
    and X105. White horizontal rules separate pruning-only configurations,
    configurations with frontier scheduling, and configurations with
    symmetry breaking.  (B) Relative running time of compositional pruning
    to mass-interval pruning, without frontier scheduling (B3/B1) and with
    frontier scheduling (B9/B8). Markers show ratios of arithmetic mean
    running times and shaded bands show the interquartile range of
    seed-matched running time ratios.  Values above the dashed line at one
    favour M, whereas values below it favour C. (C,D) Seed-matched speedup
    factors after adding frontier scheduling or symmetry breaking,
    respectively. Each speedup is the baseline running time divided by the
    optimized running time, so values above the dashed line at one indicate
    an improvement. Violin envelopes show kernel-density estimates, boxes
    show the interquartile range and median, circles denote individual
    network seeds, and white diamonds denote ratios of arithmetic mean
    running times.}
  \label{fig:nog-ablation-landscape}
\end{figure}

\begin{figure}[H]
  \centering
  \includegraphics[width=\linewidth]{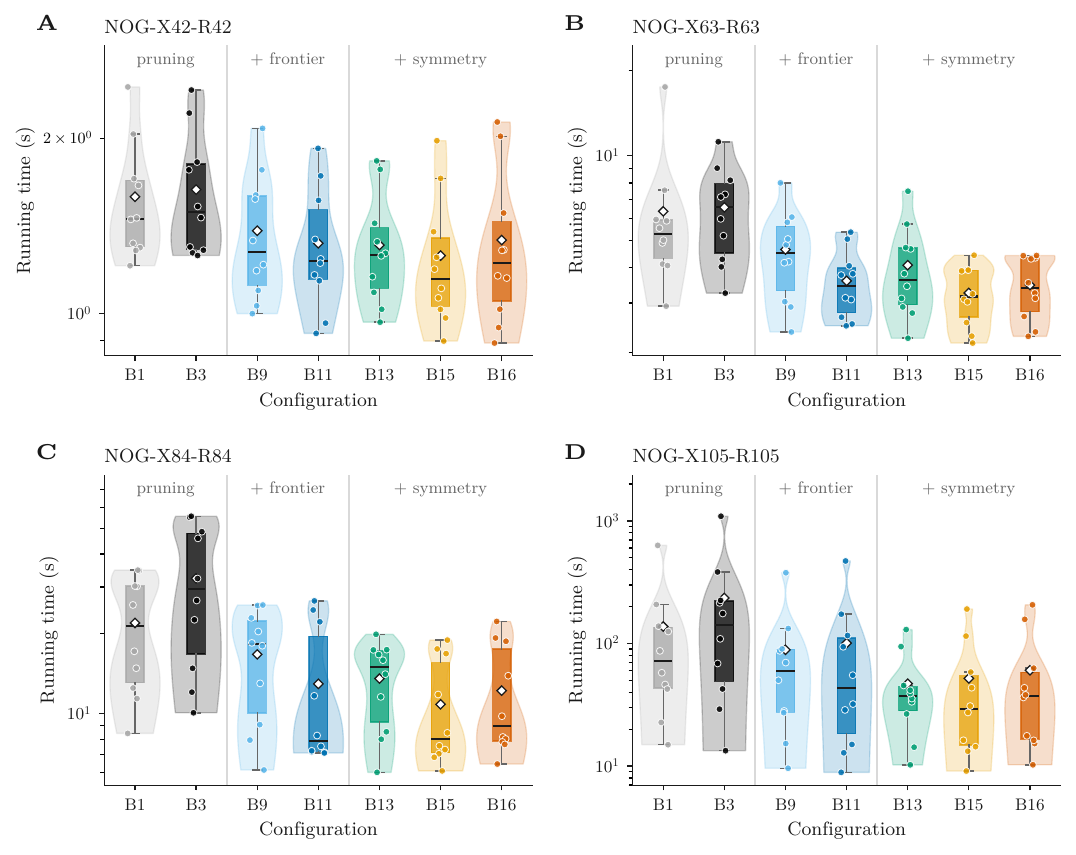}

  \caption{Seed-level running time distributions for representative NOG
    configurations. Panels A--D correspond to the four original network
    sizes. Violin envelopes show kernel-density estimates across the ten
    network seeds (bandwidth factor 0.45); boxes show the interquartile
    range and median, whiskers extend to the most extreme observations
    within 1.5 interquartile ranges, circles denote individual seeds, and
    diamonds denote arithmetic means. Vertical rules separate pruning-only
    configurations (B1 and B3), configurations with frontier scheduling (B9
    and B11), and configurations with symmetry breaking (B13, B15, and
    B16). Running time axes are logarithmic and scaled independently in
    each panel.}
  \label{fig:nog-running time-distributions}
\end{figure}

\subsubsection{Reaction-enriched networks}
The effect of increasing the number of reactions is shown in Supplemental
Table~\ref{tab:reaction-extended-results}. Since the set of molecules is
unchanged, these comparisons isolate the information contributed by the
additional reactions. For X42 and X63 the differences are small and are not
consistent across the three configurations. For X84, increasing $|R|$ from
84 to 132 reduces the mean running time from $13.55$ to $7.52$ seconds for
B13, from $10.85$ to $8.29$ seconds for B15, and from $12.21$ to $9.06$
seconds for B16. At X105, enrichment from 105 to 165 reactions reduces the
corresponding means from $47.06$, $51.93$, and $60.45$ seconds to $33.90$,
$41.35$, and $46.60$ seconds. Thus, additional reactions do not necessarily
increase the running time. For the two larger NOG data sets, the reduction
of the candidate domains due to the additional constraints more than
compensates for the additional reaction-mechanism checks.

Figure~\ref{fig:nog-reaction-enrichment} shows that this effect is driven
by a reduction in search effort rather than by uniformly lower per-reaction
work. The benefit is most evident for X84 and X105, while the smaller
networks show little or inconsistent running time gain.
\begin{figure}[H]
  \centering
  \includegraphics[width=\linewidth]{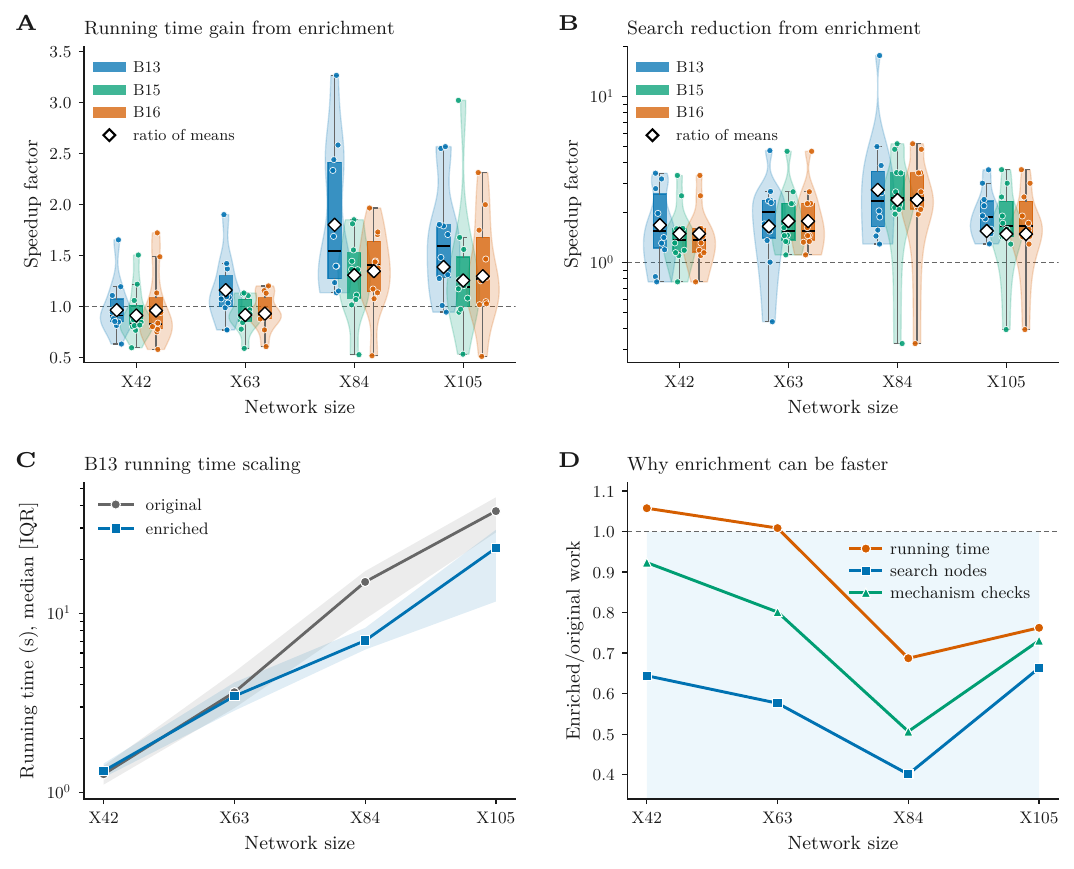}

  \caption{Effect of reaction enrichment with fixed entity and molecule
    sets. (A) Seed-matched running time speedup and (B) seed-matched
    reduction in explored search nodes for configurations B13, B15, and
    B16. Ratios are original divided by enriched, so values above the
    dashed line at one favour the enriched networks. Violin envelopes show
    kernel-density estimates, boxes show the interquartile range and
    median, circles denote individual network seeds, and white diamonds
    denote ratios of arithmetic means. The search-reduction axis in panel B
    is logarithmic.  (C) Median running time and interquartile range for
    B13 on the original and enriched networks. (D) Enriched-to-original
    ratios of mean running time, explored search nodes, and
    reaction-mechanism checks, averaged over B13, B15, and B16. Values
    below one indicate less work on the enriched networks.}
  \label{fig:nog-reaction-enrichment}
\end{figure}

\subsubsection{KEGG glycolysis module M00001}

The structural characteristics of the KEGG M00001 instance are
included in Supplemental Table~\ref{tab:network-information}. The network
contains 25 abstract entities and 15 reactions, with
$\bar d_R=2.32$ and $d_R^{\max}=6$. Of the 15 reactions, three
($20.0\%$) are one-to-one, two ($13.3\%$) are one-to-many, and ten
($66.7\%$) are many-to-many. The conservation laws yield three obligatory
isomer relations and two obligatory oligomer relations, corresponding to
densities of $1.00\%$ and $0.67\%$, respectively. Compared with the NOG
instances, the KEGG network therefore combines a smaller number of
entities with more heterogeneous reaction arities and a predominance of
many-to-many reactions. This reflects the explicit representation of
cofactors and non-trivial stoichiometric coefficients.

The automorphism analysis identifies the two nontrivial species orbits
$\{L,P\}$ and $\{N,Q\}$. These orbits are not independently permutable.
Instead, the corresponding automorphism acts as
$\vartheta_X=(L\,P)(N\,Q)$ on the species and as
$\vartheta_R=(r_8\,r_9)$ on the reactions. Thus, exchanging $L$ and $P$
requires the simultaneous exchange of $N$ and $Q$. This coupling is
important for symmetry breaking because the two transpositions must be
treated as a single network automorphism.

The backtracking results are given in Supplemental
Table~\ref{tab:m00001-backtracking}. All three configurations employ
compositional pruning, frontier scheduling, and symmetry breaking.
Configuration B13 explores 145 search nodes in $489.11$ seconds. Adding
obligatory isomer and oligomer pruning in B15 does not reduce the search
effort for this instance: the number of explored nodes increases to 152 and
the running time to $494.90$ seconds. With mass-interval pruning
additionally enabled, B16 explores 144 nodes in $439.35$ seconds. Thus, B16
is the fastest of the three configurations, reducing the running time by
$10.2\%$ relative to B13 and by $11.2\%$ relative to B15. The three
configurations nevertheless explore nearly the same number of nodes.
Hence, for this instance, the obligatory isomer and oligomer relations and
the mass intervals do not substantially contract the search tree beyond the
restrictions already imposed by compositional pruning, frontier scheduling,
and symmetry breaking.

Despite having fewer entities and reactions than the NOG networks, M00001
requires between $439$ and $495$ seconds, while exploring only 144--152
search nodes. This result illustrates that the running time is not
determined by network size or the number of explored search nodes alone:
the cost of validating the chemically heterogeneous reactions at each
search node is also consequential.

\subsubsection{Non-uniqueness beyond network symmetry}
\label{sec:nog-counterexample}

There is no guarantee that solutions to the assignment problems
considered here are unique. This is obvious if the target network contains
symmetries. Non-unique solutions, however, are not limited to such cases.
The analysis of the NOG benchmark identified an ambiguity that cannot be
attributed to the prescribed molecular symmetry orbits. The corresponding
instance contains the two aldol reactions in
Fig.~\ref{fig:nog-product-exchange}.

\begin{figure}[htbp]
  \centering
  \resizebox{0.98\linewidth}{!}{\begin{tikzpicture}[
  x=1cm,y=1cm,font=\rmfamily\footnotesize,
  line cap=round,line join=round,
  rulecard/.style={rounded corners=2.6pt,draw=hair!85,fill=white,
    line width=0.6pt},
  rulearrow/.style={-{Stealth[length=1.8mm,width=1.35mm]},
    line width=0.7pt,draw=ink!48},
  dpoatom/.style={circle,minimum size=5.2mm,inner sep=0pt,
    line width=0.55pt,
    font=\figatomfont\bfseries\fontsize{6.4}{6.4}\selectfont},
  dpoC/.style={dpoatom,draw=cpkC!70,fill=cpkC,text=white},
  dpoO/.style={dpoatom,draw=cpkO!80!black,fill=cpkO,text=white},
  dpoH/.style={dpoatom,draw=hair,fill=white,text=ink2},
  keepbond/.style={draw=ink!55,line width=1.0pt},
  delbond/.style={draw=pbrk,line width=1.45pt},
  newbond/.style={draw=npgGreen!70!black,line width=1.55pt}
]
  \draw[spanel] (0.00,5.10) rectangle (13.60,8.35);
  \node[badge] at (0.36,8.01) {A};

  \draw[rulecard,draw=pbrk!42] (0.56,5.38) rectangle (4.06,7.72);
  \draw[rulecard]               (4.86,5.38) rectangle (8.74,7.72);
  \draw[rulecard,draw=npgGreen!52!black]
                                  (9.54,5.38) rectangle (13.04,7.72);

  \node[font=\rmfamily\bfseries\large,text=ink] at (2.31,7.42) {\(L\)};
  \node[font=\rmfamily\bfseries\large,text=ink] at (6.80,7.42) {\(K\)};
  \node[font=\rmfamily\bfseries\large,text=ink] at (11.29,7.42) {\(R\)};

  \begin{scope}[yshift=-2mm]
  \node[dpoC] (LdC) at (1.86,6.05) {C};
  \node[dpoH] (LH)  at (1.86,6.95) {H};
  \draw[delbond] (LdC)--(LH);
  \node[dpoC] (LaC) at (2.76,6.05) {C};
  \node[dpoO] (LO)  at (2.76,6.95) {O};
  \draw[keepbond]
    ($(LaC)+(-1.55pt,2.15mm)$)--($(LO)+(-1.55pt,-2.15mm)$);
  \draw[delbond]
    ($(LaC)+(1.55pt,2.15mm)$)--($(LO)+(1.55pt,-2.15mm)$);
  \node[dpoC] (KdC) at (6.35,6.05) {C};
  \node[dpoH] (KH)  at (6.35,6.95) {H};
  \node[dpoC] (KaC) at (7.25,6.05) {C};
  \node[dpoO] (KO)  at (7.25,6.95) {O};
  \draw[keepbond] (KaC)--(KO);

  \node[dpoC] (RdC) at (10.84,6.05) {C};
  \node[dpoH] (RH)  at (10.84,6.95) {H};
  \node[dpoC] (RaC) at (11.74,6.05) {C};
  \node[dpoO] (RO)  at (11.74,6.95) {O};
  \draw[newbond]  (RdC)--(RaC);
  \draw[keepbond] (RaC)--(RO);
  \draw[newbond]  (RO)--(RH);
  \end{scope}

  \draw[rulearrow] (4.78,6.30)--(4.14,6.30)
    node[midway,above=2pt,font=\rmfamily\scriptsize,text=mut] {\(l\)};
  \draw[rulearrow] (8.82,6.30)--(9.46,6.30)
    node[midway,above=2pt,font=\rmfamily\scriptsize,text=mut] {\(r\)};

  \draw[spanel] (0.00,2.60) rectangle (13.60,4.85);
  \node[badge] at (0.36,4.51) {B};
  \node[inner sep=0] at (7.02,3.70) {%
    \resizebox{12.25cm}{!}{%
      \schemestart
        \chemfig{HO-[:-30]-[:30](=[:90]O)
          -[:-30](-[:-90]OH)
                    (-[:90,,,,draw=pbrk,line width=1.4pt]H)
          -[:30](-[:90]OH)-[:-30](-[:-90]OH)-[:30]-[:-30]OH}
        \+
        \chemfig{O=[:-30,,,,draw=pbrk,line width=1.4pt]
          -[:30]-[:-30]O-P(=[:90]O)(-[:-90]OH)-OH}
        \arrow{->}
        \chemfig{HO-P(=[:90]O)(-[:-90]OH)-O-[:-30]
          -[:30](-[:90]O-[:25,,,,draw=npgGreen!70!black,
                            line width=1.4pt]H)
          -[:-30,,,,draw=npgGreen!70!black,line width=1.4pt]
            (-[:-90]OH)
          -[:30](=[:90]O)-[:-30](-[:-90]OH)-[:30](-[:90]OH)
          -[:-30](-[:-90]OH)-[:30]-[:-30]OH}
      \schemestop
    }%
  };

  \draw[spanel] (0.00,0.10) rectangle (13.60,2.35);
  \node[badge] at (0.36,2.01) {C};
  \node[inner sep=0] at (7.02,1.20) {%
    \resizebox{12.25cm}{!}{%
      \schemestart
        \chemfig{HO-[:-30]-[:30](=[:90]O)
          -[:-30](-[:-90]OH)
                    (-[:90,,,,draw=pbrk,line width=1.4pt]H)
          -[:30](-[:90]OH)-[:-30]-[:30]O-P(=[:90]O)(-[:-90]OH)-OH}
        \+
        \chemfig{O=[:-30,,,,draw=pbrk,line width=1.4pt]
          -[:30](-[:90]OH)-[:-30]-[:30]OH}
        \arrow{->}
        \chemfig{HO-P(=[:90]O)(-[:-90]OH)-O-[:-30]
          -[:30](-[:90]OH)-[:-30]
          -[:30](-[:90]O-[:25,,,,draw=npgGreen!70!black,
                            line width=1.4pt]H)
          (-[:-90,,,,draw=npgGreen!70!black,line width=1.4pt]
            -[:-150](=[:-90]O)-[:150]-[:-150]OH)
          -[:-30](-[:-90]OH)-[:30](-[:90]OH)-[:-30]-[:30]OH}
      \schemestop
    }%
  };
\end{tikzpicture}}
  \caption{Aldol-rule validation under product exchange.
    (A) DPO rule $L\xleftarrow{\,l\,}K\xrightarrow{\,r\,}R$.
    (B,C) Its application to two complete reactions.
    Red bonds are deleted and green bonds are formed.}
  \label{fig:nog-product-exchange}
\end{figure}
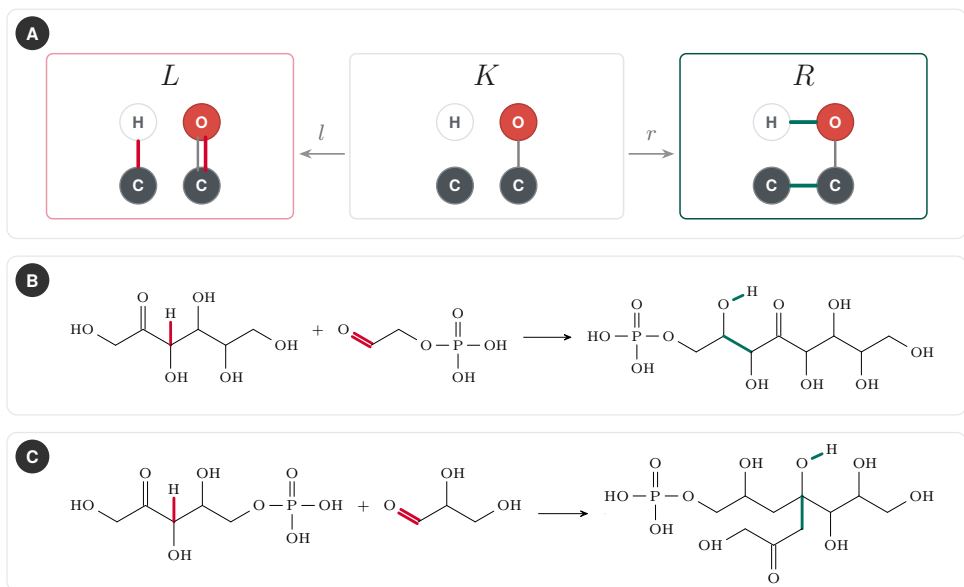

The two products have the same molecular formula,
$\mathrm{C_8H_{17}O_{11}P}$. When they are exchanged between these two
reactions, both resulting reactions still satisfy elemental balance and
the implemented aldol-reaction template.
Fig.~\ref{fig:nog-product-exchange}A makes the common local rewrite
explicit: the indicated hydrogen is transferred from the donor carbon to
the acceptor oxygen, the acceptor carbonyl is reduced to a single bond,
and the new carbon--carbon bond is formed. Nevertheless, this exchange is
not generated by any of the prescribed symmetry orbits of the abstract
network. The ambiguity is therefore chemical rather than a consequence of
an automorphism of $(X,R)$: two distinct assignments satisfy the same
formula and local reaction-mechanism constraints.

The KEGG M00001 instance provides a complementary example. All three
configurations return exactly two solutions. One agrees with the
ground-truth assignment $\varphi^*$ up to an automorphism of $(X,R)$,
whereas the other exchanges $G$ and $T$, which represent fructose
6-phosphate and glucose 6-phosphate, respectively. The pair $\{G,T\}$ is
identified as an obligatory isomer relation, but neither entity belongs
to a nontrivial species orbit. Consequently, compositional constraints and
symmetry breaking cannot distinguish the two assignments. Their ambiguity
instead arises because, in particular, the reaction center of
\texttt{class:10} describes only the local bond changes of the
isomerization and lacks sufficient chemical context to distinguish these
isomers consistently.

Together, the NOG and KEGG examples show that uniqueness modulo network
symmetry is stronger than feasibility under elemental balance and local
reaction-mechanism validation. The ground-truth assignment may therefore
be recovered among the feasible solutions without being uniquely
determined by the available constraints.

\section{Concluding Remarks}

Abstract reaction networks provide an appealing paradigm for non-standard
computation and a framework for modelling a wide range of complex
systems. It is therefore natural to ask whether and how a given reaction
network can be implemented by means of chemical reactions. In recent work
\cite{Mueller:22a}, it was shown that every conservative reaction network
has a realization in terms of chemical graphs and (atom-)balanced
reactions. Here we ask how difficult it is to find such a realization in
practice when candidate molecules and, possibly, the relevant reaction
mechanisms are given. We show that the problem is NP-complete even in the
simplest case, in which all reactions have the form $A+B\longrightarrow C$
and all molecules are oligomers of a single moiety. Nevertheless, we
develop a backtrace-and-prune algorithm inspired by the subgraph
isomorphism problem and general principles from constraint programming that
solves the problem efficiently, at least for moderate-size instances. We
use chemical reaction networks produced by generative graph grammars as
benchmarking data and explore the impact of different strategies for
pruning the search trees. We observe that symmetries of the abstract
network have a substantial impact on performance and thus require
systematic handling.

The benchmark results further show that the utility of a pruning criterion
depends on the search strategy with which it is combined: frontier
scheduling becomes increasingly important with network size, while enabling
all pruning procedures is not always optimal. Increasing the number of
reactions can even reduce the running time when the additional constraints
contract the search sufficiently to compensate for the additional
reaction-mechanism checks. The KEGG M00001 example further demonstrates
that neither network size nor the number of explored search nodes alone
determines the running time, since chemically heterogeneous reactions can
make validation at each search node costly. Finally, the NOG and KEGG
examples show that recovering the ground-truth assignment among the
feasible solutions does not imply uniqueness modulo network symmetry.

So far, we have considered only the simplest case $|X|=|M|$. Clearly, the
results remain valid if $|M|>|X|$ remains bounded and $\varphi:X\to M$ is
required to be injective, i.e., different abstract entities must be
realized by different molecules. Some further generalizations can also be
handled in a similar manner, at least in principle.
\begin{definition}
  A reaction network is \emph{generated from $S\subseteq X$} if there is a
  (partial) order $\prec$ of the reactions that for each $r\in R$,
  $s^-_{xr}>0$ it holds that $x\in S$ or there is a reaction $r'\in R$ with
  $r'\preceq r$ such that $s^-_{xr}>0$. A set $S$ from which $(X,R)$ is
  generated is called a \emph{seed set}. 
\end{definition}
It then suffices to provide a finite candidate set $M$ for the seed
set $S\subseteq X$ to allow the construction of a finite candidate set
$M^*\subset\mathfrak{M}$. For larger networks, however, the set $M^*$ may
become too large to be of practical use. We expect this to be the case in
particular when the reactions are not constrained by a small set of allowed
mechanisms.

On the other hand, a key result of \cite{Mueller:22a},
Prop.~\ref{prop:exists}, guarantees that, without constraints on molecules
and reactions, it is always possible to find a solution. It remains an open
question whether this is also true for a small set $\mathfrak{R}$ of
prescribed reaction mechanisms and unconstrained choices of molecules,
i.e., $M=\mathfrak{M}$. In this case, the infinite set $\mathfrak{M}$ is
not part of the input. It is also unclear whether this variant of the
problem remains in NP, since the molecules $m\in\mathfrak{M}$ needed for a
realization may grow exponentially with the input size
$N\coloneqq |X|+|R|+|\mathfrak{R}|$. Thus, there is no longer a guarantee
that a certificate can be verified in polynomial time.

Several interesting questions remain open for future research. While a
given seed set at least in principle enables the strategy pursued here,
this is no longer true if no constraints on the set of molecules are
provided as part of the input. Moreover, it is often useful to consider
(chemical) reaction networks embedded in an open system that provides
``food'' from an external reservoir and removes a commensurate output into
an external reservoir. In such a setting, the appropriate notion of
realizability will also involve energetic considerations, e.g., an overall
exothermic transformation from ``food'' reactants to output products and,
depending on the application, even kinetic considerations.

\backmatter

\section*{Declarations}

\bmhead{Funding}

This work has received funding from the European Union’s Horizon Europe
Doctoral Network programme under the Marie Sk{\l}odowska-Curie grant
agreement No. 101072930 (TACsy---Training Alliance for Computational
systems chemistry).  Views and opinions expressed are however those of the
author(s) only and do not necessarily reflect those of the European
Union. Neither the European Union nor the granting authority can be held
responsible for them.

Further support was provided by the Novo Nordisk Foundation under grant
reference NNF21OC0066551 (Mathematical Modelling for Microbial Community
Induced Metabolic Diseases, MATOMIC).

PFS acknowledges support by the Federal Ministry of Research, Technology,
and Space of Germany through DAAD project 57616814 (SECAI, School of
Embedded Composite AI), jointly with S{\"a}chsische Staatsministerium für
Wissenschaft, Kultur und Tourismus in the programme Center of Excellence
for AI-research \emph{Center for Scalable Data Analytics and Artificial
Intelligence Dresden/Leipzig} (ScaDS.AI, proj.\ no.\ SCADS24B), and by the
Deutsche Forschungsgemeinschaft (DFG, German Research Foundation) under
Germany's Excellence Strategy (EXC-3105/1, 533765739).

\bmhead{Competing interests}
The authors declare no competing interests.

\bmhead{Author contributions}
T.-M.P. developed the software, curated the benchmark data, performed the
experiments, and prepared the initial manuscript. D.M., T.-L.P., C.F., and
P.F.S. contributed to the methodology, formal analysis, interpretation of the
results, and review and editing of the manuscript. P.F.S. supervised the
project and acquired funding. All authors reviewed and approved the final
manuscript.

\bmhead{Data availability}
The benchmark inputs and result tables supporting this study are distributed
with \texttt{SynDOKU} v0.2.0. The source repository is available at
\url{https://github.com/tuyetminhphan/SynDOKU}, and the archived release is
available at \url{https://doi.org/10.5281/zenodo.21705350}.

\bmhead{Code availability}
\texttt{SynDOKU} v0.2.0 is available under the Apache License 2.0 from
\href{https://pypi.org/project/syndoku/}{PyPI},
\href{https://anaconda.org/tuyetminhphan/syndoku}{Anaconda.org}, and
\href{https://hub.docker.com/r/tuyetminhphan/syndoku}{Docker Hub}. Source code
is available at \url{https://github.com/tuyetminhphan/SynDOKU}.

\bmhead{Ethics approval and consent to participate}
Not applicable.

\bmhead{Consent for publication}
Not applicable.

\bibliography{syndoku}

\clearpage
\setcounter{page}{1}
\renewcommand{\thepage}{S\arabic{page}}
\section*{Supplemental Material}

\setcounter{figure}{0}
\renewcommand{\thefigure}{S\arabic{figure}}
\renewcommand{\theHfigure}{S\arabic{figure}}
\setcounter{table}{0}
\renewcommand{\thetable}{S\arabic{table}}
\renewcommand{\theHtable}{S\arabic{table}}

\subsection*{S1: Software Installation and Benchmark Environment}

\subsubsection*{S1.1: Installation}

\texttt{SynDOKU} \cite{tuyet_minh_phan_2026_21705350} requires
Python~3.11 or newer. The software can be installed from
\href{https://pypi.org/project/syndoku/}{PyPI} with:
\begin{lstlisting}[style=shell]
pip install syndoku
\end{lstlisting}
It is also available as a platform-independent package from the
\href{https://anaconda.org/tuyetminhphan/syndoku}{author's channel on
Anaconda.org}:
\begin{lstlisting}[style=shell]
conda install tuyetminhphan::syndoku
\end{lstlisting}
A pre-built Linux container, including the runtime environment and all
dependencies, can be obtained from
\href{https://hub.docker.com/r/tuyetminhphan/syndoku}{Docker Hub}:
\begin{lstlisting}[style=shell]
docker pull tuyetminhphan/syndoku
\end{lstlisting}
Alternatively, a source checkout can be installed directly from the
\texttt{GitHub} repository:
\begin{lstlisting}[style=shell]
git clone https://github.com/tuyetminhphan/SynDOKU.git
cd SynDOKU
python -m pip install .
\end{lstlisting}
Each method installs \texttt{SynDOKU} together with the dependencies needed
to run the software.

\subsubsection*{S1.2: Benchmark Environment}

The benchmarks were run with \texttt{SynDOKU} version 0.2.0 on an HP
EliteBook 8 G1a 14-inch computer running Fedora Linux 43. The computer was
equipped with an 8-core/16-thread AMD
Ryzen 7 250 processor with a maximum boost frequency of 5.13\,GHz,
integrated AMD Radeon 780M graphics, an AMD NPU, 64\,GB RAM, and a 2\,TB
SK hynix NVMe solid-state drive.

\clearpage

\subsection*{S2: Details of Benchmarking Data Sets}

\subsubsection*{S2.1: Rule-Generated NOG Networks}

\begin{table}[!ht]
  \centering
  \caption{Rule classes represented in the M{\O}D-generated NOG
    benchmark. The stoichiometric schema indicates the species-level
    structure induced by each rule class.}
  \label{tab:nog-rules}
  \begin{tabular}{lll}
  \hline
  Rule(s) & Reaction class & Stoichiometric schema \\
  \hline
  AL
    & Aldol-type C--C bond formation
    & $X_1+X_2\to X_3$ \\
  AlKe/KeAl
    & Aldose--ketose interconversion
    & $X_1\to X_2$ \\
  PHL
    & Phosphate ester hydrolysis
    & $X_1+\mathrm{H_2O}\to P_i+X_2$ \\
  PK
    & Phosphoketolase cleavage
    & $P_i+X_1\to \mathrm{H_2O}+X_2+X_3$ \\
  TAL/TKL
    & Carbon-fragment transfer
    & $X_1+X_2\to X_3+X_4$ \\
  \hline
  \end{tabular}
\end{table}

\clearpage

\begin{figure}[!ht]
  \centering
  \input{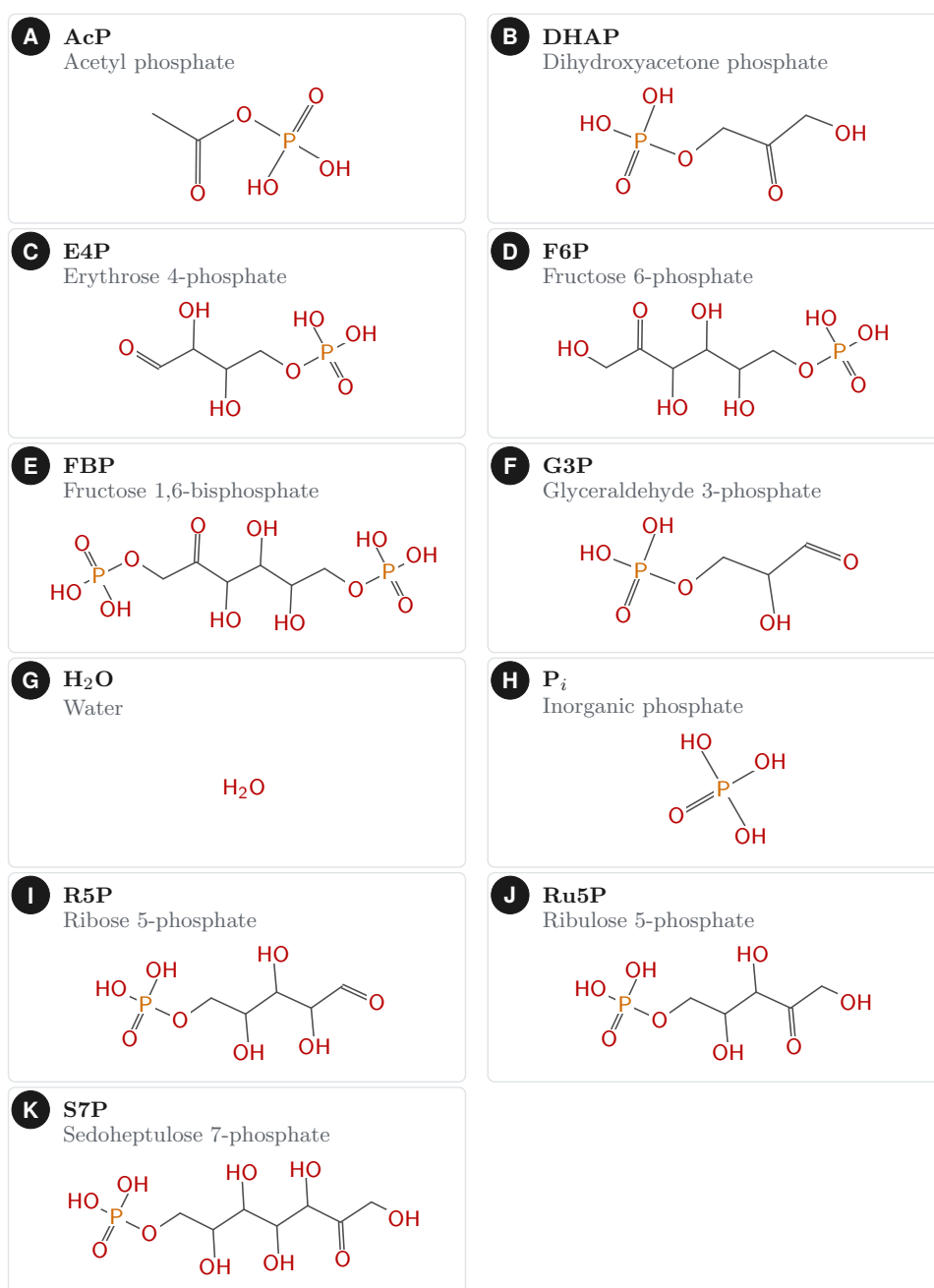}
  \resizebox{0.965\linewidth}{!}{\nogseedgallery}
  \caption{Seed molecules used for the M{\O}D-generated
    non-oxidative glycolysis benchmark. Panels (A)--(K) show
    AcP, DHAP, E4P, F6P, FBP, G3P, water, inorganic phosphate, R5P, Ru5P,
    and S7P, respectively.}
  \label{fig:nog-seeds}
\end{figure}

\clearpage

\begingroup

\providecolor{ink}{HTML}{252A2E}
\providecolor{ink2}{HTML}{4E5961}
\providecolor{hair}{HTML}{D8DEE2}
\providecolor{pbrk}{HTML}{D55E45}
\providecolor{npgGreen}{HTML}{198754}
\providecolor{cpkC}{HTML}{4B5358}
\providecolor{cpkO}{HTML}{D94A42}
\providecolor{cpkP}{HTML}{E49B29}

\tikzset{
  badge/.style={
    circle,fill=ink,text=white,draw=none,minimum size=5.4mm,
    inner sep=0pt,font=\figatomfont\bfseries\footnotesize
  },
  dporule/.style={line cap=round,line join=round},
  rulecard/.style={
    rounded corners=1.4mm,draw=hair,fill=white,line width=0.55pt
  },
  cardlab/.style={font=\rmfamily\bfseries\large,text=ink},
  spanlab/.style={font=\rmfamily\bfseries\small,text=ink},
  rulearrow/.style={-{Stealth[length=2.4mm,width=1.7mm]},draw=ink2,line width=0.75pt},
  keepbond/.style={draw=ink2,line width=0.78pt},
  delbond/.style={draw=pbrk,line width=1.05pt},
  newbond/.style={draw=npgGreen!86!black,line width=1.05pt},
  dpoC/.style={circle,draw=ink,fill=cpkC,text=white,inner sep=0pt},
  dpoO/.style={circle,draw=pbrk!85!black,fill=cpkO,text=white,inner sep=0pt},
  dpoP/.style={circle,draw=cpkP!70!black,fill=cpkP,text=ink,inner sep=0pt},
  dpoH/.style={circle,draw=ink2,fill=white,text=ink,inner sep=0pt},
  atomid/.style={
    label distance=1.25pt,
    label={[font=\rmfamily\bfseries\tiny,text=ink,inner sep=0.35pt]45:#1}
  },
  atomidpos/.style args={#1/#2}{
    label distance=1.25pt,
    label={[font=\rmfamily\bfseries\tiny,text=ink,inner sep=0.35pt]#2:#1}
  }
}

\providecommand{\dbond}[4]{}
\renewcommand{\dbond}[4]{%
  \draw[#3,shorten <=\dposhorten,shorten >=\dposhorten]
    ($(#1)!1.15pt!90:(#2)$)--($(#2)!1.15pt!-90:(#1)$);
  \draw[#4,shorten <=\dposhorten,shorten >=\dposhorten]
    ($(#1)!1.15pt!-90:(#2)$)--($(#2)!1.15pt!90:(#1)$);
}

\newcommand{\dpospanlegenditem}[2]{%
  \tikz[baseline=-0.6ex]{\draw[#1] (0,0)--(0.75,0);}%
  \hspace{0.35em}{\rmfamily\scriptsize\color{ink2}#2}%
}
\newcommand{\dpospanlegend}{%
  \noindent\hbox to \linewidth{%
    \dpospanlegenditem{keepbond}{preserved}\hfil
    \dpospanlegenditem{delbond}{broken}\hfil
    \dpospanlegenditem{newbond}{formed}\hfil
    {\rmfamily\scriptsize\color{ink2}numbers = atom map}%
  }%
}

\newcommand{\nogdpopanelsABC}{%
  \noindent\hfill\begin{minipage}[t]{0.940\linewidth}
  \raggedright
  \tikz[baseline=-0.65ex]{\node[badge] {A};}\hspace{0.45em}{\rmfamily\bfseries\small\texttt{AL}}
  \par\vspace{0.4mm}
  \resizebox{\linewidth}{!}{%
\begin{tikzpicture}[
    x=0.895cm,y=0.895cm,dporule,
    dpodisc/.style={minimum size=4.11mm},
    dpofont/.style={font=\figatomfont\bfseries\fontsize{6.00}{6.00}\selectfont}]
    \def\dposhorten{2.20mm}
    \draw[rulecard,draw=pbrk!42]          (0.00,0) rectangle (4.20,4.10);
    \draw[rulecard]                        (5.50,0) rectangle (9.70,4.10);
    \draw[rulecard,draw=npgGreen!52!black] (11.00,0) rectangle (15.20,4.10);
    \node[cardlab] at (2.10,3.60) {\(L\)};
    \node[cardlab] at (7.60,3.60) {\(K\)};
    \node[cardlab] at (13.10,3.60) {\(R\)};
    \node[dpoC,dpodisc,dpofont,atomid=1] (L1) at (1.600,1.600) {C};
    \node[dpoO,dpodisc,dpofont,atomid=2] (L2) at (1.600,2.600) {O};
    \node[dpoH,dpodisc,dpofont,atomid=3] (L3) at (0.600,1.600) {H};
    \node[dpoC,dpodisc,dpofont,atomid=4] (L4) at (2.600,1.600) {C};
    \node[dpoH,dpodisc,dpofont,atomid=5] (L5) at (2.600,2.600) {H};
    \node[dpoO,dpodisc,dpofont,atomid=6] (L6) at (2.600,0.600) {O};
    \node[dpoH,dpodisc,dpofont,atomid=7] (L7) at (3.600,0.600) {H};
    \node[dpoC,dpodisc,dpofont,atomid=8] (L8) at (3.600,1.600) {C};
    \node[dpoO,dpodisc,dpofont,atomid=9] (L9) at (3.600,2.600) {O};
    \dbond{L1}{L2}{keepbond}{delbond}
    \draw[keepbond] (L1)--(L3);
    \draw[delbond] (L4)--(L5);
    \draw[keepbond] (L4)--(L6);
    \draw[keepbond] (L4)--(L8);
    \draw[keepbond] (L6)--(L7);
    \dbond{L8}{L9}{keepbond}{keepbond}
    \node[dpoC,dpodisc,dpofont,atomid=1] (K1) at (7.100,1.600) {C};
    \node[dpoO,dpodisc,dpofont,atomid=2] (K2) at (7.100,2.600) {O};
    \node[dpoH,dpodisc,dpofont,atomid=3] (K3) at (6.100,1.600) {H};
    \node[dpoC,dpodisc,dpofont,atomid=4] (K4) at (8.100,1.600) {C};
    \node[dpoH,dpodisc,dpofont,atomid=5] (K5) at (8.100,2.600) {H};
    \node[dpoO,dpodisc,dpofont,atomid=6] (K6) at (8.100,0.600) {O};
    \node[dpoH,dpodisc,dpofont,atomid=7] (K7) at (9.100,0.600) {H};
    \node[dpoC,dpodisc,dpofont,atomid=8] (K8) at (9.100,1.600) {C};
    \node[dpoO,dpodisc,dpofont,atomid=9] (K9) at (9.100,2.600) {O};
    \draw[keepbond] (K1)--(K2);
    \draw[keepbond] (K1)--(K3);
    \draw[keepbond] (K4)--(K6);
    \draw[keepbond] (K4)--(K8);
    \draw[keepbond] (K6)--(K7);
    \dbond{K8}{K9}{keepbond}{keepbond}
    \node[dpoC,dpodisc,dpofont,atomid=1] (R1) at (12.600,1.600) {C};
    \node[dpoO,dpodisc,dpofont,atomid=2] (R2) at (12.600,2.600) {O};
    \node[dpoH,dpodisc,dpofont,atomid=3] (R3) at (11.600,1.600) {H};
    \node[dpoC,dpodisc,dpofont,atomid=4] (R4) at (13.600,1.600) {C};
    \node[dpoH,dpodisc,dpofont,atomid=5] (R5) at (13.600,2.600) {H};
    \node[dpoO,dpodisc,dpofont,atomid=6] (R6) at (13.600,0.600) {O};
    \node[dpoH,dpodisc,dpofont,atomid=7] (R7) at (14.600,0.600) {H};
    \node[dpoC,dpodisc,dpofont,atomid=8] (R8) at (14.600,1.600) {C};
    \node[dpoO,dpodisc,dpofont,atomid=9] (R9) at (14.600,2.600) {O};
    \draw[keepbond] (R1)--(R2);
    \draw[keepbond] (R1)--(R3);
    \draw[newbond] (R1)--(R4);
    \draw[newbond] (R2)--(R5);
    \draw[keepbond] (R4)--(R6);
    \draw[keepbond] (R4)--(R8);
    \draw[keepbond] (R6)--(R7);
    \dbond{R8}{R9}{keepbond}{keepbond}
    \draw[rulearrow] (5.35,2.05)--(4.35,2.05) node[midway,above=2pt,spanlab] {\(l\)};
    \draw[rulearrow] (9.85,2.05)--(10.85,2.05) node[midway,above=2pt,spanlab] {\(r\)};
  \end{tikzpicture}
  }
\end{minipage}\hfill\mbox{}
  \par\vspace{1mm}
  \noindent\hfill\begin{minipage}[t]{0.940\linewidth}
  \raggedright
  \tikz[baseline=-0.65ex]{\node[badge] {B};}\hspace{0.45em}{\rmfamily\bfseries\small\texttt{AlKe}}
  \par\vspace{0.4mm}
  \resizebox{\linewidth}{!}{%
\begin{tikzpicture}[
    x=0.799cm,y=0.799cm,dporule,
    dpodisc/.style={minimum size=4.39mm},
    dpofont/.style={font=\figatomfont\bfseries\fontsize{6.40}{6.40}\selectfont,
      yshift=2.5mm}]
    \def\dposhorten{2.35mm}
    \draw[rulecard,draw=pbrk!42]          (0.00,0.45) rectangle (4.81,4.65);
    \draw[rulecard]                        (6.11,0.45) rectangle (10.92,4.65);
    \draw[rulecard,draw=npgGreen!52!black] (12.22,0.45) rectangle (17.03,4.65);
    \node[cardlab] at (2.41,4.25) {\(L\)};
    \node[cardlab] at (8.52,4.25) {\(K\)};
    \node[cardlab] at (14.62,4.25) {\(R\)};
    \node[dpoC,dpodisc,dpofont,atomidpos={1/225}] (L1) at (2.405,1.600) {C};
    \node[dpoO,dpodisc,dpofont,atomidpos={2/135}] (L2) at (2.405,2.600) {O};
    \node[dpoH,dpodisc,dpofont,atomidpos={3/135}] (L3) at (1.405,1.600) {H};
    \node[dpoC,dpodisc,dpofont,atomidpos={4/315}] (L4) at (3.405,1.600) {C};
    \node[dpoH,dpodisc,dpofont,atomidpos={5/315}] (L5) at (3.405,0.600) {H};
    \node[dpoO,dpodisc,dpofont,atomidpos={6/45}] (L6) at (3.405,2.600) {O};
    \node[dpoH,dpodisc,dpofont,atomidpos={7/45}] (L7) at (3.405,3.600) {H};
    \dbond{L1}{L2}{keepbond}{delbond}
    \draw[keepbond] (L1)--(L3);
    \draw[keepbond] (L1)--(L4);
    \draw[delbond] (L4)--(L5);
    \draw[keepbond] (L4)--(L6);
    \draw[delbond] (L6)--(L7);
    \node[dpoC,dpodisc,dpofont,atomidpos={1/225}] (K1) at (8.515,1.600) {C};
    \node[dpoO,dpodisc,dpofont,atomidpos={2/135}] (K2) at (8.515,2.600) {O};
    \node[dpoH,dpodisc,dpofont,atomidpos={3/135}] (K3) at (7.515,1.600) {H};
    \node[dpoC,dpodisc,dpofont,atomidpos={4/315}] (K4) at (9.515,1.600) {C};
    \node[dpoH,dpodisc,dpofont,atomidpos={5/315}] (K5) at (9.515,0.600) {H};
    \node[dpoO,dpodisc,dpofont,atomidpos={6/45}] (K6) at (9.515,2.600) {O};
    \node[dpoH,dpodisc,dpofont,atomidpos={7/45}] (K7) at (9.515,3.600) {H};
    \draw[keepbond] (K1)--(K2);
    \draw[keepbond] (K1)--(K3);
    \draw[keepbond] (K1)--(K4);
    \draw[keepbond] (K4)--(K6);
    \node[dpoC,dpodisc,dpofont,atomidpos={1/225}] (R1) at (14.625,1.600) {C};
    \node[dpoO,dpodisc,dpofont,atomidpos={2/135}] (R2) at (14.625,2.600) {O};
    \node[dpoH,dpodisc,dpofont,atomidpos={3/135}] (R3) at (13.625,1.600) {H};
    \node[dpoC,dpodisc,dpofont,atomidpos={4/315}] (R4) at (15.625,1.600) {C};
    \node[dpoH,dpodisc,dpofont,atomidpos={5/315}] (R5) at (15.625,0.600) {H};
    \node[dpoO,dpodisc,dpofont,atomidpos={6/45}] (R6) at (15.625,2.600) {O};
    \node[dpoH,dpodisc,dpofont,atomidpos={7/45}] (R7) at (15.625,3.600) {H};
    \draw[keepbond] (R1)--(R2);
    \draw[keepbond] (R1)--(R3);
    \draw[keepbond] (R1)--(R4);
    \draw[newbond] (R1)--(R5);
    \draw[newbond] (R2)--(R7);
    \dbond{R4}{R6}{keepbond}{newbond}
    \draw[rulearrow] (5.96,2.55)--(4.96,2.55) node[midway,above=2pt,spanlab] {\(l\)};
    \draw[rulearrow] (11.07,2.55)--(12.07,2.55) node[midway,above=2pt,spanlab] {\(r\)};
  \end{tikzpicture}
  }
\end{minipage}\hfill\mbox{}
  \par\vspace{1mm}
  \noindent\hfill\begin{minipage}[t]{0.940\linewidth}
  \raggedright
  \tikz[baseline=-0.65ex]{\node[badge] {C};}\hspace{0.45em}{\rmfamily\bfseries\small\texttt{KeAl}}
  \par\vspace{0.4mm}
  \resizebox{\linewidth}{!}{%
\begin{tikzpicture}[
    x=0.799cm,y=0.799cm,dporule,
    dpodisc/.style={minimum size=4.39mm},
    dpofont/.style={font=\figatomfont\bfseries\fontsize{6.40}{6.40}\selectfont,
      yshift=2.5mm}]
    \def\dposhorten{2.35mm}
    \draw[rulecard,draw=pbrk!42]          (0.00,0.45) rectangle (4.81,4.65);
    \draw[rulecard]                        (6.11,0.45) rectangle (10.92,4.65);
    \draw[rulecard,draw=npgGreen!52!black] (12.22,0.45) rectangle (17.03,4.65);
    \node[cardlab] at (2.41,4.25) {\(L\)};
    \node[cardlab] at (8.52,4.25) {\(K\)};
    \node[cardlab] at (14.62,4.25) {\(R\)};
    \node[dpoC,dpodisc,dpofont,atomidpos={1/225}] (L1) at (2.405,1.600) {C};
    \node[dpoO,dpodisc,dpofont,atomidpos={2/135}] (L2) at (2.405,2.600) {O};
    \node[dpoH,dpodisc,dpofont,atomidpos={3/135}] (L3) at (1.405,1.600) {H};
    \node[dpoC,dpodisc,dpofont,atomidpos={4/315}] (L4) at (3.405,1.600) {C};
    \node[dpoH,dpodisc,dpofont,atomidpos={5/315}] (L5) at (3.405,0.600) {H};
    \node[dpoO,dpodisc,dpofont,atomidpos={6/45}] (L6) at (3.405,2.600) {O};
    \node[dpoH,dpodisc,dpofont,atomidpos={7/45}] (L7) at (3.405,3.600) {H};
    \draw[keepbond] (L1)--(L2);
    \draw[keepbond] (L1)--(L3);
    \draw[keepbond] (L1)--(L4);
    \draw[delbond] (L1)--(L5);
    \draw[delbond] (L2)--(L7);
    \dbond{L4}{L6}{keepbond}{delbond}
    \node[dpoC,dpodisc,dpofont,atomidpos={1/225}] (K1) at (8.515,1.600) {C};
    \node[dpoO,dpodisc,dpofont,atomidpos={2/135}] (K2) at (8.515,2.600) {O};
    \node[dpoH,dpodisc,dpofont,atomidpos={3/135}] (K3) at (7.515,1.600) {H};
    \node[dpoC,dpodisc,dpofont,atomidpos={4/315}] (K4) at (9.515,1.600) {C};
    \node[dpoH,dpodisc,dpofont,atomidpos={5/315}] (K5) at (9.515,0.600) {H};
    \node[dpoO,dpodisc,dpofont,atomidpos={6/45}] (K6) at (9.515,2.600) {O};
    \node[dpoH,dpodisc,dpofont,atomidpos={7/45}] (K7) at (9.515,3.600) {H};
    \draw[keepbond] (K1)--(K2);
    \draw[keepbond] (K1)--(K3);
    \draw[keepbond] (K1)--(K4);
    \draw[keepbond] (K4)--(K6);
    \node[dpoC,dpodisc,dpofont,atomidpos={1/225}] (R1) at (14.625,1.600) {C};
    \node[dpoO,dpodisc,dpofont,atomidpos={2/135}] (R2) at (14.625,2.600) {O};
    \node[dpoH,dpodisc,dpofont,atomidpos={3/135}] (R3) at (13.625,1.600) {H};
    \node[dpoC,dpodisc,dpofont,atomidpos={4/315}] (R4) at (15.625,1.600) {C};
    \node[dpoH,dpodisc,dpofont,atomidpos={5/315}] (R5) at (15.625,0.600) {H};
    \node[dpoO,dpodisc,dpofont,atomidpos={6/45}] (R6) at (15.625,2.600) {O};
    \node[dpoH,dpodisc,dpofont,atomidpos={7/45}] (R7) at (15.625,3.600) {H};
    \dbond{R1}{R2}{keepbond}{newbond}
    \draw[keepbond] (R1)--(R3);
    \draw[keepbond] (R1)--(R4);
    \draw[newbond] (R4)--(R5);
    \draw[keepbond] (R4)--(R6);
    \draw[newbond] (R6)--(R7);
    \draw[rulearrow] (5.96,2.55)--(4.96,2.55) node[midway,above=2pt,spanlab] {\(l\)};
    \draw[rulearrow] (11.07,2.55)--(12.07,2.55) node[midway,above=2pt,spanlab] {\(r\)};
  \end{tikzpicture}
  }
\end{minipage}\hfill\mbox{}
}

\newcommand{\nogdpopanelD}{%
  \noindent\hfill\begin{minipage}[t]{0.940\linewidth}
  \raggedright
  \tikz[baseline=-0.65ex]{\node[badge] {D};}\hspace{0.45em}{\rmfamily\bfseries\small\texttt{TKL}}
  \par\vspace{0.4mm}
  \resizebox{\linewidth}{!}{%
\begin{tikzpicture}[
    x=0.799cm,y=0.799cm,dporule,
    dpodisc/.style={minimum size=4.39mm},
    dpofont/.style={font=\figatomfont\bfseries\fontsize{6.40}{6.40}\selectfont,
      yshift=2.5mm}]
    \def\dposhorten{2.35mm}
    \draw[rulecard,draw=pbrk!42]          (0.00,0.35) rectangle (4.81,5.60);
    \draw[rulecard]                        (6.11,0.35) rectangle (10.92,5.60);
    \draw[rulecard,draw=npgGreen!52!black] (12.22,0.35) rectangle (17.03,5.60);
    \node[cardlab] at (2.41,5.10) {\(L\)};
    \node[cardlab] at (8.52,5.10) {\(K\)};
    \node[cardlab] at (14.62,5.10) {\(R\)};
    \node[dpoC,dpodisc,dpofont,atomid=1] (L1) at (2.340,1.600) {C};
    \node[dpoO,dpodisc,dpofont,atomid=2] (L2) at (3.210,1.100) {O};
    \node[dpoH,dpodisc,dpofont,atomidpos={3/225}] (L3) at (2.340,0.600) {H};
    \node[dpoC,dpodisc,dpofont,atomid=4] (L4) at (1.470,3.100) {C};
    \node[dpoO,dpodisc,dpofont,atomid=5] (L5) at (1.470,4.100) {O};
    \node[dpoH,dpodisc,dpofont,atomid=6] (L6) at (0.600,3.600) {H};
    \node[dpoH,dpodisc,dpofont,atomid=7] (L7) at (0.600,2.600) {H};
    \node[dpoH,dpodisc,dpofont,atomid=8] (L8) at (1.470,2.100) {H};
    \node[dpoC,dpodisc,dpofont,atomid=9] (L9) at (2.340,2.600) {C};
    \node[dpoO,dpodisc,dpofont,atomid=10] (L10) at (2.340,3.600) {O};
    \node[dpoC,dpodisc,dpofont,atomid=11] (L11) at (3.340,2.600) {C};
    \node[dpoH,dpodisc,dpofont,atomid=12] (L12) at (3.340,3.600) {H};
    \node[dpoO,dpodisc,dpofont,atomid=13] (L13) at (4.210,2.100) {O};
    \node[dpoH,dpodisc,dpofont,atomid=14] (L14) at (4.090,1.100) {H};
    \dbond{L1}{L2}{keepbond}{delbond}
    \draw[keepbond] (L1)--(L3);
    \draw[keepbond] (L4)--(L5);
    \draw[keepbond] (L4)--(L7);
    \draw[keepbond] (L4)--(L8);
    \draw[keepbond] (L4)--(L9);
    \draw[keepbond] (L5)--(L6);
    \dbond{L9}{L10}{keepbond}{keepbond}
    \draw[delbond] (L9)--(L11);
    \draw[keepbond] (L11)--(L12);
    \draw[keepbond] (L11)--(L13);
    \draw[delbond] (L13)--(L14);
    \node[dpoC,dpodisc,dpofont,atomid=1] (K1) at (8.450,1.600) {C};
    \node[dpoO,dpodisc,dpofont,atomid=2] (K2) at (9.320,1.100) {O};
    \node[dpoH,dpodisc,dpofont,atomidpos={3/225}] (K3) at (8.450,0.600) {H};
    \node[dpoC,dpodisc,dpofont,atomid=4] (K4) at (7.580,3.100) {C};
    \node[dpoO,dpodisc,dpofont,atomid=5] (K5) at (7.580,4.100) {O};
    \node[dpoH,dpodisc,dpofont,atomid=6] (K6) at (6.710,3.600) {H};
    \node[dpoH,dpodisc,dpofont,atomid=7] (K7) at (6.710,2.600) {H};
    \node[dpoH,dpodisc,dpofont,atomid=8] (K8) at (7.580,2.100) {H};
    \node[dpoC,dpodisc,dpofont,atomid=9] (K9) at (8.450,2.600) {C};
    \node[dpoO,dpodisc,dpofont,atomid=10] (K10) at (8.450,3.600) {O};
    \node[dpoC,dpodisc,dpofont,atomid=11] (K11) at (9.450,2.600) {C};
    \node[dpoH,dpodisc,dpofont,atomid=12] (K12) at (9.450,3.600) {H};
    \node[dpoO,dpodisc,dpofont,atomid=13] (K13) at (10.320,2.100) {O};
    \node[dpoH,dpodisc,dpofont,atomid=14] (K14) at (10.200,1.100) {H};
    \draw[keepbond] (K1)--(K2);
    \draw[keepbond] (K1)--(K3);
    \draw[keepbond] (K4)--(K5);
    \draw[keepbond] (K4)--(K7);
    \draw[keepbond] (K4)--(K8);
    \draw[keepbond] (K4)--(K9);
    \draw[keepbond] (K5)--(K6);
    \dbond{K9}{K10}{keepbond}{keepbond}
    \draw[keepbond] (K11)--(K12);
    \draw[keepbond] (K11)--(K13);
    \node[dpoC,dpodisc,dpofont,atomid=1] (R1) at (14.560,1.600) {C};
    \node[dpoO,dpodisc,dpofont,atomid=2] (R2) at (15.430,1.100) {O};
    \node[dpoH,dpodisc,dpofont,atomidpos={3/225}] (R3) at (14.560,0.600) {H};
    \node[dpoC,dpodisc,dpofont,atomid=4] (R4) at (13.690,3.100) {C};
    \node[dpoO,dpodisc,dpofont,atomid=5] (R5) at (13.690,4.100) {O};
    \node[dpoH,dpodisc,dpofont,atomid=6] (R6) at (12.820,3.600) {H};
    \node[dpoH,dpodisc,dpofont,atomid=7] (R7) at (12.820,2.600) {H};
    \node[dpoH,dpodisc,dpofont,atomid=8] (R8) at (13.690,2.100) {H};
    \node[dpoC,dpodisc,dpofont,atomid=9] (R9) at (14.560,2.600) {C};
    \node[dpoO,dpodisc,dpofont,atomid=10] (R10) at (14.560,3.600) {O};
    \node[dpoC,dpodisc,dpofont,atomid=11] (R11) at (15.560,2.600) {C};
    \node[dpoH,dpodisc,dpofont,atomid=12] (R12) at (15.560,3.600) {H};
    \node[dpoO,dpodisc,dpofont,atomid=13] (R13) at (16.430,2.100) {O};
    \node[dpoH,dpodisc,dpofont,atomid=14] (R14) at (16.310,1.100) {H};
    \draw[keepbond] (R1)--(R2);
    \draw[keepbond] (R1)--(R3);
    \draw[newbond] (R1)--(R9);
    \draw[newbond] (R2)--(R14);
    \draw[keepbond] (R4)--(R5);
    \draw[keepbond] (R4)--(R7);
    \draw[keepbond] (R4)--(R8);
    \draw[keepbond] (R4)--(R9);
    \draw[keepbond] (R5)--(R6);
    \dbond{R9}{R10}{keepbond}{keepbond}
    \draw[keepbond] (R11)--(R12);
    \dbond{R11}{R13}{keepbond}{newbond}
    \draw[rulearrow] (5.96,2.80)--(4.96,2.80) node[midway,above=2pt,spanlab] {\(l\)};
    \draw[rulearrow] (11.07,2.80)--(12.07,2.80) node[midway,above=2pt,spanlab] {\(r\)};
  \end{tikzpicture}
  }
\end{minipage}\hfill\mbox{}
}

\newcommand{\nogdpopanelE}{%
  \noindent\hfill\begin{minipage}[t]{0.940\linewidth}
  \raggedright
  \tikz[baseline=-0.65ex]{\node[badge] {E};}\hspace{0.45em}{\rmfamily\bfseries\small\texttt{TAL}}
  \par\vspace{0.4mm}
  \resizebox{\linewidth}{!}{%
\begin{tikzpicture}[
    x=0.676cm,y=0.676cm,dporule,
    dpodisc/.style={minimum size=3.72mm},
    dpofont/.style={font=\figatomfont\bfseries\fontsize{5.40}{5.40}\selectfont}]
    \def\dposhorten{2.01mm}
    \draw[rulecard,draw=pbrk!42]          (0.00,-0.25) rectangle (5.84,5.20);
    \draw[rulecard]                        (7.14,-0.25) rectangle (12.98,5.20);
    \draw[rulecard,draw=npgGreen!52!black] (14.28,-0.25) rectangle (20.12,5.20);
    \node[cardlab] at (2.92,4.70) {\(L\)};
    \node[cardlab] at (10.06,4.70) {\(K\)};
    \node[cardlab] at (17.20,4.70) {\(R\)};
    \node[dpoC,dpodisc,dpofont,atomidpos={1/45}] (L1) at (4.210,1.700) {C};
    \node[dpoO,dpodisc,dpofont,atomidpos={2/315}] (L2) at (5.080,1.200) {O};
    \node[dpoH,dpodisc,dpofont,atomidpos={3/225}] (L3) at (4.210,0.700) {H};
    \node[dpoC,dpodisc,dpofont,atomidpos={4/180}] (L4) at (1.470,2.700) {C};
    \node[dpoO,dpodisc,dpofont,atomidpos={5/45}] (L5) at (1.470,3.700) {O};
    \node[dpoH,dpodisc,dpofont,atomidpos={6/135}] (L6) at (0.600,3.200) {H};
    \node[dpoH,dpodisc,dpofont,atomidpos={7/225}] (L7) at (0.600,2.200) {H};
    \node[dpoH,dpodisc,dpofont,atomidpos={8/225}] (L8) at (1.470,1.700) {H};
    \node[dpoC,dpodisc,dpofont,atomidpos={9/270}] (L9) at (2.340,2.200) {C};
    \node[dpoO,dpodisc,dpofont,atomidpos={10/135}] (L10) at (2.340,3.200) {O};
    \node[dpoC,dpodisc,dpofont,atomidpos={11/225}] (L11) at (3.340,2.200) {C};
    \node[dpoH,dpodisc,dpofont,atomidpos={12/135}] (L12) at (3.340,3.200) {H};
    \node[dpoO,dpodisc,dpofont,atomidpos={13/315}] (L13) at (3.340,1.200) {O};
    \node[dpoH,dpodisc,dpofont,atomidpos={14/225}] (L14) at (2.540,0.600) {H};
    \node[dpoC,dpodisc,dpofont,atomidpos={15/135}] (L15) at (4.210,2.700) {C};
    \node[dpoH,dpodisc,dpofont,atomidpos={16/135}] (L16) at (4.210,3.700) {H};
    \node[dpoO,dpodisc,dpofont,atomidpos={17/45}] (L17) at (4.940,3.300) {O};
    \node[dpoH,dpodisc,dpofont,atomidpos={18/0}] (L18) at (5.240,2.400) {H};
    \dbond{L1}{L2}{keepbond}{delbond}
    \draw[keepbond] (L1)--(L3);
    \draw[keepbond] (L4)--(L5);
    \draw[keepbond] (L4)--(L7);
    \draw[keepbond] (L4)--(L8);
    \draw[keepbond] (L4)--(L9);
    \draw[keepbond] (L5)--(L6);
    \dbond{L9}{L10}{keepbond}{keepbond}
    \draw[keepbond] (L9)--(L11);
    \draw[keepbond] (L11)--(L12);
    \draw[keepbond] (L11)--(L13);
    \draw[delbond] (L11)--(L15);
    \draw[keepbond] (L13)--(L14);
    \draw[keepbond] (L15)--(L16);
    \draw[keepbond] (L15)--(L17);
    \draw[delbond] (L17)--(L18);
    \node[dpoC,dpodisc,dpofont,atomidpos={1/45}] (K1) at (11.350,1.700) {C};
    \node[dpoO,dpodisc,dpofont,atomidpos={2/315}] (K2) at (12.220,1.200) {O};
    \node[dpoH,dpodisc,dpofont,atomidpos={3/225}] (K3) at (11.350,0.700) {H};
    \node[dpoC,dpodisc,dpofont,atomidpos={4/180}] (K4) at (8.610,2.700) {C};
    \node[dpoO,dpodisc,dpofont,atomidpos={5/45}] (K5) at (8.610,3.700) {O};
    \node[dpoH,dpodisc,dpofont,atomidpos={6/135}] (K6) at (7.740,3.200) {H};
    \node[dpoH,dpodisc,dpofont,atomidpos={7/225}] (K7) at (7.740,2.200) {H};
    \node[dpoH,dpodisc,dpofont,atomidpos={8/225}] (K8) at (8.610,1.700) {H};
    \node[dpoC,dpodisc,dpofont,atomidpos={9/270}] (K9) at (9.480,2.200) {C};
    \node[dpoO,dpodisc,dpofont,atomidpos={10/135}] (K10) at (9.480,3.200) {O};
    \node[dpoC,dpodisc,dpofont,atomidpos={11/225}] (K11) at (10.480,2.200) {C};
    \node[dpoH,dpodisc,dpofont,atomidpos={12/135}] (K12) at (10.480,3.200) {H};
    \node[dpoO,dpodisc,dpofont,atomidpos={13/315}] (K13) at (10.480,1.200) {O};
    \node[dpoH,dpodisc,dpofont,atomidpos={14/225}] (K14) at (9.680,0.600) {H};
    \node[dpoC,dpodisc,dpofont,atomidpos={15/135}] (K15) at (11.350,2.700) {C};
    \node[dpoH,dpodisc,dpofont,atomidpos={16/135}] (K16) at (11.350,3.700) {H};
    \node[dpoO,dpodisc,dpofont,atomidpos={17/45}] (K17) at (12.080,3.300) {O};
    \node[dpoH,dpodisc,dpofont,atomidpos={18/0}] (K18) at (12.380,2.400) {H};
    \draw[keepbond] (K1)--(K2);
    \draw[keepbond] (K1)--(K3);
    \draw[keepbond] (K4)--(K5);
    \draw[keepbond] (K4)--(K7);
    \draw[keepbond] (K4)--(K8);
    \draw[keepbond] (K4)--(K9);
    \draw[keepbond] (K5)--(K6);
    \dbond{K9}{K10}{keepbond}{keepbond}
    \draw[keepbond] (K9)--(K11);
    \draw[keepbond] (K11)--(K12);
    \draw[keepbond] (K11)--(K13);
    \draw[keepbond] (K13)--(K14);
    \draw[keepbond] (K15)--(K16);
    \draw[keepbond] (K15)--(K17);
    \node[dpoC,dpodisc,dpofont,atomidpos={1/45}] (R1) at (18.490,1.700) {C};
    \node[dpoO,dpodisc,dpofont,atomidpos={2/315}] (R2) at (19.360,1.200) {O};
    \node[dpoH,dpodisc,dpofont,atomidpos={3/225}] (R3) at (18.490,0.700) {H};
    \node[dpoC,dpodisc,dpofont,atomidpos={4/180}] (R4) at (15.750,2.700) {C};
    \node[dpoO,dpodisc,dpofont,atomidpos={5/45}] (R5) at (15.750,3.700) {O};
    \node[dpoH,dpodisc,dpofont,atomidpos={6/135}] (R6) at (14.880,3.200) {H};
    \node[dpoH,dpodisc,dpofont,atomidpos={7/225}] (R7) at (14.880,2.200) {H};
    \node[dpoH,dpodisc,dpofont,atomidpos={8/225}] (R8) at (15.750,1.700) {H};
    \node[dpoC,dpodisc,dpofont,atomidpos={9/270}] (R9) at (16.620,2.200) {C};
    \node[dpoO,dpodisc,dpofont,atomidpos={10/135}] (R10) at (16.620,3.200) {O};
    \node[dpoC,dpodisc,dpofont,atomidpos={11/225}] (R11) at (17.620,2.200) {C};
    \node[dpoH,dpodisc,dpofont,atomidpos={12/135}] (R12) at (17.620,3.200) {H};
    \node[dpoO,dpodisc,dpofont,atomidpos={13/315}] (R13) at (17.620,1.200) {O};
    \node[dpoH,dpodisc,dpofont,atomidpos={14/225}] (R14) at (16.820,0.600) {H};
    \node[dpoC,dpodisc,dpofont,atomidpos={15/135}] (R15) at (18.490,2.700) {C};
    \node[dpoH,dpodisc,dpofont,atomidpos={16/135}] (R16) at (18.490,3.700) {H};
    \node[dpoO,dpodisc,dpofont,atomidpos={17/45}] (R17) at (19.220,3.300) {O};
    \node[dpoH,dpodisc,dpofont,atomidpos={18/0}] (R18) at (19.520,2.400) {H};
    \draw[keepbond] (R1)--(R2);
    \draw[keepbond] (R1)--(R3);
    \draw[newbond] (R1)--(R11);
    \draw[newbond] (R2)--(R18);
    \draw[keepbond] (R4)--(R5);
    \draw[keepbond] (R4)--(R7);
    \draw[keepbond] (R4)--(R8);
    \draw[keepbond] (R4)--(R9);
    \draw[keepbond] (R5)--(R6);
    \dbond{R9}{R10}{keepbond}{keepbond}
    \draw[keepbond] (R9)--(R11);
    \draw[keepbond] (R11)--(R12);
    \draw[keepbond] (R11)--(R13);
    \draw[keepbond] (R13)--(R14);
    \draw[keepbond] (R15)--(R16);
    \dbond{R15}{R17}{keepbond}{newbond}
    \draw[rulearrow] (6.99,2.60)--(5.99,2.60) node[midway,above=2pt,spanlab] {\(l\)};
    \draw[rulearrow] (13.13,2.60)--(14.13,2.60) node[midway,above=2pt,spanlab] {\(r\)};
  \end{tikzpicture}
  }
\end{minipage}\hfill\mbox{}
}

\newcommand{\nogdpopanelF}{%
  \noindent\hfill\begin{minipage}[t]{0.940\linewidth}
  \raggedright
  \tikz[baseline=-0.65ex]{\node[badge] {F};}\hspace{0.45em}{\rmfamily\bfseries\small\texttt{PK}}
  \par\vspace{0.4mm}
  \resizebox{\linewidth}{!}{%
\begin{tikzpicture}[
    x=0.719cm,y=0.719cm,dporule,
    dpodisc/.style={minimum size=3.95mm},
    dpofont/.style={font=\figatomfont\bfseries\fontsize{5.70}{5.70}\selectfont}]
    \def\dposhorten{2.13mm}
    \draw[rulecard,draw=pbrk!42]          (0.00,-0.25) rectangle (5.44,7.15);
    \draw[rulecard]                        (6.74,-0.25) rectangle (12.18,7.15);
    \draw[rulecard,draw=npgGreen!52!black] (13.48,-0.25) rectangle (18.92,7.15);
    \node[cardlab] at (2.72,6.65) {\(L\)};
    \node[cardlab] at (9.46,6.65) {\(K\)};
    \node[cardlab] at (16.20,6.65) {\(R\)};
    \node[dpoC,dpodisc,dpofont,atomid=1] (L1) at (1.600,4.550) {C};
    \node[dpoH,dpodisc,dpofont,atomid=2] (L2) at (0.600,4.550) {H};
    \node[dpoO,dpodisc,dpofont,atomid=3] (L3) at (1.600,5.550) {O};
    \node[dpoH,dpodisc,dpofont,atomid=4] (L4) at (2.600,5.650) {H};
    \node[dpoC,dpodisc,dpofont,atomidpos={5/90}] (L5) at (2.470,4.050) {C};
    \node[dpoO,dpodisc,dpofont,atomidpos={6/225}] (L6) at (1.970,3.180) {O};
    \node[dpoC,dpodisc,dpofont,atomid=7] (L7) at (3.340,4.550) {C};
    \node[dpoH,dpodisc,dpofont,atomid=8] (L8) at (4.200,5.050) {H};
    \node[dpoO,dpodisc,dpofont,atomid=9] (L9) at (3.840,3.680) {O};
    \node[dpoH,dpodisc,dpofont,atomid=10] (L10) at (4.840,3.680) {H};
    \node[dpoP,dpodisc,dpofont,atomidpos={11/135}] (L11) at (2.500,2.300) {P};
    \node[dpoO,dpodisc,dpofont,atomid=12] (L12) at (2.970,3.180) {O};
    \node[dpoH,dpodisc,dpofont,atomid=13] (L13) at (3.700,2.700) {H};
    \node[dpoO,dpodisc,dpofont,atomidpos={14/225}] (L14) at (1.600,1.850) {O};
    \node[dpoO,dpodisc,dpofont,atomid=15] (L15) at (2.500,1.300) {O};
    \node[dpoH,dpodisc,dpofont,atomid=16] (L16) at (3.200,0.600) {H};
    \node[dpoO,dpodisc,dpofont,atomid=17] (L17) at (3.400,1.850) {O};
    \node[dpoH,dpodisc,dpofont,atomid=18] (L18) at (4.200,1.250) {H};
    \draw[keepbond] (L1)--(L2);
    \draw[keepbond] (L1)--(L3);
    \draw[delbond] (L1)--(L5);
    \draw[delbond] (L3)--(L4);
    \dbond{L5}{L6}{keepbond}{keepbond}
    \draw[keepbond] (L5)--(L7);
    \draw[keepbond] (L7)--(L8);
    \draw[delbond] (L7)--(L9);
    \draw[keepbond] (L9)--(L10);
    \draw[keepbond] (L11)--(L12);
    \dbond{L11}{L14}{keepbond}{keepbond}
    \draw[keepbond] (L11)--(L15);
    \draw[keepbond] (L11)--(L17);
    \draw[delbond] (L12)--(L13);
    \draw[keepbond] (L15)--(L16);
    \draw[keepbond] (L17)--(L18);
    \node[dpoC,dpodisc,dpofont,atomid=1] (K1) at (8.340,4.550) {C};
    \node[dpoH,dpodisc,dpofont,atomid=2] (K2) at (7.340,4.550) {H};
    \node[dpoO,dpodisc,dpofont,atomid=3] (K3) at (8.340,5.550) {O};
    \node[dpoH,dpodisc,dpofont,atomid=4] (K4) at (9.340,5.650) {H};
    \node[dpoC,dpodisc,dpofont,atomidpos={5/90}] (K5) at (9.210,4.050) {C};
    \node[dpoO,dpodisc,dpofont,atomidpos={6/225}] (K6) at (8.710,3.180) {O};
    \node[dpoC,dpodisc,dpofont,atomid=7] (K7) at (10.080,4.550) {C};
    \node[dpoH,dpodisc,dpofont,atomid=8] (K8) at (10.940,5.050) {H};
    \node[dpoO,dpodisc,dpofont,atomid=9] (K9) at (10.580,3.680) {O};
    \node[dpoH,dpodisc,dpofont,atomid=10] (K10) at (11.580,3.680) {H};
    \node[dpoP,dpodisc,dpofont,atomidpos={11/135}] (K11) at (9.240,2.300) {P};
    \node[dpoO,dpodisc,dpofont,atomid=12] (K12) at (9.710,3.180) {O};
    \node[dpoH,dpodisc,dpofont,atomid=13] (K13) at (10.440,2.700) {H};
    \node[dpoO,dpodisc,dpofont,atomidpos={14/225}] (K14) at (8.340,1.850) {O};
    \node[dpoO,dpodisc,dpofont,atomid=15] (K15) at (9.240,1.300) {O};
    \node[dpoH,dpodisc,dpofont,atomid=16] (K16) at (9.940,0.600) {H};
    \node[dpoO,dpodisc,dpofont,atomid=17] (K17) at (10.140,1.850) {O};
    \node[dpoH,dpodisc,dpofont,atomid=18] (K18) at (10.940,1.250) {H};
    \draw[keepbond] (K1)--(K2);
    \draw[keepbond] (K1)--(K3);
    \dbond{K5}{K6}{keepbond}{keepbond}
    \draw[keepbond] (K5)--(K7);
    \draw[keepbond] (K7)--(K8);
    \draw[keepbond] (K9)--(K10);
    \draw[keepbond] (K11)--(K12);
    \dbond{K11}{K14}{keepbond}{keepbond}
    \draw[keepbond] (K11)--(K15);
    \draw[keepbond] (K11)--(K17);
    \draw[keepbond] (K15)--(K16);
    \draw[keepbond] (K17)--(K18);
    \node[dpoC,dpodisc,dpofont,atomid=1] (R1) at (15.080,4.550) {C};
    \node[dpoH,dpodisc,dpofont,atomid=2] (R2) at (14.080,4.550) {H};
    \node[dpoO,dpodisc,dpofont,atomid=3] (R3) at (15.080,5.550) {O};
    \node[dpoH,dpodisc,dpofont,atomid=4] (R4) at (16.080,5.650) {H};
    \node[dpoC,dpodisc,dpofont,atomidpos={5/90}] (R5) at (15.950,4.050) {C};
    \node[dpoO,dpodisc,dpofont,atomidpos={6/225}] (R6) at (15.450,3.180) {O};
    \node[dpoC,dpodisc,dpofont,atomid=7] (R7) at (16.820,4.550) {C};
    \node[dpoH,dpodisc,dpofont,atomid=8] (R8) at (17.680,5.050) {H};
    \node[dpoO,dpodisc,dpofont,atomid=9] (R9) at (17.320,3.680) {O};
    \node[dpoH,dpodisc,dpofont,atomid=10] (R10) at (18.320,3.680) {H};
    \node[dpoP,dpodisc,dpofont,atomidpos={11/135}] (R11) at (15.980,2.300) {P};
    \node[dpoO,dpodisc,dpofont,atomid=12] (R12) at (16.450,3.180) {O};
    \node[dpoH,dpodisc,dpofont,atomid=13] (R13) at (17.180,2.700) {H};
    \node[dpoO,dpodisc,dpofont,atomidpos={14/225}] (R14) at (15.080,1.850) {O};
    \node[dpoO,dpodisc,dpofont,atomid=15] (R15) at (15.980,1.300) {O};
    \node[dpoH,dpodisc,dpofont,atomid=16] (R16) at (16.680,0.600) {H};
    \node[dpoO,dpodisc,dpofont,atomid=17] (R17) at (16.880,1.850) {O};
    \node[dpoH,dpodisc,dpofont,atomid=18] (R18) at (17.680,1.250) {H};
    \draw[keepbond] (R1)--(R2);
    \dbond{R1}{R3}{keepbond}{newbond}
    \draw[newbond] (R4)--(R7);
    \dbond{R5}{R6}{keepbond}{keepbond}
    \draw[keepbond] (R5)--(R7);
    \draw[newbond] (R5)--(R12);
    \draw[keepbond] (R7)--(R8);
    \draw[keepbond] (R9)--(R10);
    \draw[newbond] (R9)--(R13);
    \draw[keepbond] (R11)--(R12);
    \dbond{R11}{R14}{keepbond}{keepbond}
    \draw[keepbond] (R11)--(R15);
    \draw[keepbond] (R11)--(R17);
    \draw[keepbond] (R15)--(R16);
    \draw[keepbond] (R17)--(R18);
    \draw[rulearrow] (6.59,3.58)--(5.59,3.58) node[midway,above=2pt,spanlab] {\(l\)};
    \draw[rulearrow] (12.33,3.58)--(13.33,3.58) node[midway,above=2pt,spanlab] {\(r\)};
  \end{tikzpicture}
  }
\end{minipage}\hfill\mbox{}
}

\newcommand{\nogdpopanelG}{%
  \noindent\hfill\begin{minipage}[t]{0.940\linewidth}
  \raggedright
  \tikz[baseline=-0.65ex]{\node[badge] {G};}\hspace{0.45em}{\rmfamily\bfseries\small\texttt{PHL}}
  \par\vspace{0.4mm}
  \resizebox{\linewidth}{!}{%
\begin{tikzpicture}[
    x=0.747cm,y=0.747cm,dporule,
    dpodisc/.style={minimum size=4.11mm},
    dpofont/.style={font=\figatomfont\bfseries\fontsize{6.00}{6.00}\selectfont}]
    \def\dposhorten{2.20mm}
    \draw[rulecard,draw=pbrk!42]          (0.00,-0.25) rectangle (5.20,5.40);
    \draw[rulecard]                        (6.50,-0.25) rectangle (11.70,5.40);
    \draw[rulecard,draw=npgGreen!52!black] (13.00,-0.25) rectangle (18.20,5.40);
    \node[cardlab] at (2.60,4.90) {\(L\)};
    \node[cardlab] at (9.10,4.90) {\(K\)};
    \node[cardlab] at (15.60,4.90) {\(R\)};
    \node[dpoC,dpodisc,dpofont,atomidpos={1/180}] (L1) at (0.600,2.600) {C};
    \node[dpoO,dpodisc,dpofont,atomidpos={2/90}] (L2) at (1.600,2.600) {O};
    \node[dpoP,dpodisc,dpofont,atomidpos={3/225}] (L3) at (2.600,2.600) {P};
    \node[dpoO,dpodisc,dpofont,atomidpos={4/135}] (L4) at (2.100,3.470) {O};
    \node[dpoO,dpodisc,dpofont,atomidpos={5/135}] (L5) at (3.300,3.300) {O};
    \node[dpoH,dpodisc,dpofont,atomidpos={6/45}] (L6) at (4.100,3.900) {H};
    \node[dpoO,dpodisc,dpofont,atomidpos={7/0}] (L7) at (2.600,1.600) {O};
    \node[dpoH,dpodisc,dpofont,atomidpos={8/225}] (L8) at (1.600,1.600) {H};
    \node[dpoO,dpodisc,dpofont,atomidpos={9/270}] (L9) at (3.600,2.350) {O};
    \node[dpoH,dpodisc,dpofont,atomidpos={10/0}] (L10) at (4.600,2.350) {H};
    \node[dpoH,dpodisc,dpofont,atomidpos={11/270}] (L11) at (2.600,0.600) {H};
    \draw[keepbond] (L1)--(L2);
    \draw[delbond] (L2)--(L3);
    \dbond{L3}{L4}{keepbond}{keepbond}
    \draw[keepbond] (L3)--(L5);
    \draw[keepbond] (L3)--(L9);
    \draw[keepbond] (L5)--(L6);
    \draw[delbond] (L7)--(L8);
    \draw[keepbond] (L7)--(L11);
    \draw[keepbond] (L9)--(L10);
    \node[dpoC,dpodisc,dpofont,atomidpos={1/180}] (K1) at (7.100,2.600) {C};
    \node[dpoO,dpodisc,dpofont,atomidpos={2/90}] (K2) at (8.100,2.600) {O};
    \node[dpoP,dpodisc,dpofont,atomidpos={3/225}] (K3) at (9.100,2.600) {P};
    \node[dpoO,dpodisc,dpofont,atomidpos={4/135}] (K4) at (8.600,3.470) {O};
    \node[dpoO,dpodisc,dpofont,atomidpos={5/135}] (K5) at (9.800,3.300) {O};
    \node[dpoH,dpodisc,dpofont,atomidpos={6/45}] (K6) at (10.600,3.900) {H};
    \node[dpoO,dpodisc,dpofont,atomidpos={7/0}] (K7) at (9.100,1.600) {O};
    \node[dpoH,dpodisc,dpofont,atomidpos={8/225}] (K8) at (8.100,1.600) {H};
    \node[dpoO,dpodisc,dpofont,atomidpos={9/270}] (K9) at (10.100,2.350) {O};
    \node[dpoH,dpodisc,dpofont,atomidpos={10/0}] (K10) at (11.100,2.350) {H};
    \node[dpoH,dpodisc,dpofont,atomidpos={11/270}] (K11) at (9.100,0.600) {H};
    \draw[keepbond] (K1)--(K2);
    \dbond{K3}{K4}{keepbond}{keepbond}
    \draw[keepbond] (K3)--(K5);
    \draw[keepbond] (K3)--(K9);
    \draw[keepbond] (K5)--(K6);
    \draw[keepbond] (K7)--(K11);
    \draw[keepbond] (K9)--(K10);
    \node[dpoC,dpodisc,dpofont,atomidpos={1/180}] (R1) at (13.600,2.600) {C};
    \node[dpoO,dpodisc,dpofont,atomidpos={2/90}] (R2) at (14.600,2.600) {O};
    \node[dpoP,dpodisc,dpofont,atomidpos={3/225}] (R3) at (15.600,2.600) {P};
    \node[dpoO,dpodisc,dpofont,atomidpos={4/135}] (R4) at (15.100,3.470) {O};
    \node[dpoO,dpodisc,dpofont,atomidpos={5/135}] (R5) at (16.300,3.300) {O};
    \node[dpoH,dpodisc,dpofont,atomidpos={6/45}] (R6) at (17.100,3.900) {H};
    \node[dpoO,dpodisc,dpofont,atomidpos={7/0}] (R7) at (15.600,1.600) {O};
    \node[dpoH,dpodisc,dpofont,atomidpos={8/225}] (R8) at (14.600,1.600) {H};
    \node[dpoO,dpodisc,dpofont,atomidpos={9/270}] (R9) at (16.600,2.350) {O};
    \node[dpoH,dpodisc,dpofont,atomidpos={10/0}] (R10) at (17.600,2.350) {H};
    \node[dpoH,dpodisc,dpofont,atomidpos={11/270}] (R11) at (15.600,0.600) {H};
    \draw[keepbond] (R1)--(R2);
    \draw[newbond] (R2)--(R8);
    \dbond{R3}{R4}{keepbond}{keepbond}
    \draw[keepbond] (R3)--(R5);
    \draw[newbond] (R3)--(R7);
    \draw[keepbond] (R3)--(R9);
    \draw[keepbond] (R5)--(R6);
    \draw[keepbond] (R7)--(R11);
    \draw[keepbond] (R9)--(R10);
    \draw[rulearrow] (6.35,2.70)--(5.35,2.70) node[midway,above=2pt,spanlab] {\(l\)};
    \draw[rulearrow] (11.85,2.70)--(12.85,2.70) node[midway,above=2pt,spanlab] {\(r\)};
  \end{tikzpicture}
  }
\end{minipage}\hfill\mbox{}
}

\newcommand{\nogdpopageone}{%
  \nogdpopanelsABC
  \par\vspace{1.5mm}
  \nogdpopanelD
  \par\vspace{1.5mm}
  \dpospanlegend
}

\newcommand{\nogdpopagetwo}{%
  \nogdpopanelE
  \par\vspace{1.5mm}
  \nogdpopanelF
  \par\vspace{1.5mm}
  \nogdpopanelG
  \par\vspace{1.5mm}
  \dpospanlegend
}

\begin{figure}[p]
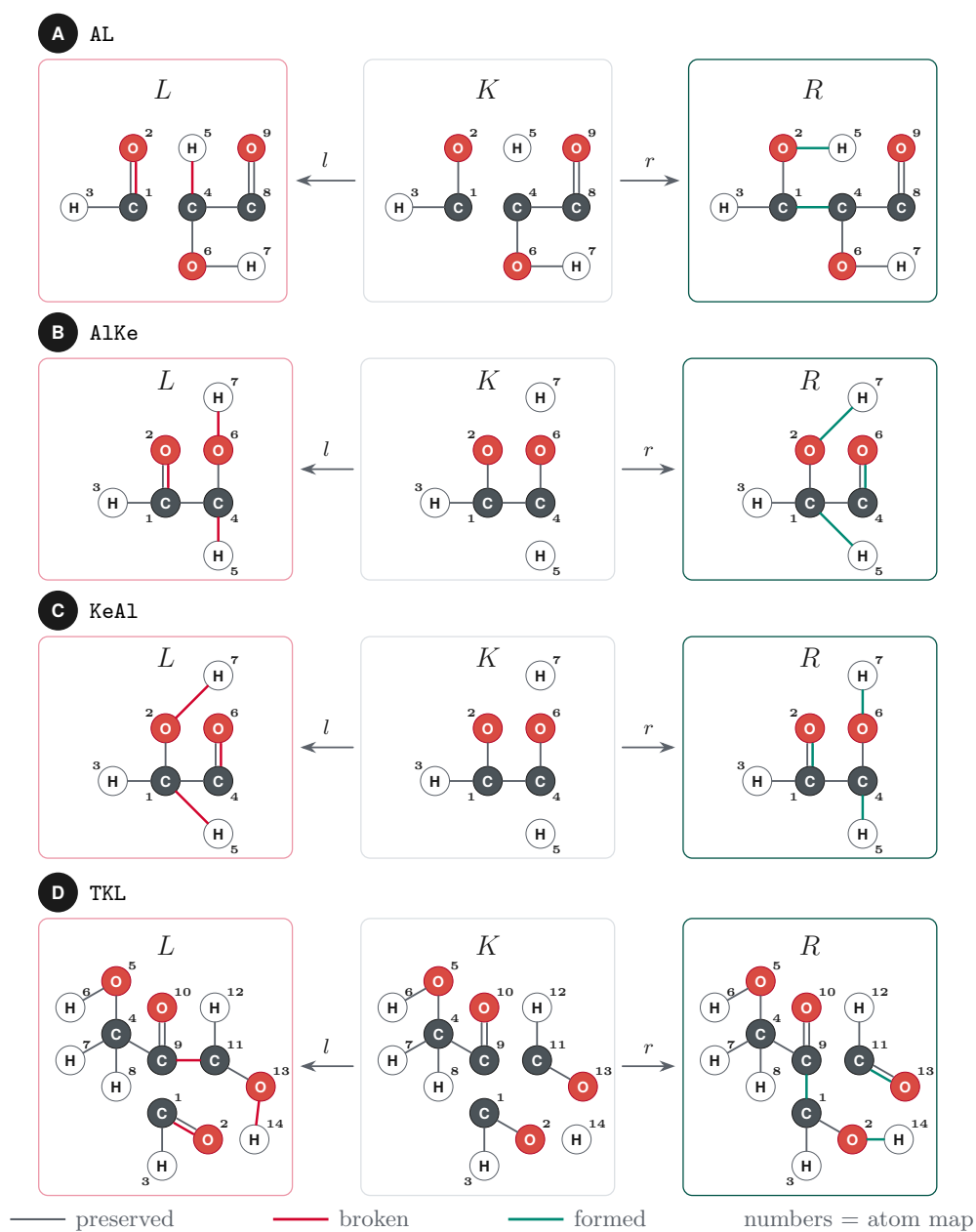

  \centering
  \nogdpopageone
  \caption{Double-pushout rules used in the non-oxidative glycolysis
    benchmark. Each panel shows a rule as the span
    \(L\xleftarrow{\,l\,}K\xrightarrow{\,r\,}R\):
    (A) \texttt{AL},
    (B) \texttt{AlKe},
    (C) \texttt{KeAl},
    (D) \texttt{TKL},
    (E) \texttt{TAL},
    (F) \texttt{PK}, and
    (G) \texttt{PHL}.
    Gray bonds are preserved by the span, red bonds are deleted by
    \(l\), and green bonds are created by \(r\). Numbered atoms retain their
    identity across \(L\), \(K\), and \(R\).}
  \label{fig:si-nog-dpo}
\end{figure}

\clearpage

\begin{figure}[p]
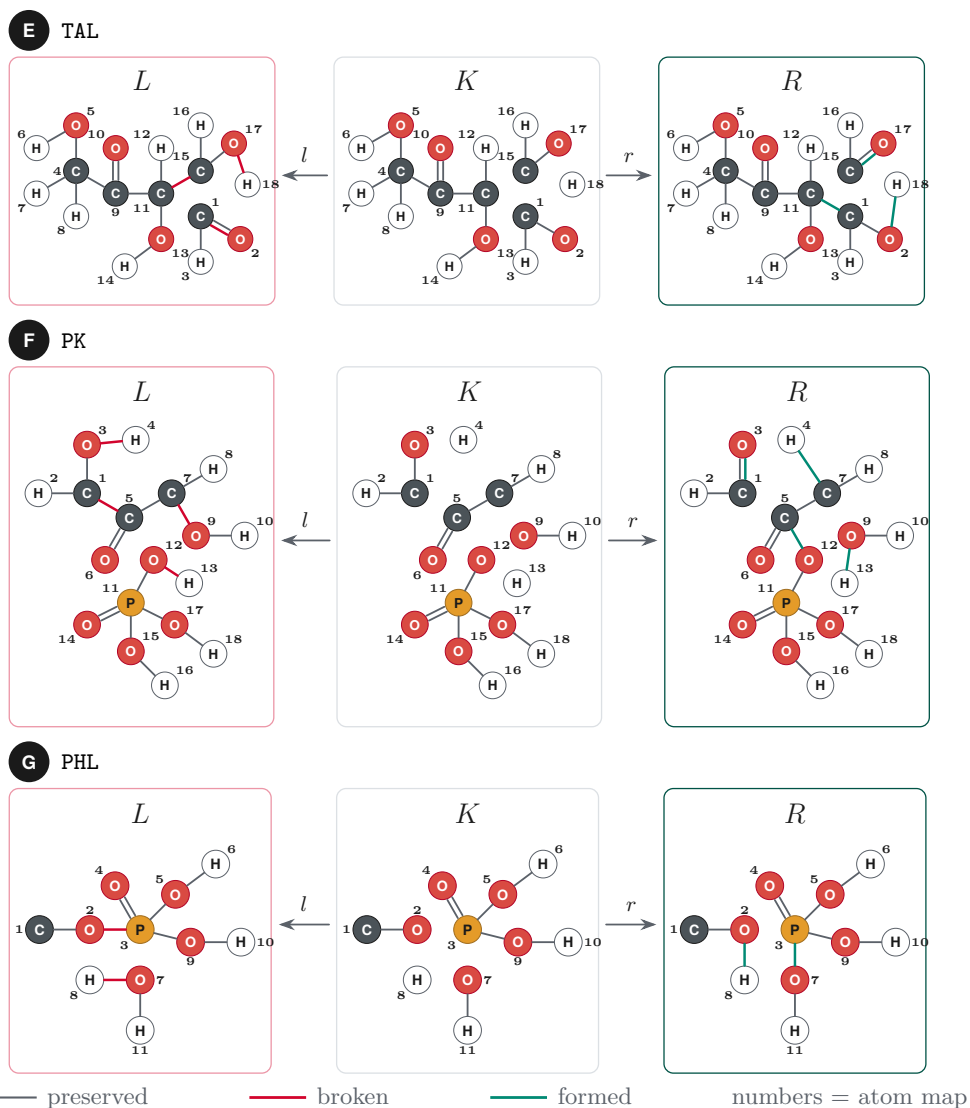

  \ContinuedFloat
  \centering
  \nogdpopagetwo
  \caption[]{Double-pushout rules used in the non-oxidative glycolysis
    benchmark (continued). Panels
    (E)--(G) show \texttt{TAL}, \texttt{PK}, and \texttt{PHL},
    respectively. Bond colors and atom numbering follow the conventions
    stated on the first page of Figure~\ref{fig:si-nog-dpo}.}
\end{figure}
\endgroup

\clearpage

\subsubsection*{S2.2: KEGG Glycolysis Module M00001}

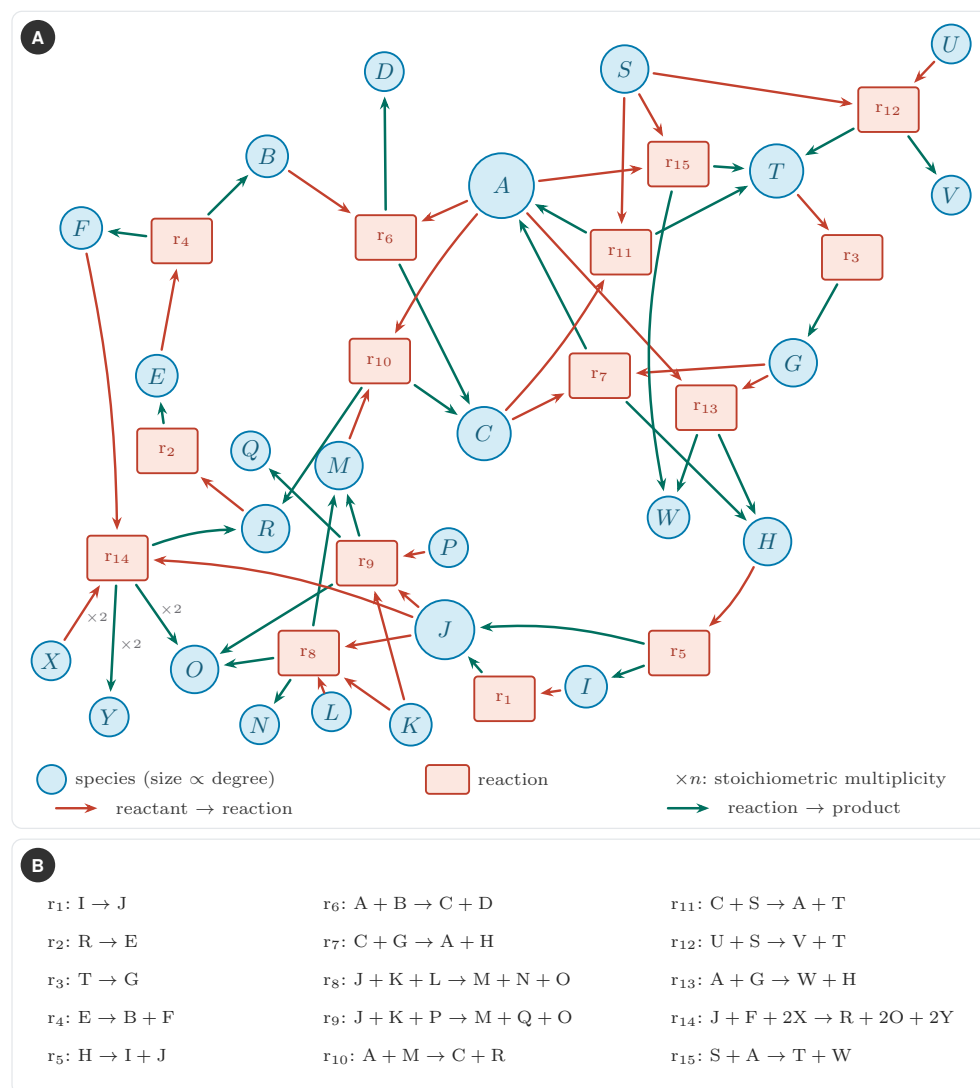
\begin{figure}[!ht]
\centering
\resizebox{\linewidth}{!}{\begin{tikzpicture}[
  x=1cm,y=1cm,font=\rmfamily\footnotesize,
  line cap=round,line join=round,
  species/.style={circle,draw=npgBlue!85!black,fill=npgBlue!17,
                  line width=.75pt,inner sep=0pt,minimum size=6.2mm,
                  font=\rmfamily\small,text=npgNavy},
  leafsp/.style={species,minimum size=5.4mm},
  reaction/.style={rounded corners=1.4pt,draw=npgRed!82!black,
                   fill=npgSalmon!24,line width=.75pt,inner sep=1.2pt,
                   minimum width=8.4mm,minimum height=6.2mm,
                   align=center,font=\rmfamily\footnotesize,
                   text=npgRed!72!black},
  inedge/.style={-{Stealth[length=1.8mm,width=1.35mm]},
                  draw=npgRed!85!black,line width=.95pt,
                  shorten >=1.8pt,shorten <=1.8pt},
  outedge/.style={-{Stealth[length=1.8mm,width=1.35mm]},
                   draw=npgGreen!70!black,line width=1pt,
                   shorten >=1.8pt,shorten <=1.8pt},
  mult/.style={font=\rmfamily\tiny,text=mut,inner sep=1pt},
  note/.style={font=\rmfamily\tiny,text=mut,align=center,inner sep=1pt},
  legend/.style={font=\rmfamily\footnotesize,text=black!78,align=center,inner sep=1pt},
  key/.style={font=\rmfamily\footnotesize,text=black!85,align=left,inner sep=1pt},
  modbox/.style={rounded corners=2.2mm,draw=mut,draw opacity=.35,dash pattern=on 2pt off 1.6pt,
                 fill=mut,fill opacity=.06}
]
  \draw[spanel] (0.00,3.80) rectangle (13.60,15.15);
  \node[badge] at (0.36,14.79) {A};
  \draw[spanel] (0.00,0.10) rectangle (13.60,3.65);
  \node[badge] at (0.36,3.29) {B};

  \begin{scope}[cm={1.05,0,0,1.20,(-0.340,2.18)}]

  \node[species,minimum size=9.0mm] (sA) at (6.80,8.79) {$A$};
  \node[species,minimum size=5.8mm] (sB) at (3.70,9.13) {$B$};
  \node[species,minimum size=7.4mm] (sC) at (6.57,5.92) {$C$};
  \node[leafsp] (sD) at (5.25,10.11) {$D$};
  \node[species,minimum size=5.8mm] (sE) at (2.24,6.59) {$E$};
  \node[species,minimum size=5.8mm] (sF) at (1.24,8.30) {$F$};
  \node[species,minimum size=6.6mm] (sG) at (10.66,6.74) {$G$};
  \node[species,minimum size=6.6mm] (sH) at (10.32,4.67) {$H$};
  \node[species,minimum size=5.8mm] (sI) at (7.92,2.99) {$I$};
  \node[species,minimum size=8.2mm] (sJ) at (6.05,3.65) {$J$};
  \node[species,minimum size=5.8mm] (sK) at (5.60,2.55) {$K$};
  \node[leafsp] (sL) at (4.55,2.70) {$L$};
  \node[species,minimum size=6.6mm] (sM) at (4.65,5.55) {$M$};
  \node[leafsp] (sN) at (3.60,2.55) {$N$};
  \node[species,minimum size=6.6mm] (sO) at (2.74,3.19) {$O$};
  \node[leafsp] (sP) at (6.10,4.60) {$P$};
  \node[leafsp] (sQ) at (3.48,5.72) {$Q$};
  \node[species,minimum size=6.6mm] (sR) at (3.68,4.82) {$R$};
  \node[species,minimum size=6.6mm] (sS) at (8.43,10.14) {$S$};
  \node[species,minimum size=7.4mm] (sT) at (10.44,8.96) {$T$};
  \node[leafsp] (sU) at (12.75,10.43) {$U$};
  \node[leafsp] (sV) at (12.75,8.68) {$V$};
  \node[species,minimum size=5.8mm] (sW) at (9.01,4.95) {$W$};
  \node[leafsp] (sX) at (0.84,3.29) {$X$};
  \node[leafsp] (sY) at (1.61,2.64) {$Y$};

  \node[reaction] (r1) at (6.84,2.85) {\(\mathrm r_1\)};
  \node[reaction] (r2) at (2.38,5.72) {\(\mathrm r_2\)};
  \node[reaction] (r3) at (11.44,7.96) {\(\mathrm r_3\)};
  \node[reaction] (r4) at (2.57,8.15) {\(\mathrm r_4\)};
  \node[reaction] (r5) at (9.15,3.38) {\(\mathrm r_5\)};
  \node[reaction] (r6) at (5.27,8.19) {\(\mathrm r_6\)};
  \node[reaction] (r7) at (8.10,6.59) {\(\mathrm r_7\)};
  \node[reaction] (r8) at (4.25,3.38) {\(\mathrm r_8\)};
  \node[reaction] (r9) at (5.02,4.42) {\(\mathrm r_9\)};
  \node[reaction] (r10) at (5.19,6.76) {\(\mathrm r_{10}\)};
  \node[reaction] (r11) at (8.38,8.02) {\(\mathrm r_{11}\)};
  \node[reaction] (r12) at (11.92,9.67) {\(\mathrm r_{12}\)};
  \node[reaction] (r13) at (9.51,6.22) {\(\mathrm r_{13}\)};
  \node[reaction] (r14) at (1.72,4.48) {\(\mathrm r_{14}\)};
  \node[reaction] (r15) at (9.14,9.04) {\(\mathrm r_{15}\)};

  \draw[inedge] (sI) to (r1);
  \draw[outedge] (r1) to (sJ);
  \draw[inedge] (sR) to (r2);
  \draw[outedge] (r2) to (sE);
  \draw[inedge] (sT) to (r3);
  \draw[outedge] (r3) to (sG);
  \draw[inedge] (sE) to (r4);
  \draw[outedge] (r4) to (sB);
  \draw[outedge] (r4) to (sF);
  \draw[inedge] (sH) to[bend left=10] (r5);
  \draw[outedge] (r5) to (sI);
  \draw[outedge] (r5) to[bend right=8] (sJ);
  \draw[inedge] (sA) to (r6);
  \draw[inedge] (sB) to (r6);
  \draw[outedge] (r6) to (sC);
  \draw[outedge] (r6) to (sD);
  \draw[inedge] (sC) to (r7);
  \draw[inedge] (sG) to (r7);
  \draw[outedge] (r7) to (sA);
  \draw[outedge] (r7) to (sH);
  \draw[inedge] (sJ) to (r8);
  \draw[inedge] (sK) to (r8);
  \draw[inedge,shorten <=0pt,shorten >=0pt] (sL) to (r8);
  \draw[outedge] (r8) to (sM);
  \draw[outedge] (r8) to (sN);
  \draw[outedge] (r8) to (sO);
  \draw[inedge] (sJ) to (r9);
  \draw[inedge] (sK) to (r9);
  \draw[inedge] (sP) to (r9);
  \draw[outedge] (r9) to (sM);
  \draw[outedge] (r9) to (sQ);
  \draw[outedge] (r9) to (sO);
  \draw[inedge] (sA) to[bend right=6] (r10);
  \draw[inedge] (sM) to (r10);
  \draw[outedge] (r10) to (sC);
  \draw[outedge] (r10) to (sR);
  \draw[inedge] (sC) to[bend right=6] (r11);
  \draw[inedge] (sS) to (r11);
  \draw[outedge] (r11) to (sA);
  \draw[outedge] (r11) to (sT);
  \draw[inedge] (sU) to (r12);
  \draw[inedge] (sS) to (r12);
  \draw[outedge] (r12) to (sV);
  \draw[outedge] (r12) to (sT);
  \draw[inedge] (sA) to (r13);
  \draw[inedge] (sG) to (r13);
  \draw[outedge] (r13) to (sW);
  \draw[outedge] (r13) to (sH);
  \draw[inedge] (sJ) to[bend right=8] (r14);
  \draw[inedge] (sF) to[bend left=6] (r14);
  \draw[inedge] (sX) to node[mult,pos=0.42,right=1pt] {\(\times2\)} (r14);
  \draw[outedge] (r14) to[bend left=8] (sR);
  \draw[outedge] (r14) to node[mult,pos=0.42,right=1pt] {\(\times2\)} (sO);
  \draw[outedge] (r14) to node[mult,pos=0.55,right=1pt] {\(\times2\)} (sY);
  \draw[inedge] (sS) to (r15);
  \draw[inedge] (sA) to (r15);
  \draw[outedge] (r15) to (sT);
  \draw[outedge] (r15) to[bend right=14] (sW);
  \end{scope}

  \node[species,minimum size=4.0mm,font=\rmfamily\scriptsize] at (0.55,4.48) {};
  \node[legend,anchor=west] at (0.84,4.48) {species (size \(\propto\) degree)};
  \node[reaction,minimum width=6.0mm,minimum height=4.0mm,font=\rmfamily\tiny]
    at (6.05,4.48) {};
  \node[legend,anchor=west] at (6.43,4.48) {reaction};
  \node[legend,anchor=west] at (9.15,4.48) {\(\times n\): stoichiometric multiplicity};
  \draw[inedge] (0.55,4.08)--(1.25,4.08);
  \node[legend,anchor=west] at (1.40,4.08) {reactant \(\to\) reaction};
  \draw[outedge] (9.05,4.08)--(9.75,4.08);
  \node[legend,anchor=west] at (9.90,4.08) {reaction \(\to\) product};

  \node[key,anchor=west] at (0.45,2.75) {\(\mathrm r_1\): \(\mathrm{I\to J}\)};
  \node[key,anchor=west] at (0.45,2.22) {\(\mathrm r_2\): \(\mathrm{R\to E}\)};
  \node[key,anchor=west] at (0.45,1.69) {\(\mathrm r_3\): \(\mathrm{T\to G}\)};
  \node[key,anchor=west] at (0.45,1.16) {\(\mathrm r_4\): \(\mathrm{E\to B+F}\)};
  \node[key,anchor=west] at (0.45,0.63) {\(\mathrm r_5\): \(\mathrm{H\to I+J}\)};
  \node[key,anchor=west] at (4.28,2.75) {\(\mathrm r_6\): \(\mathrm{A+B\to C+D}\)};
  \node[key,anchor=west] at (4.28,2.22) {\(\mathrm r_7\): \(\mathrm{C+G\to A+H}\)};
  \node[key,anchor=west] at (4.28,1.69) {\(\mathrm r_8\): \(\mathrm{J+K+L\to M+N+O}\)};
  \node[key,anchor=west] at (4.28,1.16) {\(\mathrm r_9\): \(\mathrm{J+K+P\to M+Q+O}\)};
  \node[key,anchor=west] at (4.28,0.63) {\(\mathrm r_{10}\): \(\mathrm{A+M\to C+R}\)};
  \node[key,anchor=west] at (9.11,2.75) {\(\mathrm r_{11}\): \(\mathrm{C+S\to A+T}\)};
  \node[key,anchor=west] at (9.11,2.22) {\(\mathrm r_{12}\): \(\mathrm{U+S\to V+T}\)};
  \node[key,anchor=west] at (9.11,1.69) {\(\mathrm r_{13}\): \(\mathrm{A+G\to W+H}\)};
  \node[key,anchor=west] at (9.11,1.16) {\(\mathrm r_{14}\): \(\mathrm{J+F+2X\to R+2O+2Y}\)};
  \node[key,anchor=west] at (9.11,0.63) {\(\mathrm r_{15}\): \(\mathrm{S+A\to T+W}\)};
\end{tikzpicture}}
  \caption{M00001 abstract reaction system.
  (A) Species--reaction bipartite (K{\"o}nig) representation of the
  M00001 example data. Species-node size scales with network degree; red and
  green arrows denote reactant and product incidence, respectively, and
  \(\times n\) marks stoichiometric multiplicity.
  (B) Abstract reactions used to construct the network.}
\label{fig:si-m00001-koenig}
\end{figure}

\clearpage
\begingroup
\tikzset{
  dporulerow/.style={draw=none,fill=none},
  dporulecard/.style={rounded corners=1.4mm,draw=hair,fill=white,line width=.55pt},
  dporulebadge/.style={circle,fill=ink,text=white,draw=none,
    minimum size=5.4mm,inner sep=0pt,
    font=\figatomfont\bfseries\footnotesize,xshift=2.4mm},
  dporuletitle/.style={font=\rmfamily\bfseries\small,text=ink,
    xshift=2.4mm},
  dporulecardtitle/.style={font=\rmfamily\bfseries\large,text=ink},
  dporulearrow/.style={-{Stealth[length=2.4mm,width=1.7mm]},draw=ink2,line width=.75pt},
  dporulespanlabel/.style={font=\rmfamily\bfseries\small,text=ink},
  dporulekeep/.style={draw=ink2,line width=.78pt},
  dporuledel/.style={draw=pbrk,line width=1.05pt},
  dporulenew/.style={draw=npgGreen!86!black,line width=1.05pt},
  dporuleatom/.style={circle,minimum size=4.7mm,inner sep=0pt,
    draw=ink2,line width=.45pt,font=\figatomfont\bfseries\fontsize{6.8}{6.8}\selectfont,
    label distance=1.25pt},
  dporuleC/.style={dporuleatom,draw=cpkC!70!black,fill=cpkC,text=white},
  dporuleO/.style={dporuleatom,draw=cpkO!80!black,fill=cpkO,text=white},
  dporuleN/.style={dporuleatom,draw=blue!75!black,fill=blue!64!black,text=white},
  dporuleP/.style={dporuleatom,draw=cpkP!70!black,fill=cpkP,text=ink},
  dporuleS/.style={dporuleatom,draw=yellow!55!black,fill=yellow!60,text=ink},
  dporuleFe/.style={dporuleatom,draw=brown!88!black,fill=brown!72,text=white},
  dporuleH/.style={dporuleatom,fill=white,text=ink},
  dporulemap/.style={font=\rmfamily\bfseries\tiny,
    text=ink,inner sep=.35pt}
}
\newcommand{\mzerodporulelegenditem}[2]{%
  \tikz[baseline=-.6ex]{\draw[#1] (0,0)--(.75,0);}%
  \hspace{0.35em}{\rmfamily\scriptsize\color{ink2}#2}%
}
\newcommand{\mzerodporulelegend}{%
  \noindent\hbox to \linewidth{%
    \mzerodporulelegenditem{dporulekeep}{preserved}\hfil
    \mzerodporulelegenditem{dporuledel}{broken}\hfil
    \mzerodporulelegenditem{dporulenew}{formed}\hfil
    {\rmfamily\scriptsize\color{ink2}numbers = atom map}%
  }%
}
\newcommand{\mzerodporulespageone}{%
  \resizebox{\linewidth}{!}{%
\begin{tikzpicture}[x=1cm,y=1cm,line cap=round,line join=round]
  \draw[dporulerow] (0,0) rectangle (15.2,4.48);
  \node[dporulebadge] at (0.38,4.10) {A};
  \node[anchor=west,dporuletitle] at (0.72,4.11) {\texttt{class:0}};
  \draw[dporulecard,draw=pbrk!42] (0.35,0.10) rectangle (4.65,3.78);
  \draw[dporulecard] (5.45,0.10) rectangle (9.75,3.78);
  \draw[dporulecard,draw=npgGreen!52!black] (10.55,0.10) rectangle (14.85,3.78);
  \node[dporulecardtitle] at (2.50,3.48) {\(L\)};
  \node[dporulecardtitle] at (7.60,3.48) {\(K\)};
  \node[dporulecardtitle] at (12.70,3.48) {\(R\)};
  \draw[dporulearrow] (5.30,2.05)--(4.80,2.05)
    node[midway,above=2pt,dporulespanlabel] {\(l\)};
  \draw[dporulearrow] (9.90,2.05)--(10.40,2.05)
    node[midway,above=2pt,dporulespanlabel] {\(r\)};
    \draw[dporuledel] (1.722,1.052)--(1.400,1.830);
    \draw[dporuledel] (1.400,1.830)--(1.722,2.608);
    \draw[dporuledel] (1.706,1.012)--(2.484,0.690);
    \draw[dporuledel] (1.739,1.092)--(2.516,0.770);
    \draw[dporuledel] (3.600,1.830)--(3.278,1.052);
    \draw[dporuledel] (2.500,2.930)--(3.278,2.608);
    \node[dporuleC,label={[dporulemap]225:5}] at (1.722,1.052) {C};
    \node[dporuleO,label={[dporulemap]135:6}] at (1.400,1.830) {O};
    \node[dporuleP,label={[dporulemap]135:7}] at (1.722,2.608) {P};
    \node[dporuleC,label={[dporulemap]315:8}] at (2.500,0.730) {C};
    \node[dporuleN,label={[dporulemap]45:1}] at (3.600,1.830) {N};
    \node[dporuleH,label={[dporulemap]315:2}] at (3.278,1.052) {H};
    \node[dporuleO,label={[dporulemap]45:3}] at (2.500,2.930) {O};
    \node[dporuleH,label={[dporulemap]45:4}] at (3.278,2.608) {H};
    \node[dporuleN,label={[dporulemap]45:1}] at (8.700,1.830) {N};
    \node[dporuleH,label={[dporulemap]315:2}] at (8.378,1.052) {H};
    \node[dporuleO,label={[dporulemap]45:3}] at (7.600,2.930) {O};
    \node[dporuleH,label={[dporulemap]45:4}] at (8.378,2.608) {H};
    \node[dporuleC,label={[dporulemap]225:5}] at (6.822,1.052) {C};
    \node[dporuleO,label={[dporulemap]135:6}] at (6.500,1.830) {O};
    \node[dporuleP,label={[dporulemap]135:7}] at (6.822,2.608) {P};
    \node[dporuleC,label={[dporulemap]315:8}] at (7.600,0.730) {C};
    \draw[dporulenew] (13.478,1.052)--(12.700,0.730);
    \draw[dporulenew] (12.700,0.730)--(11.922,1.052);
    \draw[dporulenew] (11.962,1.069)--(11.640,1.846);
    \draw[dporulenew] (11.882,1.036)--(11.560,1.814);
    \draw[dporulenew] (13.800,1.830)--(13.478,2.608);
    \draw[dporulenew] (12.700,2.930)--(11.922,2.608);
    \node[dporuleH,label={[dporulemap]315:2}] at (13.478,1.052) {H};
    \node[dporuleC,label={[dporulemap]315:8}] at (12.700,0.730) {C};
    \node[dporuleC,label={[dporulemap]225:5}] at (11.922,1.052) {C};
    \node[dporuleO,label={[dporulemap]135:6}] at (11.600,1.830) {O};
    \node[dporuleN,label={[dporulemap]45:1}] at (13.800,1.830) {N};
    \node[dporuleH,label={[dporulemap]45:4}] at (13.478,2.608) {H};
    \node[dporuleO,label={[dporulemap]45:3}] at (12.700,2.930) {O};
    \node[dporuleP,label={[dporulemap]135:7}] at (11.922,2.608) {P};
\end{tikzpicture}
  }
  \par\vspace{2.2mm}
  \resizebox{\linewidth}{!}{%
\begin{tikzpicture}[x=1cm,y=1cm,line cap=round,line join=round]
  \draw[dporulerow] (0,0) rectangle (15.2,4.48);
  \node[dporulebadge] at (0.38,4.10) {B};
  \node[anchor=west,dporuletitle] at (0.72,4.11) {\texttt{class:1}};
  \draw[dporulecard,draw=pbrk!42] (0.35,0.10) rectangle (4.65,3.78);
  \draw[dporulecard] (5.45,0.10) rectangle (9.75,3.78);
  \draw[dporulecard,draw=npgGreen!52!black] (10.55,0.10) rectangle (14.85,3.78);
  \node[dporulecardtitle] at (2.50,3.48) {\(L\)};
  \node[dporulecardtitle] at (7.60,3.48) {\(K\)};
  \node[dporulecardtitle] at (12.70,3.48) {\(R\)};
  \draw[dporulearrow] (5.30,2.05)--(4.80,2.05)
    node[midway,above=2pt,dporulespanlabel] {\(l\)};
  \draw[dporulearrow] (9.90,2.05)--(10.40,2.05)
    node[midway,above=2pt,dporulespanlabel] {\(r\)};
    \draw[dporuledel] (3.455,1.650)--(2.500,0.695);
    \draw[dporuledel] (3.455,1.650)--(2.500,2.605);
    \draw[dporuledel] (2.500,2.605)--(1.545,1.650);
    \node[dporuleC,label={[dporulemap]45:1}] at (3.455,1.650) {C};
    \node[dporuleH,label={[dporulemap]315:2}] at (2.500,0.695) {H};
    \node[dporuleC,label={[dporulemap]45:3}] at (2.500,2.605) {C};
    \node[dporuleO,label={[dporulemap]135:4}] at (1.545,1.650) {O};
    \node[dporuleC,label={[dporulemap]45:1}] at (8.555,1.650) {C};
    \node[dporuleH,label={[dporulemap]315:2}] at (7.600,0.695) {H};
    \node[dporuleC,label={[dporulemap]45:3}] at (7.600,2.605) {C};
    \node[dporuleO,label={[dporulemap]135:4}] at (6.645,1.650) {O};
    \draw[dporulenew] (13.685,1.680)--(12.730,2.635);
    \draw[dporulenew] (13.624,1.620)--(12.670,2.574);
    \draw[dporulenew] (12.700,0.695)--(11.745,1.650);
    \node[dporuleC,label={[dporulemap]45:1}] at (13.655,1.650) {C};
    \node[dporuleC,label={[dporulemap]45:3}] at (12.700,2.605) {C};
    \node[dporuleH,label={[dporulemap]315:2}] at (12.700,0.695) {H};
    \node[dporuleO,label={[dporulemap]135:4}] at (11.745,1.650) {O};
\end{tikzpicture}
  }
  \par\vspace{2.2mm}
  \resizebox{\linewidth}{!}{%
\begin{tikzpicture}[x=1cm,y=1cm,line cap=round,line join=round]
  \draw[dporulerow] (0,0) rectangle (15.2,4.48);
  \node[dporulebadge] at (0.38,4.10) {C};
  \node[anchor=west,dporuletitle] at (0.72,4.11) {\texttt{class:2}};
  \draw[dporulecard,draw=pbrk!42] (0.35,0.10) rectangle (4.65,3.78);
  \draw[dporulecard] (5.45,0.10) rectangle (9.75,3.78);
  \draw[dporulecard,draw=npgGreen!52!black] (10.55,0.10) rectangle (14.85,3.78);
  \node[dporulecardtitle] at (2.50,3.48) {\(L\)};
  \node[dporulecardtitle] at (7.60,3.48) {\(K\)};
  \node[dporulecardtitle] at (12.70,3.48) {\(R\)};
  \draw[dporulearrow] (5.30,2.05)--(4.80,2.05)
    node[midway,above=2pt,dporulespanlabel] {\(l\)};
  \draw[dporulearrow] (9.90,2.05)--(10.40,2.05)
    node[midway,above=2pt,dporulespanlabel] {\(r\)};
    \draw[dporuledel] (3.005,2.405)--(2.500,1.900);
    \draw[dporuledel] (1.995,2.405)--(2.500,2.910);
    \draw[dporuledel] (1.786,0.710)--(1.786,1.424);
    \draw[dporuledel] (3.214,0.710)--(2.500,0.710);
    \draw[dporulekeep] (2.500,0.710)--(2.500,1.424);
    \draw[dporuledel] (2.500,1.424)--(3.214,1.424);
    \node[dporuleC,label={[dporulemap]45:7}] at (3.005,2.405) {C};
    \node[dporuleO,label={[dporulemap]45:8}] at (2.500,1.900) {O};
    \node[dporuleC,label={[dporulemap]135:9}] at (1.995,2.405) {C};
    \node[dporuleH,label={[dporulemap]45:10}] at (2.500,2.910) {H};
    \node[dporuleH,label={[dporulemap]225:5}] at (1.786,0.710) {H};
    \node[dporuleO,label={[dporulemap]225:6}] at (1.786,1.424) {O};
    \node[dporuleP,label={[dporulemap]315:1}] at (3.214,0.710) {P};
    \node[dporuleO,label={[dporulemap]315:2}] at (2.500,0.710) {O};
    \node[dporuleP,label={[dporulemap]315:3}] at (2.500,1.424) {P};
    \node[dporuleO,label={[dporulemap]315:4}] at (3.214,1.424) {O};
    \draw[dporulekeep] (7.600,0.710)--(7.600,1.424);
    \node[dporuleP,label={[dporulemap]315:1}] at (8.314,0.710) {P};
    \node[dporuleO,label={[dporulemap]315:2}] at (7.600,0.710) {O};
    \node[dporuleP,label={[dporulemap]315:3}] at (7.600,1.424) {P};
    \node[dporuleO,label={[dporulemap]315:4}] at (8.314,1.424) {O};
    \node[dporuleH,label={[dporulemap]225:5}] at (6.886,0.710) {H};
    \node[dporuleO,label={[dporulemap]225:6}] at (6.886,1.424) {O};
    \node[dporuleC,label={[dporulemap]45:7}] at (8.105,2.405) {C};
    \node[dporuleO,label={[dporulemap]45:8}] at (7.600,1.900) {O};
    \node[dporuleC,label={[dporulemap]135:9}] at (7.095,2.405) {C};
    \node[dporuleH,label={[dporulemap]45:10}] at (7.600,2.910) {H};
    \draw[dporulenew] (13.205,2.405)--(12.700,2.910);
    \draw[dporulekeep] (12.700,0.710)--(12.700,1.424);
    \draw[dporulenew] (12.700,1.424)--(11.986,1.424);
    \draw[dporulenew] (12.700,0.710)--(11.986,0.710);
    \draw[dporulenew] (12.700,1.900)--(12.195,2.405);
    \draw[dporulenew] (13.414,0.710)--(13.414,1.424);
    \node[dporuleC,label={[dporulemap]45:7}] at (13.205,2.405) {C};
    \node[dporuleH,label={[dporulemap]45:10}] at (12.700,2.910) {H};
    \node[dporuleO,label={[dporulemap]315:2}] at (12.700,0.710) {O};
    \node[dporuleP,label={[dporulemap]315:3}] at (12.700,1.424) {P};
    \node[dporuleO,label={[dporulemap]225:6}] at (11.986,1.424) {O};
    \node[dporuleH,label={[dporulemap]225:5}] at (11.986,0.710) {H};
    \node[dporuleO,label={[dporulemap]45:8}] at (12.700,1.900) {O};
    \node[dporuleC,label={[dporulemap]135:9}] at (12.195,2.405) {C};
    \node[dporuleP,label={[dporulemap]315:1}] at (13.414,0.710) {P};
    \node[dporuleO,label={[dporulemap]315:4}] at (13.414,1.424) {O};
\end{tikzpicture}
  }
  \par\vspace{1.6mm}\mzerodporulelegend
}
\newcommand{\mzerodporulespagetwo}{%
  \resizebox{\linewidth}{!}{%
\begin{tikzpicture}[x=1cm,y=1cm,line cap=round,line join=round]
  \draw[dporulerow] (0,0) rectangle (15.2,4.48);
  \node[dporulebadge] at (0.38,4.10) {D};
  \node[anchor=west,dporuletitle] at (0.72,4.11) {\texttt{class:3}};
  \draw[dporulecard,draw=pbrk!42] (0.35,0.10) rectangle (4.65,3.78);
  \draw[dporulecard] (5.45,0.10) rectangle (9.75,3.78);
  \draw[dporulecard,draw=npgGreen!52!black] (10.55,0.10) rectangle (14.85,3.78);
  \node[dporulecardtitle] at (2.50,3.48) {\(L\)};
  \node[dporulecardtitle] at (7.60,3.48) {\(K\)};
  \node[dporulecardtitle] at (12.70,3.48) {\(R\)};
  \draw[dporulearrow] (5.30,2.05)--(4.80,2.05)
    node[midway,above=2pt,dporulespanlabel] {\(l\)};
  \draw[dporulearrow] (9.90,2.05)--(10.40,2.05)
    node[midway,above=2pt,dporulespanlabel] {\(r\)};
    \draw[dporuledel] (4.123,1.810)--(3.330,2.902);
    \draw[dporuledel] (3.330,2.902)--(2.046,2.485);
    \draw[dporuledel] (2.046,2.485)--(0.877,1.810);
    \draw[dporulekeep] (2.046,2.485)--(2.046,1.135);
    \draw[dporuledel] (2.033,1.094)--(3.316,0.677);
    \draw[dporuledel] (2.059,1.176)--(3.343,0.759);
    \node[dporuleH,label={[dporulemap]45:1}] at (4.123,1.810) {H};
    \node[dporuleO,label={[dporulemap]45:2}] at (3.330,2.902) {O};
    \node[dporuleC,label={[dporulemap]135:3}] at (2.046,2.485) {C};
    \node[dporuleH,label={[dporulemap]135:4}] at (0.877,1.810) {H};
    \node[dporuleC,label={[dporulemap]225:5}] at (2.046,1.135) {C};
    \node[dporuleO,label={[dporulemap]315:6}] at (3.330,0.718) {O};
    \draw[dporulekeep] (7.146,2.485)--(7.146,1.135);
    \node[dporuleH,label={[dporulemap]45:1}] at (9.223,1.810) {H};
    \node[dporuleO,label={[dporulemap]45:2}] at (8.430,2.902) {O};
    \node[dporuleC,label={[dporulemap]135:3}] at (7.146,2.485) {C};
    \node[dporuleH,label={[dporulemap]135:4}] at (5.977,1.810) {H};
    \node[dporuleC,label={[dporulemap]225:5}] at (7.146,1.135) {C};
    \node[dporuleO,label={[dporulemap]315:6}] at (8.430,0.718) {O};
    \draw[dporulenew] (14.323,1.810)--(13.530,0.718);
    \draw[dporulenew] (13.530,0.718)--(12.246,1.135);
    \draw[dporulekeep] (12.246,1.135)--(12.246,2.485);
    \draw[dporulenew] (12.259,2.444)--(13.543,2.861);
    \draw[dporulenew] (12.233,2.526)--(13.516,2.943);
    \draw[dporulenew] (12.246,1.135)--(11.077,1.810);
    \node[dporuleH,label={[dporulemap]45:1}] at (14.323,1.810) {H};
    \node[dporuleO,label={[dporulemap]315:6}] at (13.530,0.718) {O};
    \node[dporuleC,label={[dporulemap]225:5}] at (12.246,1.135) {C};
    \node[dporuleC,label={[dporulemap]135:3}] at (12.246,2.485) {C};
    \node[dporuleO,label={[dporulemap]45:2}] at (13.530,2.902) {O};
    \node[dporuleH,label={[dporulemap]135:4}] at (11.077,1.810) {H};
\end{tikzpicture}
  }
  \par\vspace{2.2mm}
  \resizebox{\linewidth}{!}{%
\begin{tikzpicture}[x=1cm,y=1cm,line cap=round,line join=round]
  \draw[dporulerow] (0,0) rectangle (15.2,4.48);
  \node[dporulebadge] at (0.38,4.10) {E};
  \node[anchor=west,dporuletitle] at (0.72,4.11) {\texttt{class:4}};
  \draw[dporulecard,draw=pbrk!42] (0.35,0.10) rectangle (4.65,3.78);
  \draw[dporulecard] (5.45,0.10) rectangle (9.75,3.78);
  \draw[dporulecard,draw=npgGreen!52!black] (10.55,0.10) rectangle (14.85,3.78);
  \node[dporulecardtitle] at (2.50,3.48) {\(L\)};
  \node[dporulecardtitle] at (7.60,3.48) {\(K\)};
  \node[dporulecardtitle] at (12.70,3.48) {\(R\)};
  \draw[dporulearrow] (5.30,2.05)--(4.80,2.05)
    node[midway,above=2pt,dporulespanlabel] {\(l\)};
  \draw[dporulearrow] (9.90,2.05)--(10.40,2.05)
    node[midway,above=2pt,dporulespanlabel] {\(r\)};
    \draw[dporuledel] (2.030,0.820)--(1.070,0.820);
    \draw[dporuledel] (2.251,2.387)--(1.671,2.717);
    \draw[dporuledel] (2.209,2.313)--(1.629,2.643);
    \draw[dporuledel] (1.650,2.680)--(1.070,2.350);
    \draw[dporuledel] (1.027,2.350)--(1.027,1.720);
    \draw[dporuledel] (1.113,2.350)--(1.113,1.720);
    \draw[dporulekeep] (1.070,1.720)--(1.650,1.390);
    \draw[dporulekeep] (1.671,1.353)--(2.251,1.683);
    \draw[dporulekeep] (1.629,1.427)--(2.209,1.757);
    \draw[dporuledel] (2.990,0.820)--(3.900,0.820);
    \draw[dporulekeep] (2.230,1.720)--(2.230,2.350);
    \node[dporuleC,label={[dporulemap]315:1}] at (2.030,0.820) {C};
    \node[dporuleH,label={[dporulemap]225:2}] at (1.070,0.820) {H};
    \node[dporuleN,label={[dporulemap]45:5}] at (2.230,2.350) {\(\mathsf{N}^{+}\)};
    \node[dporuleC,label={[dporulemap]135:10}] at (1.650,2.680) {C};
    \node[dporuleC,label={[dporulemap]135:9}] at (1.070,2.350) {C};
    \node[dporuleC,label={[dporulemap]135:8}] at (1.070,1.720) {C};
    \node[dporuleC,label={[dporulemap]225:7}] at (1.650,1.390) {C};
    \node[dporuleC,label={[dporulemap]45:6}] at (2.230,1.720) {C};
    \node[dporuleO,label={[dporulemap]315:3}] at (2.990,0.820) {O};
    \node[dporuleH,label={[dporulemap]315:4}] at (3.900,0.820) {H};
    \draw[dporulekeep] (6.170,1.720)--(6.750,1.390);
    \draw[dporulekeep] (6.771,1.353)--(7.351,1.683);
    \draw[dporulekeep] (6.729,1.427)--(7.309,1.757);
    \draw[dporulekeep] (7.330,1.720)--(7.330,2.350);
    \node[dporuleC,label={[dporulemap]315:1}] at (7.130,0.820) {C};
    \node[dporuleH,label={[dporulemap]225:2}] at (6.170,0.820) {H};
    \node[dporuleO,label={[dporulemap]315:3}] at (8.090,0.820) {O};
    \node[dporuleH,label={[dporulemap]315:4}] at (9.000,0.820) {H};
    \node[dporuleN,label={[dporulemap]45:5}] at (7.330,2.350) {N};
    \node[dporuleC,label={[dporulemap]45:6}] at (7.330,1.720) {C};
    \node[dporuleC,label={[dporulemap]225:7}] at (6.750,1.390) {C};
    \node[dporuleC,label={[dporulemap]135:8}] at (6.170,1.720) {C};
    \node[dporuleC,label={[dporulemap]135:9}] at (6.170,2.350) {C};
    \node[dporuleC,label={[dporulemap]135:10}] at (6.750,2.680) {C};
    \draw[dporulenew] (12.230,0.820)--(13.190,0.820);
    \draw[dporulenew] (11.270,0.820)--(11.270,1.720);
    \draw[dporulekeep] (11.270,1.720)--(11.850,1.390);
    \draw[dporulekeep] (11.871,1.353)--(12.451,1.683);
    \draw[dporulekeep] (11.829,1.427)--(12.409,1.757);
    \draw[dporulekeep] (12.430,1.720)--(12.430,2.350);
    \draw[dporulenew] (12.430,2.350)--(11.850,2.680);
    \draw[dporulenew] (11.829,2.717)--(11.249,2.387);
    \draw[dporulenew] (11.871,2.643)--(11.291,2.313);
    \draw[dporulenew] (11.270,2.350)--(11.270,1.720);
    \node[dporuleC,label={[dporulemap]315:1}] at (12.230,0.820) {C};
    \node[dporuleO,label={[dporulemap]315:3}] at (13.190,0.820) {O};
    \node[dporuleH,label={[dporulemap]315:4}] at (14.100,0.820) {\(\mathsf{H}^{+}\)};
    \node[dporuleH,label={[dporulemap]225:2}] at (11.270,0.820) {H};
    \node[dporuleC,label={[dporulemap]135:8}] at (11.270,1.720) {C};
    \node[dporuleC,label={[dporulemap]225:7}] at (11.850,1.390) {C};
    \node[dporuleC,label={[dporulemap]45:6}] at (12.430,1.720) {C};
    \node[dporuleN,label={[dporulemap]45:5}] at (12.430,2.350) {N};
    \node[dporuleC,label={[dporulemap]135:10}] at (11.850,2.680) {C};
    \node[dporuleC,label={[dporulemap]135:9}] at (11.270,2.350) {C};
\end{tikzpicture}
  }
  \par\vspace{2.2mm}
  \resizebox{\linewidth}{!}{%
\begin{tikzpicture}[x=1cm,y=1cm,line cap=round,line join=round]
  \draw[dporulerow] (0,0) rectangle (15.2,4.48);
  \node[dporulebadge] at (0.38,4.10) {F};
  \node[anchor=west,dporuletitle] at (0.72,4.11) {\texttt{class:5}};
  \draw[dporulecard,draw=pbrk!42] (0.35,0.10) rectangle (4.65,3.78);
  \draw[dporulecard] (5.45,0.10) rectangle (9.75,3.78);
  \draw[dporulecard,draw=npgGreen!52!black] (10.55,0.10) rectangle (14.85,3.78);
  \node[dporulecardtitle] at (2.50,3.48) {\(L\)};
  \node[dporulecardtitle] at (7.60,3.48) {\(K\)};
  \node[dporulecardtitle] at (12.70,3.48) {\(R\)};
  \draw[dporulearrow] (5.30,2.05)--(4.80,2.05)
    node[midway,above=2pt,dporulespanlabel] {\(l\)};
  \draw[dporulearrow] (9.90,2.05)--(10.40,2.05)
    node[midway,above=2pt,dporulespanlabel] {\(r\)};
    \draw[dporuledel] (0.997,2.360)--(0.997,1.260);
    \draw[dporuledel] (0.997,1.260)--(2.097,1.260);
    \draw[dporuledel] (2.097,1.260)--(3.050,0.710);
    \draw[dporulekeep] (2.097,1.260)--(2.097,2.360);
    \draw[dporuledel] (2.097,2.360)--(3.050,2.910);
    \draw[dporuledel] (3.050,2.910)--(4.003,2.360);
    \draw[dporuledel] (4.003,2.360)--(4.003,1.260);
    \node[dporuleH,label={[dporulemap]135:1}] at (0.997,2.360) {H};
    \node[dporuleO,label={[dporulemap]225:2}] at (0.997,1.260) {O};
    \node[dporuleC,label={[dporulemap]225:3}] at (2.097,1.260) {C};
    \node[dporuleO,label={[dporulemap]315:4}] at (3.050,0.710) {O};
    \node[dporuleC,label={[dporulemap]135:8}] at (2.097,2.360) {C};
    \node[dporuleC,label={[dporulemap]45:5}] at (3.050,2.910) {C};
    \node[dporuleO,label={[dporulemap]45:6}] at (4.003,2.360) {O};
    \node[dporuleH,label={[dporulemap]315:7}] at (4.003,1.260) {H};
    \draw[dporulekeep] (7.197,1.260)--(7.197,2.360);
    \node[dporuleH,label={[dporulemap]135:1}] at (6.097,2.360) {H};
    \node[dporuleO,label={[dporulemap]225:2}] at (6.097,1.260) {O};
    \node[dporuleC,label={[dporulemap]225:3}] at (7.197,1.260) {C};
    \node[dporuleO,label={[dporulemap]315:4}] at (8.150,0.710) {O};
    \node[dporuleC,label={[dporulemap]45:5}] at (8.150,2.910) {C};
    \node[dporuleO,label={[dporulemap]45:6}] at (9.103,2.360) {O};
    \node[dporuleH,label={[dporulemap]315:7}] at (9.103,1.260) {H};
    \node[dporuleC,label={[dporulemap]135:8}] at (7.197,2.360) {C};
    \draw[dporulenew] (13.229,2.873)--(14.181,2.323);
    \draw[dporulenew] (13.271,2.947)--(14.224,2.397);
    \draw[dporulenew] (11.197,2.360)--(12.297,2.360);
    \draw[dporulekeep] (12.297,2.360)--(12.297,1.260);
    \draw[dporulenew] (12.297,1.303)--(11.197,1.303);
    \draw[dporulenew] (12.297,1.217)--(11.197,1.217);
    \draw[dporulenew] (13.250,0.710)--(14.203,1.260);
    \node[dporuleC,label={[dporulemap]45:5}] at (13.250,2.910) {C};
    \node[dporuleO,label={[dporulemap]45:6}] at (14.203,2.360) {O};
    \node[dporuleH,label={[dporulemap]135:1}] at (11.197,2.360) {H};
    \node[dporuleC,label={[dporulemap]135:8}] at (12.297,2.360) {C};
    \node[dporuleC,label={[dporulemap]225:3}] at (12.297,1.260) {C};
    \node[dporuleO,label={[dporulemap]225:2}] at (11.197,1.260) {O};
    \node[dporuleO,label={[dporulemap]315:4}] at (13.250,0.710) {O};
    \node[dporuleH,label={[dporulemap]315:7}] at (14.203,1.260) {H};
\end{tikzpicture}
  }
  \par\vspace{1.6mm}\mzerodporulelegend
}
\newcommand{\mzerodporulespagethree}{%
  \resizebox{\linewidth}{!}{%
\begin{tikzpicture}[x=1cm,y=1cm,line cap=round,line join=round]
  \draw[dporulerow] (0,0) rectangle (15.2,4.48);
  \node[dporulebadge] at (0.38,4.10) {G};
  \node[anchor=west,dporuletitle] at (0.72,4.11) {\texttt{class:6}};
  \draw[dporulecard,draw=pbrk!42] (0.35,0.10) rectangle (4.65,3.78);
  \draw[dporulecard] (5.45,0.10) rectangle (9.75,3.78);
  \draw[dporulecard,draw=npgGreen!52!black] (10.55,0.10) rectangle (14.85,3.78);
  \node[dporulecardtitle] at (2.50,3.48) {\(L\)};
  \node[dporulecardtitle] at (7.60,3.48) {\(K\)};
  \node[dporulecardtitle] at (12.70,3.48) {\(R\)};
  \draw[dporulearrow] (5.30,2.05)--(4.80,2.05)
    node[midway,above=2pt,dporulespanlabel] {\(l\)};
  \draw[dporulearrow] (9.90,2.05)--(10.40,2.05)
    node[midway,above=2pt,dporulespanlabel] {\(r\)};
    \draw[dporuledel] (3.770,1.810)--(3.135,0.710);
    \draw[dporuledel] (1.230,1.810)--(1.865,0.710);
    \draw[dporuledel] (3.135,2.910)--(1.865,2.910);
    \node[dporuleH,label={[dporulemap]315:1}] at (3.770,1.810) {H};
    \node[dporuleO,label={[dporulemap]315:2}] at (3.135,0.710) {O};
    \node[dporuleO,label={[dporulemap]225:3}] at (1.230,1.810) {O};
    \node[dporuleH,label={[dporulemap]225:4}] at (1.865,0.710) {H};
    \node[dporuleO,label={[dporulemap]45:5}] at (3.135,2.910) {O};
    \node[dporuleP,label={[dporulemap]135:6}] at (1.865,2.910) {P};
    \node[dporuleH,label={[dporulemap]315:1}] at (8.870,1.810) {H};
    \node[dporuleO,label={[dporulemap]315:2}] at (8.235,0.710) {O};
    \node[dporuleO,label={[dporulemap]225:3}] at (6.330,1.810) {O};
    \node[dporuleH,label={[dporulemap]225:4}] at (6.965,0.710) {H};
    \node[dporuleO,label={[dporulemap]45:5}] at (8.235,2.910) {O};
    \node[dporuleP,label={[dporulemap]135:6}] at (6.965,2.910) {P};
    \draw[dporulenew] (13.970,1.810)--(13.335,2.910);
    \draw[dporulenew] (13.335,0.710)--(12.065,0.710);
    \draw[dporulenew] (11.430,1.810)--(12.065,2.910);
    \node[dporuleH,label={[dporulemap]315:1}] at (13.970,1.810) {H};
    \node[dporuleO,label={[dporulemap]45:5}] at (13.335,2.910) {O};
    \node[dporuleO,label={[dporulemap]315:2}] at (13.335,0.710) {O};
    \node[dporuleH,label={[dporulemap]225:4}] at (12.065,0.710) {H};
    \node[dporuleO,label={[dporulemap]225:3}] at (11.430,1.810) {O};
    \node[dporuleP,label={[dporulemap]135:6}] at (12.065,2.910) {P};
\end{tikzpicture}
  }
  \par\vspace{2.2mm}
  \resizebox{\linewidth}{!}{%
\begin{tikzpicture}[x=1cm,y=1cm,line cap=round,line join=round]
  \draw[dporulerow] (0,0) rectangle (15.2,4.48);
  \node[dporulebadge] at (0.38,4.10) {H};
  \node[anchor=west,dporuletitle] at (0.72,4.11) {\texttt{class:7}};
  \draw[dporulecard,draw=pbrk!42] (0.35,0.10) rectangle (4.65,3.78);
  \draw[dporulecard] (5.45,0.10) rectangle (9.75,3.78);
  \draw[dporulecard,draw=npgGreen!52!black] (10.55,0.10) rectangle (14.85,3.78);
  \node[dporulecardtitle] at (2.50,3.48) {\(L\)};
  \node[dporulecardtitle] at (7.60,3.48) {\(K\)};
  \node[dporulecardtitle] at (12.70,3.48) {\(R\)};
  \draw[dporulearrow] (5.30,2.05)--(4.80,2.05)
    node[midway,above=2pt,dporulespanlabel] {\(l\)};
  \draw[dporulearrow] (9.90,2.05)--(10.40,2.05)
    node[midway,above=2pt,dporulespanlabel] {\(r\)};
    \draw[dporuledel] (1.545,1.650)--(2.500,2.605);
    \draw[dporuledel] (3.455,1.650)--(2.500,0.695);
    \node[dporuleO,label={[dporulemap]135:3}] at (1.545,1.650) {O};
    \node[dporuleP,label={[dporulemap]45:4}] at (2.500,2.605) {P};
    \node[dporuleO,label={[dporulemap]45:1}] at (3.455,1.650) {O};
    \node[dporuleH,label={[dporulemap]315:2}] at (2.500,0.695) {H};
    \node[dporuleO,label={[dporulemap]45:1}] at (8.555,1.650) {O};
    \node[dporuleH,label={[dporulemap]315:2}] at (7.600,0.695) {H};
    \node[dporuleO,label={[dporulemap]135:3}] at (6.645,1.650) {O};
    \node[dporuleP,label={[dporulemap]45:4}] at (7.600,2.605) {P};
    \draw[dporulenew] (12.700,0.695)--(11.745,1.650);
    \draw[dporulenew] (13.655,1.650)--(12.700,2.605);
    \node[dporuleH,label={[dporulemap]315:2}] at (12.700,0.695) {H};
    \node[dporuleO,label={[dporulemap]135:3}] at (11.745,1.650) {O};
    \node[dporuleO,label={[dporulemap]45:1}] at (13.655,1.650) {O};
    \node[dporuleP,label={[dporulemap]45:4}] at (12.700,2.605) {P};
\end{tikzpicture}
  }
  \par\vspace{2.2mm}
  \resizebox{\linewidth}{!}{%
\begin{tikzpicture}[x=1cm,y=1cm,line cap=round,line join=round]
  \draw[dporulerow] (0,0) rectangle (15.2,4.48);
  \node[dporulebadge] at (0.38,4.10) {I};
  \node[anchor=west,dporuletitle] at (0.72,4.11) {\texttt{class:8}};
  \draw[dporulecard,draw=pbrk!42] (0.35,0.10) rectangle (4.65,3.78);
  \draw[dporulecard] (5.45,0.10) rectangle (9.75,3.78);
  \draw[dporulecard,draw=npgGreen!52!black] (10.55,0.10) rectangle (14.85,3.78);
  \node[dporulecardtitle] at (2.50,3.48) {\(L\)};
  \node[dporulecardtitle] at (7.60,3.48) {\(K\)};
  \node[dporulecardtitle] at (12.70,3.48) {\(R\)};
  \draw[dporulearrow] (5.30,2.05)--(4.80,2.05)
    node[midway,above=2pt,dporulespanlabel] {\(l\)};
  \draw[dporulearrow] (9.90,2.05)--(10.40,2.05)
    node[midway,above=2pt,dporulespanlabel] {\(r\)};
    \draw[dporuledel] (1.545,1.650)--(2.500,2.605);
    \draw[dporuledel] (3.455,1.650)--(2.500,0.695);
    \node[dporuleO,label={[dporulemap]135:3}] at (1.545,1.650) {O};
    \node[dporuleC,label={[dporulemap]45:4}] at (2.500,2.605) {C};
    \node[dporuleO,label={[dporulemap]45:1}] at (3.455,1.650) {O};
    \node[dporuleP,label={[dporulemap]315:2}] at (2.500,0.695) {P};
    \node[dporuleO,label={[dporulemap]45:1}] at (8.555,1.650) {O};
    \node[dporuleP,label={[dporulemap]315:2}] at (7.600,0.695) {P};
    \node[dporuleO,label={[dporulemap]135:3}] at (6.645,1.650) {O};
    \node[dporuleC,label={[dporulemap]45:4}] at (7.600,2.605) {C};
    \draw[dporulenew] (13.655,1.650)--(12.700,2.605);
    \draw[dporulenew] (12.700,0.695)--(11.745,1.650);
    \node[dporuleO,label={[dporulemap]45:1}] at (13.655,1.650) {O};
    \node[dporuleC,label={[dporulemap]45:4}] at (12.700,2.605) {C};
    \node[dporuleP,label={[dporulemap]315:2}] at (12.700,0.695) {P};
    \node[dporuleO,label={[dporulemap]135:3}] at (11.745,1.650) {O};
\end{tikzpicture}
  }
  \par\vspace{1.6mm}\mzerodporulelegend
}
\newcommand{\mzerodporulespagefour}{%
  \resizebox{\linewidth}{!}{%
\begin{tikzpicture}[x=1cm,y=1cm,line cap=round,line join=round]
  \draw[dporulerow] (0,0) rectangle (15.2,4.48);
  \node[dporulebadge] at (0.38,4.10) {J};
  \node[anchor=west,dporuletitle] at (0.72,4.11) {\texttt{class:9}};
  \draw[dporulecard,draw=pbrk!42] (0.35,0.10) rectangle (4.65,3.78);
  \draw[dporulecard] (5.45,0.10) rectangle (9.75,3.78);
  \draw[dporulecard,draw=npgGreen!52!black] (10.55,0.10) rectangle (14.85,3.78);
  \node[dporulecardtitle] at (2.50,3.48) {\(L\)};
  \node[dporulecardtitle] at (7.60,3.48) {\(K\)};
  \node[dporulecardtitle] at (12.70,3.48) {\(R\)};
  \draw[dporulearrow] (5.30,2.05)--(4.80,2.05)
    node[midway,above=2pt,dporulespanlabel] {\(l\)};
  \draw[dporulearrow] (9.90,2.05)--(10.40,2.05)
    node[midway,above=2pt,dporulespanlabel] {\(r\)};
    \draw[dporuledel] (1.550,1.480)--(3.300,1.480);
    \draw[dporuledel] (3.300,0.680)--(1.550,0.680);
    \node[dporuleC,label={[dporulemap]225:1}] at (1.550,1.480) {C};
    \node[dporuleH,label={[dporulemap]315:2}] at (3.300,1.480) {H};
    \node[dporuleFe,label={[dporulemap]135:5}] at (1.770,2.550) {\(\mathsf{Fe}^{+}\)};
    \node[dporuleFe,label={[dporulemap]45:6}] at (3.070,2.550) {\(\mathsf{Fe}^{+}\)};
    \node[dporuleH,label={[dporulemap]315:3}] at (3.300,0.680) {H};
    \node[dporuleO,label={[dporulemap]225:4}] at (1.550,0.680) {O};
    \node[dporuleC,label={[dporulemap]225:1}] at (6.650,1.480) {C};
    \node[dporuleH,label={[dporulemap]315:2}] at (8.400,1.480) {H};
    \node[dporuleH,label={[dporulemap]315:3}] at (8.400,0.680) {H};
    \node[dporuleO,label={[dporulemap]225:4}] at (6.650,0.680) {O};
    \node[dporuleFe,label={[dporulemap]135:5}] at (6.870,2.550) {Fe};
    \node[dporuleFe,label={[dporulemap]45:6}] at (8.170,2.550) {Fe};
    \draw[dporulenew] (11.750,1.480)--(11.750,0.680);
    \node[dporuleC,label={[dporulemap]225:1}] at (11.750,1.480) {C};
    \node[dporuleO,label={[dporulemap]225:4}] at (11.750,0.680) {O};
    \node[dporuleFe,label={[dporulemap]135:5}] at (11.970,2.550) {Fe};
    \node[dporuleFe,label={[dporulemap]45:6}] at (13.270,2.550) {Fe};
    \node[dporuleH,label={[dporulemap]315:3}] at (13.500,0.680) {\(\mathsf{H}^{+}\)};
    \node[dporuleH,label={[dporulemap]315:2}] at (13.500,1.480) {\(\mathsf{H}^{+}\)};
\end{tikzpicture}
  }
  \par\vspace{2.2mm}
  \resizebox{\linewidth}{!}{%
\begin{tikzpicture}[x=1cm,y=1cm,line cap=round,line join=round]
  \draw[dporulerow] (0,0) rectangle (15.2,4.48);
  \node[dporulebadge] at (0.38,4.10) {K};
  \node[anchor=west,dporuletitle] at (0.72,4.11) {\texttt{class:10}};
  \draw[dporulecard,draw=pbrk!42] (0.35,0.34) rectangle (4.65,3.78);
  \draw[dporulecard] (5.45,0.34) rectangle (9.75,3.78);
  \draw[dporulecard,draw=npgGreen!52!black] (10.55,0.34) rectangle (14.85,3.78);
  \node[dporulecardtitle] at (2.50,3.48) {\(L\)};
  \node[dporulecardtitle] at (7.60,3.48) {\(K\)};
  \node[dporulecardtitle] at (12.70,3.48) {\(R\)};
  \draw[dporulearrow] (5.30,2.05)--(4.80,2.05)
    node[midway,above=2pt,dporulespanlabel] {\(l\)};
  \draw[dporulearrow] (9.90,2.05)--(10.40,2.05)
    node[midway,above=2pt,dporulespanlabel] {\(r\)};
    \draw[dporuledel] (2.500,2.325)--(3.669,1.650);
    \draw[dporulekeep] (2.500,2.325)--(2.500,0.975);
    \draw[dporuledel] (2.500,0.975)--(1.331,1.650);
    \node[dporuleC,label={[dporulemap]45:1}] at (2.500,2.325) {C};
    \node[dporuleH,label={[dporulemap]315:2}] at (3.669,1.650) {H};
    \node[dporuleC,label={[dporulemap]315:4}] at (2.500,0.975) {C};
    \node[dporuleC,label={[dporulemap]135:3}] at (1.331,1.650) {C};
    \draw[dporulekeep] (7.600,2.325)--(7.600,0.975);
    \node[dporuleC,label={[dporulemap]45:1}] at (7.600,2.325) {C};
    \node[dporuleH,label={[dporulemap]315:2}] at (8.769,1.650) {H};
    \node[dporuleC,label={[dporulemap]135:3}] at (6.431,1.650) {C};
    \node[dporuleC,label={[dporulemap]315:4}] at (7.600,0.975) {C};
    \draw[dporulenew] (12.700,2.325)--(11.531,1.650);
    \draw[dporulekeep] (12.700,2.325)--(12.700,0.975);
    \draw[dporulenew] (12.700,0.975)--(13.869,1.650);
    \node[dporuleC,label={[dporulemap]45:1}] at (12.700,2.325) {C};
    \node[dporuleC,label={[dporulemap]135:3}] at (11.531,1.650) {C};
    \node[dporuleC,label={[dporulemap]315:4}] at (12.700,0.975) {C};
    \node[dporuleH,label={[dporulemap]315:2}] at (13.869,1.650) {H};
\end{tikzpicture}
  }
  \par\vspace{1.6mm}\mzerodporulelegend
}

\begin{figure}[p]
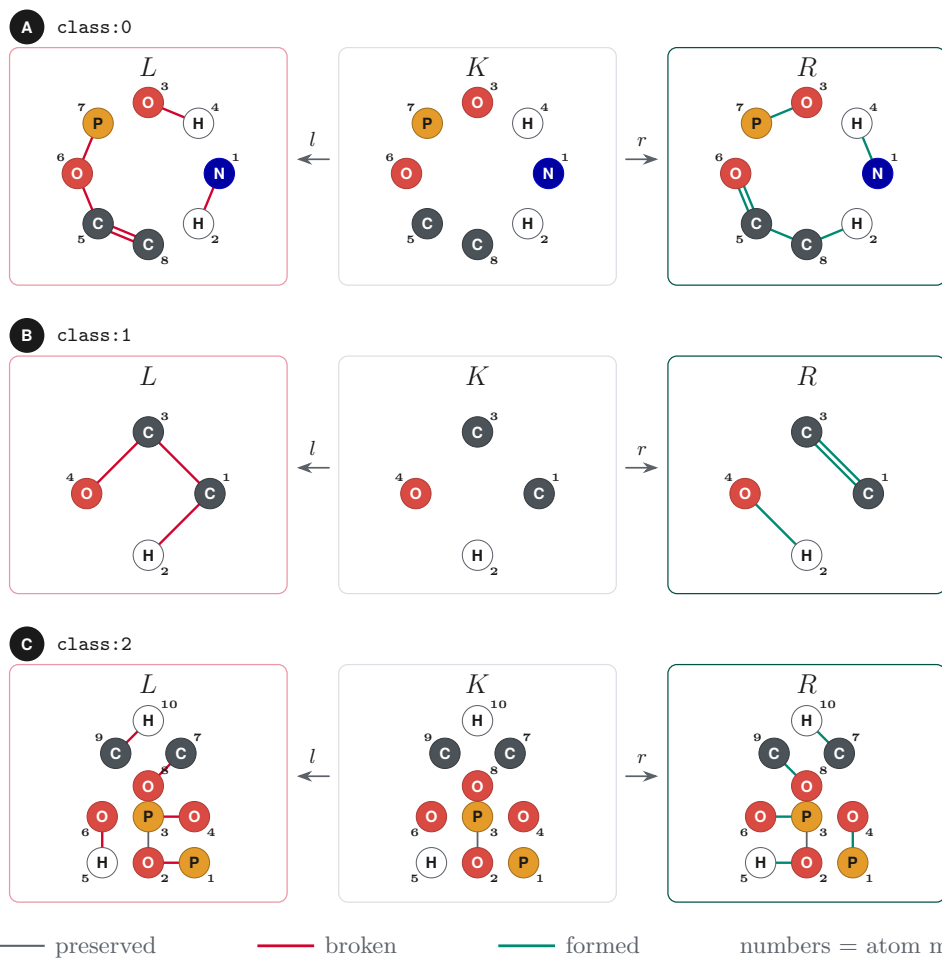

  \centering
  \mzerodporulespageone
  \caption{Double-pushout transformations for KEGG glycolysis module M00001.
    Panels (A)--(K) show the eleven transformations as DPO spans
    \(L\xleftarrow{\,l\,}K\xrightarrow{\,r\,}R\), with mapped atoms placed
    consistently across all three graphs. Gray bonds are preserved in \(K\),
    red bonds are deleted or change order in \(L\), and green bonds
    are created or change order in \(R\). Numbered atoms retain their identity
    across the span.}
  \label{fig:si-m00001-rules}
\end{figure}

\clearpage

\begin{figure}[p]
  \ContinuedFloat
  \centering
  \mzerodporulespagetwo
  \caption[]{Double-pushout transformations for KEGG glycolysis module
    M00001 (continued). Panels (D)--(F) are shown on this page; colors and
    atom maps follow the conventions stated on the first page.}
\end{figure}

\clearpage

\begin{figure}[p]
  \ContinuedFloat
  \centering
  \mzerodporulespagethree
  \caption[]{Double-pushout transformations for KEGG glycolysis module
    M00001 (continued). Panels (G)--(I) are shown on this page; colors and
    atom maps follow the conventions stated on the first page.}
\end{figure}

\clearpage

\begin{figure}[p]
  \ContinuedFloat
  \centering
  \mzerodporulespagefour
  \caption[]{Double-pushout transformations for KEGG glycolysis module
    M00001 (continued). Panels (J) and (K) are shown on this page; colors
    and atom maps follow the conventions stated on the first page of
    Figure~\ref{fig:si-m00001-rules}.}
\end{figure}
\endgroup

\clearpage

\begingroup
\input{figure/fig-m00001-spool}

\begin{figure}[p]
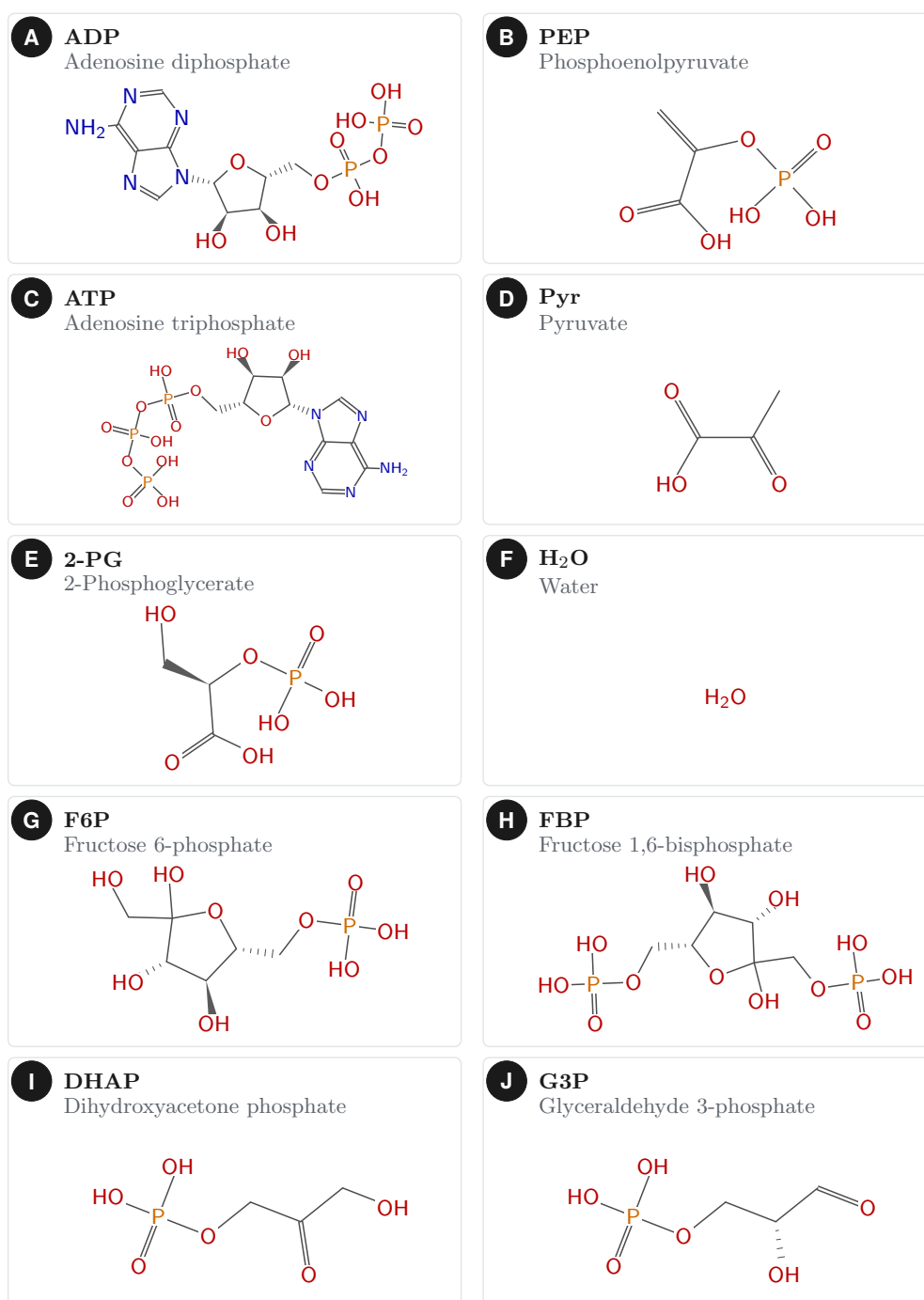

  \centering
  \resizebox{\linewidth}{!}{\mzerospeciespageone}
  \caption{Molecular structures of the 25 species in KEGG glycolysis module
    M00001. Panels (A)--(Y) show the molecular graphs as skeletal
    structures.}
  \label{fig:si-m00001-molecules}
\end{figure}

\clearpage

\begin{figure}[p]
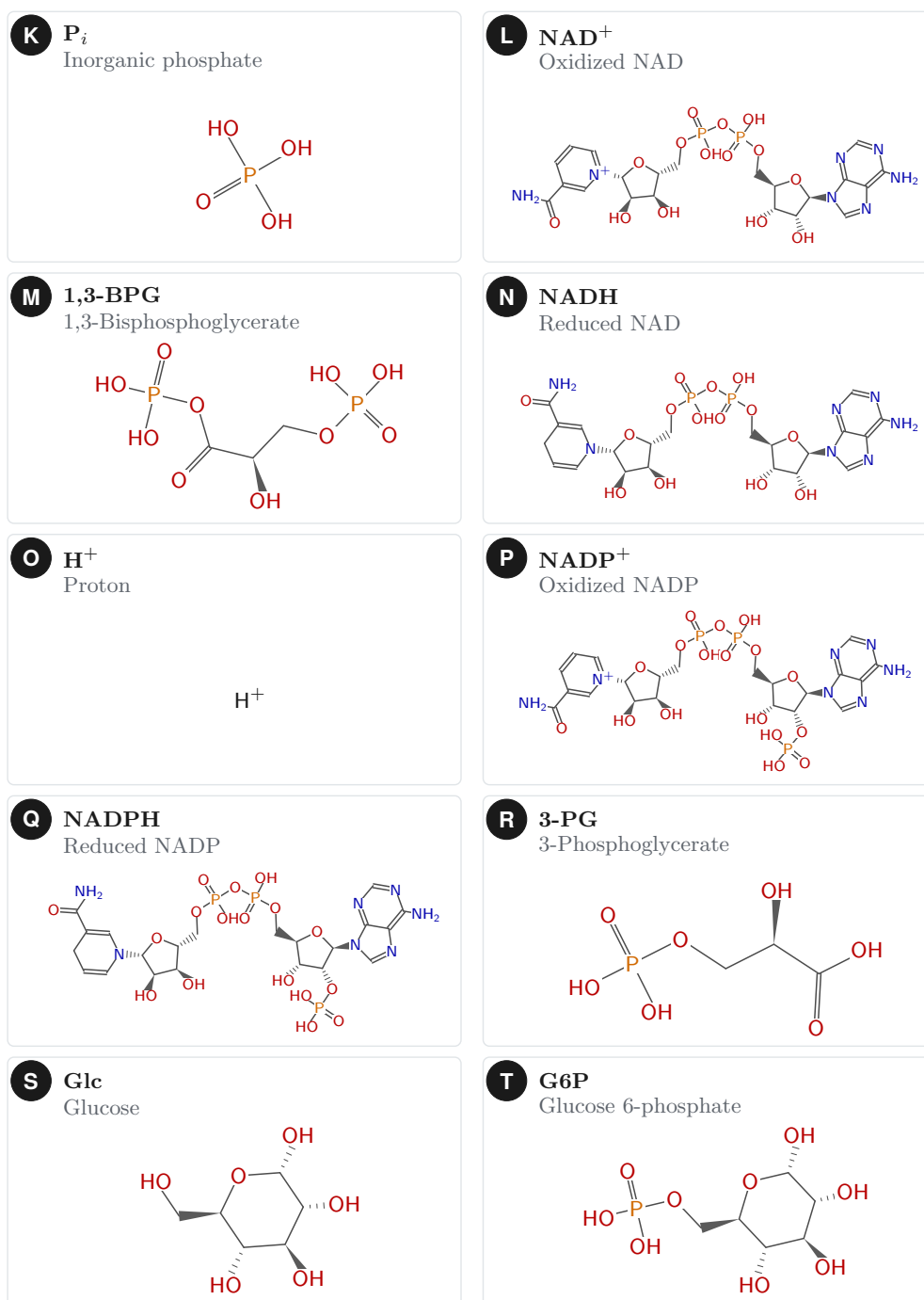

  \ContinuedFloat
  \centering
  \resizebox{\linewidth}{!}{\mzerospeciespagetwo}
  \caption[]{Molecular structures of the 25 species in KEGG glycolysis
    module M00001 (continued). Panels (K)--(T) are shown on this page.}
\end{figure}

\clearpage

\begin{figure}[p]
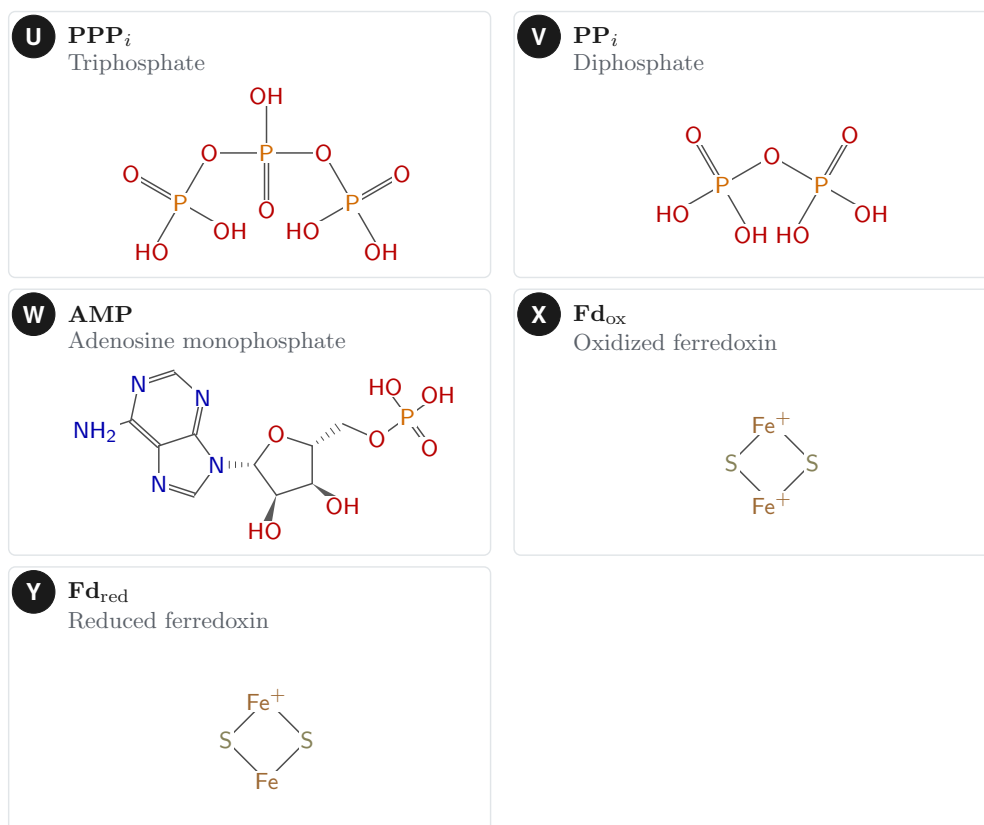

  \ContinuedFloat
  \centering
  \resizebox{\linewidth}{!}{\mzerospeciespagethree}
  \caption[]{Molecular structures of the 25 species in KEGG glycolysis
    module M00001 (continued). Panels (U)--(Y) complete
    Figure~\ref{fig:si-m00001-molecules}.}
\end{figure}
\endgroup

\clearpage

\begin{landscape}
\subsection*{S3: Structural Characteristics of the Benchmark Networks}

\begin{table}[!ht]
  \centering
  \caption{Structural summary of the benchmark networks. In
  NOG-X\(x\)-R\(r\), \(x\) and \(r\) denote the numbers of abstract entities
  and reactions. The quantities \(\bar d_R\) and \(d_R^{\max}\) are the mean
  and maximum reaction degrees. The columns \(P_{1\text{-}1}\),
  \(P_{1\text{-}m}\), and \(P_{m\text{-}m}\) give the percentages of
  one-to-one, one-to-many, and many-to-many reaction patterns; Sym.\ counts
  nontrivial symmetry orbits, and Iso.\ and Olig.\ density give the
  percentages of entity pairs linked by the corresponding obligatory
  relations. NOG values are means \(\pm\) sample standard deviations over
  ten seeds. All networks are connected and use every entity.}
  \label{tab:network-information}
  \setlength{\tabcolsep}{4.2pt}
  \renewcommand{\arraystretch}{1.08}
  \newcommand{\mstd}[2]{#1\,{\(\pm\)}\,#2}

  \scriptsize
  \begin{tabular}{lrrrrrrrrrr}
      \toprule
      &
      \multicolumn{4}{c}{Network information}
      &
      \multicolumn{3}{c}{Reaction patterns}
      &
      \multicolumn{3}{c}{Oblig.\ relations}
      \\
      \cmidrule(lr){2-5}
      \cmidrule(lr){6-8}
      \cmidrule(lr){9-11}

      Network
      & \( |X| \)
      & \( |R| \)
      & \( \bar d_R \)
      & \( d_R^{\max} \)
      & \( P_{1\text{-}1} \) (\%)
      & \( P_{1\text{-}m} \) (\%)
      & \( P_{m\text{-}m} \) (\%)
      & Sym.
      & Iso.\ density (\%)
      & Olig.\ density (\%)
      \\
      \midrule

      NOG-X42-R42
      & 42 & 42
      & 3.00 & \mstd{12.20}{1.87}
      & 33.3
      & 33.3
      & 33.3
      & 2
      & \mstd{4.27}{1.06}
      & \mstd{12.09}{11.61}
      \\

      NOG-X42-R66
      & 42 & 66
      & 4.71 & \mstd{15.20}{1.62}
      & 33.3
      & 33.3
      & 33.3
      & 2
      & \mstd{5.16}{0.60}
      & \mstd{18.35}{7.30}
      \\

      NOG-X63-R63
      & 63 & 63
      & 3.00 & \mstd{13.40}{1.65}
      & 33.3
      & 33.3
      & 33.3
      & 2
      & \mstd{3.35}{1.11}
      & \mstd{10.60}{10.47}
      \\

      NOG-X63-R99
      & 63 & 99
      & 4.71 & \mstd{17.60}{1.58}
      & 33.3
      & 33.3
      & 33.3
      & 2
      & \mstd{4.02}{0.57}
      & \mstd{17.01}{6.18}
      \\

      NOG-X84-R84
      & 84 & 84
      & 3.00 & \mstd{14.50}{1.18}
      & 33.3
      & 33.3
      & 33.3
      & 2
      & \mstd{2.92}{0.97}
      & \mstd{9.89}{9.39}
      \\

      NOG-X84-R132
      & 84 & 132
      & 4.71 & \mstd{19.50}{1.27}
      & 33.3
      & 33.3
      & 33.3
      & 2
      & \mstd{3.55}{0.39}
      & \mstd{15.71}{4.57}
      \\

      NOG-X105-R105
      & 105 & 105
      & 3.00 & \mstd{16.30}{1.70}
      & 33.3
      & 33.3
      & 33.3
      & 2
      & \mstd{2.84}{1.00}
      & \mstd{9.28}{8.78}
      \\

      NOG-X105-R165
      & 105 & 165
      & 4.71 & \mstd{21.70}{1.25}
      & 33.3
      & 33.3
      & 33.3
      & 2
      & \mstd{3.58}{0.31}
      & \mstd{14.58}{4.65}
      \\

      KEGG M00001
      & 25 & 15
      & 2.32 & 6
      & 20.0
      & 13.3
      & 66.7
      & 2
      & 1.00
      & 0.67
      \\

    \bottomrule
  \end{tabular}
\end{table}
\end{landscape}

\clearpage
\providecommand{\cmark}{\checkmark}

\subsection*{S4: NOG Ablation Results}

\begingroup
\small
\setlength{\tabcolsep}{1.6pt}
\setlength{\LTcapwidth}{\linewidth}
\begin{longtable}{llcccccrrrr}

\caption{Ablation results for the original NOG benchmark networks. All
configurations enforce atom conservation and reaction-mechanism validation.
M denotes mass-interval pruning, I/O obligatory isomer and oligomer pruning,
C compositional partial-order pruning, F frontier scheduling, and S symmetry
breaking. Running times over the available seeds 0--9 are reported as mean
\(\pm\) sample standard deviation and median [IQR]; Nodes is the mean number
of explored search nodes, and Sol.\ is the number of solutions. B2 results
for NOG-X84-R84 and NOG-X105-R105 are not reported.
Nine of ten X84 runs and seven of ten X105 runs completed within the
1800-second timeout. Bold values are the lowest mean or median within each
labeled regime and network.}
\label{tab:ablation}
\\
\toprule
\multirow{2}{*}{\shortstack{Network\\regime}}
& \multirow{2}{*}{Config.}
& \multicolumn{3}{c}{Pruning}
& \multicolumn{2}{c}{Search}
& \multicolumn{4}{c}{Results}
\\
\cmidrule(lr){3-5}
\cmidrule(lr){6-7}
\cmidrule(lr){8-11}

&
&
M
& I/O
& C
& F
& S
& Mean \(\pm\) SD (s)
& Median [IQR] (s)
& Nodes
& Sol.
\\
\midrule
\endfirsthead

\multicolumn{11}{c}{
  \tablename\ \thetable{} continued from the previous page
}
\\

\toprule
\multirow{2}{*}{\shortstack{Network\\regime}}
& \multirow{2}{*}{Config.}
& \multicolumn{3}{c}{Pruning}
& \multicolumn{2}{c}{Search}
& \multicolumn{4}{c}{Results}
\\
\cmidrule(lr){3-5}
\cmidrule(lr){6-7}
\cmidrule(lr){8-11}

&
&
M
& I/O
& C
& F
& S
& Mean \(\pm\) SD (s)
& Median [IQR] (s)
& Nodes
& Sol.
\\
\midrule
\endhead

\midrule
\multicolumn{11}{r}{Continued on the next page}
\\
\endfoot

\bottomrule
\endlastfoot

\multirow{7}{*}{\shortstack{NOG\\-X42\\-R42\\[-1pt]{\scriptsize Pruning only}}}
& B1
& \(\cmark\) &  &  &  &
& \(1.59 \pm 0.39\)
& \(1.46\,[0.39]\)
& 115.2
& 4 \\*
& B2
&  & \(\cmark\) &  &  &
& \(65.14 \pm 157.08\)
& \(4.34\,[37.27]\)
& 14384.6
& 4 \\*
& B3
&  &  & \(\cmark\) &  &
& \(1.63 \pm 0.41\)
& \(1.49\,[0.52]\)
& 130.0
& 4 \\*
& B4
& \(\cmark\) & \(\cmark\) &  &  &
& \(\mathbf{1.45} \pm 0.43\)
& \(\mathbf{1.28}\,[0.42]\)
& 104.7
& 4 \\*
& B5
& \(\cmark\) &  & \(\cmark\) &  &
& \(1.82 \pm 0.52\)
& \(1.61\,[0.47]\)
& 129.8
& 4 \\*
& B6
&  & \(\cmark\) & \(\cmark\) &  &
& \(1.49 \pm 0.40\)
& \(1.34\,[0.50]\)
& 113.2
& 4 \\*
& B7
& \(\cmark\) & \(\cmark\) & \(\cmark\) &  &
& \(1.79 \pm 0.60\)
& \(1.47\,[0.52]\)
& 113.2
& 4 \\*
\cmidrule(lr){1-11}
\multirow{5}{*}{\shortstack{NOG\\-X42\\-R42\\[-1pt]{\scriptsize + Frontier}}}
& B8
& \(\cmark\) &  &  & \(\cmark\) &
& \(1.62 \pm 0.32\)
& \(1.53\,[0.38]\)
& 115.2
& 4 \\*
& B9
&  &  & \(\cmark\) & \(\cmark\) &
& \(1.39 \pm 0.36\)
& \(1.27\,[0.47]\)
& 130.0
& 4 \\*
& B10
& \(\cmark\) & \(\cmark\) &  & \(\cmark\) &
& \(1.51 \pm 0.36\)
& \(1.40\,[0.37]\)
& 104.7
& 4 \\*
& B11
&  & \(\cmark\) & \(\cmark\) & \(\cmark\) &
& \(\mathbf{1.32} \pm 0.32\)
& \(\mathbf{1.23}\,[0.36]\)
& 113.2
& 4 \\*
& B12
& \(\cmark\) & \(\cmark\) & \(\cmark\) & \(\cmark\) &
& \(1.38 \pm 0.41\)
& \(1.24\,[0.44]\)
& 113.2
& 4 \\*
\cmidrule(lr){1-11}
\multirow{4}{*}{\shortstack{NOG\\-X42\\-R42\\[-1pt]{\scriptsize + Symmetry}}}
& B13
&  &  & \(\cmark\) & \(\cmark\) & \(\cmark\)
& \(1.31 \pm 0.29\)
& \(1.26\,[0.30]\)
& 112.2
& 1 \\*
& B14
& \(\cmark\) & \(\cmark\) &  & \(\cmark\) & \(\cmark\)
& \(1.48 \pm 0.39\)
& \(1.36\,[0.48]\)
& 90.6
& 1 \\*
& B15
&  & \(\cmark\) & \(\cmark\) & \(\cmark\) & \(\cmark\)
& \(\mathbf{1.26} \pm 0.34\)
& \(\mathbf{1.15}\,[0.32]\)
& 99.1
& 1 \\*
& B16
& \(\cmark\) & \(\cmark\) & \(\cmark\) & \(\cmark\) & \(\cmark\)
& \(1.34 \pm 0.43\)
& \(1.22\,[0.39]\)
& 99.1
& 1 \\

\midrule

\multirow{7}{*}{\shortstack{NOG\\-X63\\-R63\\[-1pt]{\scriptsize Pruning only}}}
& B1
& \(\cmark\) &  &  &  &
& \(6.34 \pm 4.12\)
& \(5.27\,[1.60]\)
& 580.1
& 4 \\*
& B2
&  & \(\cmark\) &  &  &
& \(90.70 \pm 116.98\)
& \(16.80\,[174.49]\)
& 9846.7
& 4 \\*
& B3
&  &  & \(\cmark\) &  &
& \(6.56 \pm 2.49\)
& \(6.55\,[3.44]\)
& 485.9
& 4 \\*
& B4
& \(\cmark\) & \(\cmark\) &  &  &
& \(\mathbf{4.63} \pm 1.45\)
& \(\mathbf{4.53}\,[2.29]\)
& 525.6
& 4 \\*
& B5
& \(\cmark\) &  & \(\cmark\) &  &
& \(7.70 \pm 3.03\)
& \(7.55\,[4.31]\)
& 485.9
& 4 \\*
& B6
&  & \(\cmark\) & \(\cmark\) &  &
& \(5.35 \pm 1.97\)
& \(5.09\,[2.55]\)
& 523.4
& 4 \\*
& B7
& \(\cmark\) & \(\cmark\) & \(\cmark\) &  &
& \(6.73 \pm 2.68\)
& \(6.24\,[3.91]\)
& 523.4
& 4 \\*
\cmidrule(lr){1-11}
\multirow{5}{*}{\shortstack{NOG\\-X63\\-R63\\[-1pt]{\scriptsize + Frontier}}}
& B8
& \(\cmark\) &  &  & \(\cmark\) &
& \(6.12 \pm 3.59\)
& \(5.16\,[1.38]\)
& 580.1
& 4 \\*
& B9
&  &  & \(\cmark\) & \(\cmark\) &
& \(4.64 \pm 1.70\)
& \(4.51\,[2.30]\)
& 485.9
& 4 \\*
& B10
& \(\cmark\) & \(\cmark\) &  & \(\cmark\) &
& \(4.40 \pm 1.21\)
& \(4.36\,[1.33]\)
& 525.6
& 4 \\*
& B11
&  & \(\cmark\) & \(\cmark\) & \(\cmark\) &
& \(\mathbf{3.60} \pm 1.01\)
& \(\mathbf{3.45}\,[1.23]\)
& 523.4
& 4 \\*
& B12
& \(\cmark\) & \(\cmark\) & \(\cmark\) & \(\cmark\) &
& \(4.21 \pm 1.34\)
& \(4.02\,[1.67]\)
& 523.4
& 4 \\*
\cmidrule(lr){1-11}
\multirow{4}{*}{\shortstack{NOG\\-X63\\-R63\\[-1pt]{\scriptsize + Symmetry}}}
& B13
&  &  & \(\cmark\) & \(\cmark\) & \(\cmark\)
& \(4.09 \pm 1.59\)
& \(3.63\,[1.75]\)
& 381.2
& 1 \\*
& B14
& \(\cmark\) & \(\cmark\) &  & \(\cmark\) & \(\cmark\)
& \(3.95 \pm 0.82\)
& \(4.13\,[1.34]\)
& 444.5
& 1 \\*
& B15
&  & \(\cmark\) & \(\cmark\) & \(\cmark\) & \(\cmark\)
& \(\mathbf{3.26} \pm 0.78\)
& \(\mathbf{3.16}\,[1.24]\)
& 442.3
& 1 \\*
& B16
& \(\cmark\) & \(\cmark\) & \(\cmark\) & \(\cmark\) & \(\cmark\)
& \(3.47 \pm 0.86\)
& \(3.40\,[1.54]\)
& 442.3
& 1 \\

\midrule

\multirow{7}{*}{\shortstack{NOG\\-X84\\-R84\\[-1pt]{\scriptsize Pruning only}}}
& B1
& \(\cmark\) &  &  &  &
& \(21.98 \pm 10.17\)
& \(21.43\,[17.17]\)
& 2691.3
& 4 \\*
& B2
&  & \(\cmark\) &  &  &
& \multicolumn{4}{c}{\textit{not reported}} \\*
& B3
&  &  & \(\cmark\) &  &
& \(32.32 \pm 17.73\)
& \(29.50\,[31.05]\)
& 2807.7
& 4 \\*
& B4
& \(\cmark\) & \(\cmark\) &  &  &
& \(\mathbf{20.84} \pm 12.76\)
& \(\mathbf{14.48}\,[22.92]\)
& 2453.8
& 4 \\*
& B5
& \(\cmark\) &  & \(\cmark\) &  &
& \(38.63 \pm 21.38\)
& \(35.77\,[37.70]\)
& 2805.9
& 4 \\*
& B6
&  & \(\cmark\) & \(\cmark\) &  &
& \(25.10 \pm 17.74\)
& \(15.31\,[29.56]\)
& 2191.2
& 4 \\*
& B7
& \(\cmark\) & \(\cmark\) & \(\cmark\) &  &
& \(32.29 \pm 23.38\)
& \(19.15\,[39.26]\)
& 2191.2
& 4 \\*
\cmidrule(lr){1-11}
\multirow{5}{*}{\shortstack{NOG\\-X84\\-R84\\[-1pt]{\scriptsize + Frontier}}}
& B8
& \(\cmark\) &  &  & \(\cmark\) &
& \(20.01 \pm 8.89\)
& \(19.16\,[15.33]\)
& 2691.3
& 4 \\*
& B9
&  &  & \(\cmark\) & \(\cmark\) &
& \(16.73 \pm 7.27\)
& \(18.26\,[12.24]\)
& 2807.7
& 4 \\*
& B10
& \(\cmark\) & \(\cmark\) &  & \(\cmark\) &
& \(18.18 \pm 10.96\)
& \(12.55\,[18.20]\)
& 2453.8
& 4 \\*
& B11
&  & \(\cmark\) & \(\cmark\) & \(\cmark\) &
& \(\mathbf{12.95} \pm 8.12\)
& \(\mathbf{7.90}\,[12.36]\)
& 2191.2
& 4 \\*
& B12
& \(\cmark\) & \(\cmark\) & \(\cmark\) & \(\cmark\) &
& \(15.88 \pm 10.06\)
& \(9.69\,[15.77]\)
& 2191.2
& 4 \\*
\cmidrule(lr){1-11}
\multirow{4}{*}{\shortstack{NOG\\-X84\\-R84\\[-1pt]{\scriptsize + Symmetry}}}
& B13
&  &  & \(\cmark\) & \(\cmark\) & \(\cmark\)
& \(13.55 \pm 4.74\)
& \(15.00\,[7.92]\)
& 1953.5
& 1 \\*
& B14
& \(\cmark\) & \(\cmark\) &  & \(\cmark\) & \(\cmark\)
& \(13.93 \pm 6.31\)
& \(11.53\,[11.61]\)
& 1959.1
& 1 \\*
& B15
&  & \(\cmark\) & \(\cmark\) & \(\cmark\) & \(\cmark\)
& \(\mathbf{10.85} \pm 5.04\)
& \(\mathbf{8.02}\,[8.45]\)
& 1696.5
& 1 \\*
& B16
& \(\cmark\) & \(\cmark\) & \(\cmark\) & \(\cmark\) & \(\cmark\)
& \(12.21 \pm 5.86\)
& \(8.98\,[9.65]\)
& 1696.5
& 1 \\

\midrule

\multirow{7}{*}{\shortstack{NOG\\-X105\\-R105\\[-1pt]{\scriptsize Pruning only}}}
& B1
& \(\cmark\) &  &  &  &
& \(\mathbf{137.46} \pm 183.60\)
& \(72.39\,[91.78]\)
& 13356.2
& 4 \\*
& B2
&  & \(\cmark\) &  &  &
& \multicolumn{4}{c}{\textit{not reported}} \\*
& B3
&  &  & \(\cmark\) &  &
& \(235.24 \pm 321.63\)
& \(142.48\,[173.03]\)
& 13763.0
& 4 \\*
& B4
& \(\cmark\) & \(\cmark\) &  &  &
& \(162.67 \pm 238.06\)
& \(\mathbf{71.74}\,[139.45]\)
& 13468.8
& 4 \\*
& B5
& \(\cmark\) &  & \(\cmark\) &  &
& \(296.66 \pm 409.74\)
& \(180.21\,[221.10]\)
& 13763.0
& 4 \\*
& B6
&  & \(\cmark\) & \(\cmark\) &  &
& \(229.35 \pm 329.19\)
& \(107.28\,[179.44]\)
& 13472.3
& 4 \\*
& B7
& \(\cmark\) & \(\cmark\) & \(\cmark\) &  &
& \(333.61 \pm 488.63\)
& \(149.01\,[270.01]\)
& 13472.3
& 4 \\*
\cmidrule(lr){1-11}
\multirow{5}{*}{\shortstack{NOG\\-X105\\-R105\\[-1pt]{\scriptsize + Frontier}}}
& B8
& \(\cmark\) &  &  & \(\cmark\) &
& \(125.42 \pm 162.59\)
& \(66.34\,[84.25]\)
& 13356.2
& 4 \\*
& B9
&  &  & \(\cmark\) & \(\cmark\) &
& \(\mathbf{88.81} \pm 108.73\)
& \(60.02\,[62.01]\)
& 13763.0
& 4 \\*
& B10
& \(\cmark\) & \(\cmark\) &  & \(\cmark\) &
& \(143.66 \pm 205.66\)
& \(61.77\,[126.77]\)
& 13468.8
& 4 \\*
& B11
&  & \(\cmark\) & \(\cmark\) & \(\cmark\) &
& \(100.74 \pm 140.68\)
& \(\mathbf{43.62}\,[92.27]\)
& 13472.3
& 4 \\*
& B12
& \(\cmark\) & \(\cmark\) & \(\cmark\) & \(\cmark\) &
& \(130.44 \pm 185.09\)
& \(56.42\,[115.70]\)
& 13472.3
& 4 \\*
\cmidrule(lr){1-11}
\multirow{4}{*}{\shortstack{NOG\\-X105\\-R105\\[-1pt]{\scriptsize + Symmetry}}}
& B13
&  &  & \(\cmark\) & \(\cmark\) & \(\cmark\)
& \(\mathbf{47.06} \pm 37.06\)
& \(37.29\,[16.18]\)
& 6904.4
& 1 \\*
& B14
& \(\cmark\) & \(\cmark\) &  & \(\cmark\) & \(\cmark\)
& \(65.96 \pm 73.24\)
& \(37.27\,[53.33]\)
& 6610.2
& 1 \\*
& B15
&  & \(\cmark\) & \(\cmark\) & \(\cmark\) & \(\cmark\)
& \(51.93 \pm 58.19\)
& \(\mathbf{29.14}\,[39.92]\)
& 6613.7
& 1 \\*
& B16
& \(\cmark\) & \(\cmark\) & \(\cmark\) & \(\cmark\) & \(\cmark\)
& \(60.45 \pm 67.11\)
& \(37.20\,[41.51]\)
& 6613.7
& 1 \\

\end{longtable}

\endgroup

\clearpage

\subsection*{S5: Reaction-Enriched NOG Results}

\begin{table}[!ht]
  \centering
  \small
  \setlength{\tabcolsep}{2.8pt}
  \renewcommand{\arraystretch}{1.08}
  \caption{Running time comparison between original and reaction-enriched
    NOG networks with identical molecule sets. Check marks denote enabled
    strategies using the notation of the ablation study. Values summarize
    ten seeds and are reported as mean \(\pm\) sample standard deviation
    and median [IQR]; Nodes is the mean number of explored search nodes.}
  \label{tab:reaction-extended-results}

  \begin{tabular}{rrlcccccrrr}
    \toprule
    \multicolumn{2}{c}{Network}
    & \multirow{2}{*}{Config.}
    & \multicolumn{3}{c}{Pruning}
    & \multicolumn{2}{c}{Search}
    & \multicolumn{3}{c}{Results}
    \\
    \cmidrule(lr){1-2}
    \cmidrule(lr){4-6}
    \cmidrule(lr){7-8}
    \cmidrule(lr){9-11}

    \( |X| \)
    & \( |R| \)
    &
    & M
    & I/O
    & C
    & F
    & S
    & Mean \(\pm\) SD (s)
    & Median [IQR] (s)
    & Nodes
    \\
    \midrule

    \multirow{6}{*}{42}
    & \multirow{3}{*}{42}
    & B13 &  &  & \(\cmark\) & \(\cmark\) & \(\cmark\)
    & \(1.31 \pm 0.29\) & \(1.26\,[0.30]\) & 112.2 \\
    & & B15 &  & \(\cmark\) & \(\cmark\) & \(\cmark\) & \(\cmark\)
    & \(1.26 \pm 0.34\) & \(1.15\,[0.32]\) & 99.1 \\
    & & B16 & \(\cmark\) & \(\cmark\) & \(\cmark\) & \(\cmark\) & \(\cmark\)
    & \(1.34 \pm 0.43\) & \(1.22\,[0.39]\) & 99.1 \\
    \cmidrule(lr){2-11}
    & \multirow{3}{*}{66}
    & B13 &  &  & \(\cmark\) & \(\cmark\) & \(\cmark\)
    & \(1.36 \pm 0.26\) & \(1.32\,[0.24]\) & 66.7 \\
    & & B15 &  & \(\cmark\) & \(\cmark\) & \(\cmark\) & \(\cmark\)
    & \(1.38 \pm 0.28\) & \(1.32\,[0.29]\) & 66.3 \\
    & & B16 & \(\cmark\) & \(\cmark\) & \(\cmark\) & \(\cmark\) & \(\cmark\)
    & \(1.39 \pm 0.26\) & \(1.34\,[0.15]\) & 66.3 \\
    \midrule

    \multirow{6}{*}{63}
    & \multirow{3}{*}{63}
    & B13 &  &  & \(\cmark\) & \(\cmark\) & \(\cmark\)
    & \(4.09 \pm 1.59\) & \(3.63\,[1.75]\) & 381.2 \\
    & & B15 &  & \(\cmark\) & \(\cmark\) & \(\cmark\) & \(\cmark\)
    & \(3.26 \pm 0.78\) & \(3.16\,[1.24]\) & 442.3 \\
    & & B16 & \(\cmark\) & \(\cmark\) & \(\cmark\) & \(\cmark\) & \(\cmark\)
    & \(3.47 \pm 0.86\) & \(3.40\,[1.54]\) & 442.3 \\
    \cmidrule(lr){2-11}
    & \multirow{3}{*}{99}
    & B13 &  &  & \(\cmark\) & \(\cmark\) & \(\cmark\)
    & \(3.52 \pm 0.74\) & \(3.45\,[1.25]\) & 230.3 \\
    & & B15 &  & \(\cmark\) & \(\cmark\) & \(\cmark\) & \(\cmark\)
    & \(3.55 \pm 0.76\) & \(3.43\,[1.37]\) & 248.8 \\
    & & B16 & \(\cmark\) & \(\cmark\) & \(\cmark\) & \(\cmark\) & \(\cmark\)
    & \(3.73 \pm 0.90\) & \(3.57\,[1.60]\) & 248.8 \\
    \midrule

    \multirow{6}{*}{84}
    & \multirow{3}{*}{84}
    & B13 &  &  & \(\cmark\) & \(\cmark\) & \(\cmark\)
    & \(13.55 \pm 4.74\) & \(15.00\,[7.92]\) & 1953.5 \\
    & & B15 &  & \(\cmark\) & \(\cmark\) & \(\cmark\) & \(\cmark\)
    & \(10.85 \pm 5.04\) & \(8.02\,[8.45]\) & 1696.5 \\
    & & B16 & \(\cmark\) & \(\cmark\) & \(\cmark\) & \(\cmark\) & \(\cmark\)
    & \(12.21 \pm 5.86\) & \(8.98\,[9.65]\) & 1696.5 \\
    \cmidrule(lr){2-11}
    & \multirow{3}{*}{132}
    & B13 &  &  & \(\cmark\) & \(\cmark\) & \(\cmark\)
    & \(7.52 \pm 2.06\) & \(7.05\,[2.02]\) & 711.7 \\
    & & B15 &  & \(\cmark\) & \(\cmark\) & \(\cmark\) & \(\cmark\)
    & \(8.29 \pm 2.56\) & \(7.60\,[2.93]\) & 711.7 \\
    & & B16 & \(\cmark\) & \(\cmark\) & \(\cmark\) & \(\cmark\) & \(\cmark\)
    & \(9.06 \pm 3.15\) & \(8.24\,[3.61]\) & 711.7 \\
    \midrule

    \multirow{6}{*}{105}
    & \multirow{3}{*}{105}
    & B13 &  &  & \(\cmark\) & \(\cmark\) & \(\cmark\)
    & \(47.06 \pm 37.06\) & \(37.29\,[16.18]\) & 6904.4 \\
    & & B15 &  & \(\cmark\) & \(\cmark\) & \(\cmark\) & \(\cmark\)
    & \(51.93 \pm 58.19\) & \(29.14\,[39.92]\) & 6613.7 \\
    & & B16 & \(\cmark\) & \(\cmark\) & \(\cmark\) & \(\cmark\) & \(\cmark\)
    & \(60.45 \pm 67.11\) & \(37.20\,[41.51]\) & 6613.7 \\
    \cmidrule(lr){2-11}
    & \multirow{3}{*}{165}
    & B13 &  &  & \(\cmark\) & \(\cmark\) & \(\cmark\)
    & \(33.90 \pm 38.63\) & \(23.18\,[17.67]\) & 4448.6 \\
    & & B15 &  & \(\cmark\) & \(\cmark\) & \(\cmark\) & \(\cmark\)
    & \(41.35 \pm 50.63\) & \(26.08\,[23.49]\) & 4448.6 \\
    & & B16 & \(\cmark\) & \(\cmark\) & \(\cmark\) & \(\cmark\) & \(\cmark\)
    & \(46.60 \pm 56.84\) & \(29.72\,[28.10]\) & 4448.6 \\

    \bottomrule
  \end{tabular}
\end{table}

\clearpage

\subsection*{S6: KEGG M00001 Backtracking Results}

\begin{table}[!ht]
  \centering
  \small
  \setlength{\tabcolsep}{2.8pt}
  \renewcommand{\arraystretch}{1.08}
  \caption{Backtracking results for KEGG M00001. Check marks denote enabled
    strategies using the notation of the NOG ablation study. Time is the
    running time of the individual run, Nodes is the number of explored
    search nodes, and Sol.\ is the number of solutions.}
  \label{tab:m00001-backtracking}

  \begin{tabular}{lcccccrrr}
    \toprule
    \multirow{2}{*}{Config.}
    & \multicolumn{3}{c}{Pruning}
    & \multicolumn{2}{c}{Search}
    & \multicolumn{3}{c}{Results}
    \\
    \cmidrule(lr){2-4}
    \cmidrule(lr){5-6}
    \cmidrule(lr){7-9}
    & M
    & I/O
    & C
    & F
    & S
    & Time (s)
    & Nodes
    & Sol.
    \\
    \midrule
    B13
    &  &  & \(\cmark\) & \(\cmark\) & \(\cmark\)
    & 489.11 & 145 & 2
    \\
    B15
    &  & \(\cmark\) & \(\cmark\) & \(\cmark\) & \(\cmark\)
    & 494.90 & 152 & 2
    \\
    B16
    & \(\cmark\) & \(\cmark\) & \(\cmark\) & \(\cmark\) & \(\cmark\)
    & 439.35 & 144 & 2
    \\
    \bottomrule
  \end{tabular}
\end{table}

\end{document}